 \documentclass[final,3p,times]{elsarticle}
\usepackage{ulem}
\usepackage{cancel}
\usepackage{graphicx}
\usepackage{exscale}
\usepackage{times}
\usepackage{color}
\usepackage{amssymb}
\usepackage{setspace}
\usepackage{natbib}
\usepackage{amsmath}
\usepackage{subfig}
\usepackage{algorithmic}
\usepackage{algorithm}

\font\ppppppcarac=ptmr8y at 4pt
\font\pppppcarac=ptmr8y at 5pt
\font\ppppcarac=ptmr8y at 6pt
\font\pppcarac=ptmr8y at 7pt
\font\ppcarac=ptmr8y at 8pt
\font\pcarac=ptmr8y at 9pt

\font\Ppcarac=ptmr8y at 10pt

\font\bf=ptmb8y at 10pt

\newcommand{\bfE}{{\textbf{E}}}
\newcommand{\bfF}{{\textbf{F}}}
\newcommand{\bfG}{{\textbf{G}}}
\newcommand{\bfH}{{\textbf{H}}}

\newcommand{\bfK}{{\textbf{K}}}
\newcommand{\bfL}{{\textbf{L}}}

\newcommand{\bfQ}{{\textbf{Q}}}

\newcommand{\bfS}{{\textbf{S}}}

\newcommand{\bfU}{{\textbf{U}}}
\newcommand{\bfV}{{\textbf{V}}}
\newcommand{\bfW}{{\textbf{W}}}
\newcommand{\bfX}{{\textbf{X}}}
\newcommand{\bfY}{{\textbf{Y}}}
\newcommand{\bfZ}{{\textbf{Z}}}

\newcommand{\bfzero}{{ \hbox{\bf 0} }}

\newcommand{\bfb}{{\textbf{b}}}

\newcommand{\bff}{{\textbf{f}}}
\newcommand{\bfg}{{\textbf{g}}}
\newcommand{\bfh}{{\textbf{h}}}

\newcommand{\bfm}{{\textbf{m}}}
\newcommand{\bfn}{{\textbf{n}}}

\newcommand{\bfq}{{\textbf{q}}}

\newcommand{\bfs}{{\textbf{s}}}

\newcommand{\bfu}{{\textbf{u}}}

\newcommand{\bfw}{{\textbf{w}}}
\newcommand{\bfx}{{\textbf{x}}}

\newcommand{\bfz}{{\textbf{z}}}

\newcommand{\bfEta}{{\bfH}}

\newcommand{\bfSigma}{{\boldsymbol{\Sigma}}}

\newcommand{\bfeta}{{\boldsymbol{\eta}}}

\newcommand{\bflambda}{{\boldsymbol{\lambda}}}

\newcommand{\bfvarphi}{{\boldsymbol{\varphi}}}
\newcommand{\bfpsi}{{\boldsymbol{\psi}}}
\newcommand{\bfxi}{{\boldsymbol{\xi}}}
\newcommand{\bfzeta}{{\boldsymbol{\zeta}}}

\newcommand{\Eta}{{\rm H}}

\newcommand{\KK}{{\mathbb{K}}}

\newcommand{\MM}{{\mathbb{M}}}

\newcommand{\PP}{{\mathbb{P}}}
\newcommand{\RR}{{\mathbb{R}}}

\DeclareMathAlphabet{\mathonebb}{U}{bbold}{m}{n}
\def\11{{\ensuremath{\mathonebb{1}}}}

\newcommand{\curG}{{\mathcal{G}}}
\newcommand{\curM}{{\mathcal{M}}}
\newcommand{\curP}{{\mathcal{P}}}
\newcommand{\curT}{{\mathcal{T}}}
\newcommand{\curU}{{\mathcal{U}}}
\newcommand{\bfcurG}{ {\boldsymbol{{\mathcal{G}}}} }
\newcommand{\bfcurM}{ {\boldsymbol{{\mathcal{M}}}} }
\newcommand{\bfcurU}{ {\boldsymbol{{\mathcal{U}}}} }

\newcommand{\divergence}{\hbox{{\Ppcarac div}}\,}

\newcommand{\Jump}{{\hbox{{\ppcarac Jump}}}}
\newcommand{\perr}{{\hbox{{\pcarac err}}}}
\newcommand{\pperr}{{\hbox{{\ppcarac err}}}}
\newcommand{\pref}{{\hbox{{\pppcarac ref}}}}

\newcommand{\DM}{{\hbox{{\pppppcarac DM}}}}
\newcommand{\opt}{{\hbox{{\ppcarac o}}}}
\newcommand{\popt}{{\hbox{{\ppppcarac o}}}}
\newcommand{\ppopt}{{\hbox{{\ppppcarac o}}}}
\newcommand{\ar}{{\hbox{{\pppcarac ar}}}}

\newcommand{\MC}{{\hbox{{\ppppcarac MC}}}}

\newcommand{\ppMC}{{\hbox{{\ppppppcarac MC}}}}
\newcommand{\PCA}{{\hbox{{\ppppcarac PCA}}}}
\newcommand{\pPCA}{{\hbox{{\pppppcarac PCA}}}}

\newcommand{\tr}{{\hbox{{\textrm tr}}}}

\newcommand{\pcorr}{{\hbox{{\ppppcarac corr}}}}

\newtheorem{proposition}{Proposition}

\newproof{proof}{Proof}
\newdefinition{remark}{Remark}
\newdefinition{hypothesis}{Hypothesis}
\newdefinition{notation}{Notation}

\journal{arXiv}

\begin{document}

\begin{frontmatter}

\title{Probabilistic Learning on Manifolds (PLoM) with Partition}

\author[1]{C. Soize \corref{cor1}}
\ead{christian.soize@univ-eiffel.fr}
\author[2]{R. Ghanem}
\ead{ghanem@usc.edu}

\cortext[cor1]{Corresponding author: C. Soize, christian.soize@univ-eiffel.fr}
\address[1]{Universit\'e Gustave Eiffel, MSME UMR 8208 CNRS, 5 bd Descartes, 77454 Marne-la-Vall\'ee, France}
\address[2]{University of Southern California, 210 KAP Hall, Los Angeles, CA 90089, United States}

\begin{abstract}
The probabilistic learning on manifolds (PLoM) introduced in 2016 has solved difficult supervised problems for the ``small data'' limit where the number $N$ of points in the training set is small. Many extensions have since been proposed, making it possible to deal with increasingly complex cases. However, the performance limit has been observed and explained for applications for which $N$ is very small ($50$ for example) and for which the dimension of the diffusion-map basis is close to $N$. For these cases, we propose a novel extension based on the introduction of a partition in independent random vectors. We take advantage of this novel development to present improvements of the PLoM such as a simplified algorithm for constructing the diffusion-map basis and a new mathematical result for quantifying the concentration of the probability measure in terms of a probability upper bound. The analysis of the  efficiency of this novel extension is presented through two applications.
\end{abstract}



\end{frontmatter}

\section*{Notations}
The following notations are used:\\
\noindent A lower case letter such as $x$, $\eta$, or $u$, is a  real deterministic variable.\\
A boldface lower case letter such as $\bfx$, $\bfeta$, or $\bfu$ is a real deterministic vector.\\
An upper case letter such as $X$, $\Eta$, or $U$, is a real random variable (except for $E$).\\
A boldface upper case letter, $\bfX$, $\bfEta$, or $\bfU$, is a real random vector.\\
A letter between brackets such as  $[x]$, $[\eta]$, $[u]$ or $[C]$, is a real deterministic matrix.\\
A boldface upper case letter between brackets such as $[\bfX]$, $[\bfEta]$, or $[\bfU]$, is a real random matrix.\\
$n,n_q,n_u,n_w\,$: dimensions of random vectors $\bfX,\bfQ,\bfU,\bfW$.\\
$n_\MC\,$: number of additional realizations for random matrix $[\bfH^N_{m_\ppopt}]$.\\
$n_p\,$: number of groups in the partition of $\bfH$.\\
$m, m_i\,$: dimension of the reduced-order diffusion-map bases $[g_m], [g^i_m]$.\\
$m_\opt, m_{i,\opt}\,$: optimal value of $m, m_i$.\\
$E\,$: mathematical expectation.\\
$N,N_\ar\,$: number of points in the training, learned sets.\\
$\nu,\nu_i\,$: dimensions of random vectors $\bfH,\bfY^i$.\\
$\RR,\RR^n\,$: real line, Euclidean vector space of dimension $n$.\\
$\MM_n,\MM_{n,N}\,$: sets of all the $(n\times n),(n\times N)$ real matrices.\\
$\MM_n^+\,$: set of all the positive-definite symmetric $(n\times n)$ real matrices.\\
$\Vert\bfx\Vert\,$: Euclidean norm when $\bfx$ is the vector or Frobenius norm when $[x]$ is the matrix.\\

%
\section{Introduction}
\label{sec:Intro}
%
\paragraph{(i) About the PLoM}
The PLoM (probabilistic learning on manifolds) method was proposed in 2016 \cite{Soize2016} as a complementary approach to existing methods
in machine learning. It allows for solving unsupervised and supervised problems under uncertainty for which the training sets are small. This situation is
encountered in many problems of physics and engineering sciences with
expensive function evaluations. The exploration of the admissible
solution space in these situations is thus hampered by available
computational resources.  The PLoM was successfully adapted to tackle
these challenges for several related problems including nonconvex optimization
under uncertainty \cite{Ghanem2020,Ghanem2019} and the calculation of Sobol's indices  \cite{Arnst2020}.

\paragraph{(ii) Brief discussion on the hypotheses on which the PLoM method has been built}
The PLoM approach starts from a training set made up of a relatively small number $N$ of points (initial realizations).
For the supervised case, it is assumed that the training set is related to an underlying stochastic manifold related to a $\RR^{n}$-valued random variable $\bfX = (\bfQ,\bfW)$ with $\bfQ = \bff(\bfU,\bfW) = \bfF(\bfW)$. The measurable mapping $\bff$ is unknown and $\bfF$ is also an unknown stochastic mapping. The probability distributions of the vector-valued random variables $\bfW$ (control parameter) and $\bfU$ (non-controlled parameter) are given. The stochastic manifold is defined by the unknown graph $\{\bfw, \bfF(\bfw)\}$ for $\bfw$ belonging to the support of the probability distribution of $\bfW$. In the PLoM construction, its is not assumed that this stochastic manifold can directly be described; for instance, it is not assumed that there exist properties of local differentiability (moreover, the manifold is stochastic). Under these conditions, the probability measure of $\bfX$ is concentrated in a region of $\RR^{n}$ for which the only available information is the cloud of  the points of the training set. The PLoM method makes it possible to generate the learned set whose $n_\ar \gg N$ points (additional realizations) are generated by the probability measure that is estimated from the training set. The concentration of the probability measure is preserved thanks to the use of the diffusion-maps basis that allows to enrich the available information from the training set. It should also be noted that the estimate of this unknown probability measure cannot be performed from the training set by using an arbitrarily estimator. It must be parameterized in a manner that permits convergence to any probability measure as its number of points goes towards infinity. The PLoM method therefore does not only consist in generating points that belong to the region in which the measure is concentrated, but also allows these additional points to be realizations of the estimate probability measure with the convergence properties evoked above. The choice of the kernel estimation method for estimating the probability measure from the training set guarantees that this required fundamental property is satisfied (see \cite{Soize2016} and in particular, Section~5.3 of \cite{Soize2020c}).

\paragraph{(iii) A difficulty of the PLoM method for certain applications}
Since its introduction in 2016, extensions of the original method
\cite{Soize2016} have been developed in order to address increasingly
complex problems for the case of small data: sampling of Bayesian
posteriors with a non-Gaussian probabilistic learning on manifolds in
very high dimension \cite{Soize2020b}, physics-constrained
non-Gaussian probabilistic learning on manifolds, for instance, to
integrate information coming from experimental measurements during the
learning \cite{Soize2020a}, probabilistic learning on manifolds
constrained by nonlinear partial differential equations for small
data \cite{Soize2020e}. During this period, a number of
applications were addressed making it possible to refine the method,
to validate it, and to better assess and relax its limitations.
However, some challenges have remained. These are cases where the number of points (realizations) in the training set is very small (for example
$50$) and for which the dimension of the subspace generated by the diffusion-map basis is very close to this number.
In this case, the PLoM may not be more efficient than a standard MCMC
algorithm that is agnostic to any concentration of the probability
measure. One possible way to improve the PLoM for these very
challenging cases is to partition the random vector of which the
training set is a realization, into statistically independent
groups in a non-Gaussian framework. In this manner, statistical
knowledge about the data set, beyond its localization to a
manifold, is relied upon to enhance information extraction and
representation. One difficulty with standard grouping into independent
components is that they typically divvy-up available samples into
seperate groups, with each group consisting of a smaller number of
samples. These approach would not be suitable  to the present setting
given the already small sample size or the training set.  A more
useful approach, which is adopted in this paper, results in groups
that are each equal in size to the training set, but for which the
dimension of the diffusion-map basis is significantly reduced.
This approach corresponds to the novel extension of the PLoM method that we present in this paper.

\paragraph{(iv) A novel extension of the PLOM method to get around the difficulty}
One of the ingredients for this novel extension of the PLoM method is
the construction of a partition in non-Gaussian independent random
vectors, assuming that it exists. Indeed, there may very well be
applications for which the partition yields a single group, identical
to the initial random vector.  For such cases the present approach of
PLoM with partitions affords no further reduction.
Concerning the construction of a partition of independent random vectors, a popular method for testing the statistical independence of the $\nu$ components of a random vector from a given set of $N$ realizations is the use of the frequency distribution \cite{Fleiss2013} coupled with the use of the Pearson chi-squared ($\chi^2$) test \cite{Greenwood1996,Pearson1900}.
For the high dimensions ($\nu$ big) and a relatively small value of
$N$, such an approach does not give sufficiently accurate results. In
addition, even when this type of methods permits testing for
independence, the need remains for a fast algorithm for constructing
the optimal partition .
The independent component analysis (ICA) \cite{Boscolo2004,Comon1994,Comon1991,Herault1986,Hyvarinen1999,Hyvarinen2000,Jutten1991,Lee2000} is a method that consists of extracting independent source signals as a linear mixture of mutually statistically dependent signals, and is often used for source-separation problems.
The fundamental hypothesis that introduced in the ICA methodology is
that the observed vector-valued signal is a linear transformation of
statistically independent real-valued signals (that is to say, is a
linear transformation of a vector-valued signal whose components are
mutually independent) and the objective of the ICA algorithms is to
identify the best linear operator.
In this paper, for the PLoM with partition, we
use the procedure proposed in \cite{Soize2017f}, which is an ICA by mutual information and which does not use the construction of a linear transformation. This direct algorithm permits the identification of an optimal partition in terms of independent random vectors for any non-Gaussian vector in high dimension, which is defined by a relatively small number $N$ of realizations, and which is based on Information Theory.

\paragraph{(v) Organization of the paper}
In Section~\ref{sec:Meth}, we present the methodology of the PLoM method with
partition. In the process, we reintroduce some necessary key notions
of the PLoM method. Sections~\ref{sec:Ap1} and \ref{sec:Ap2} are both devoted
toapplications. Finally, a discussion of the method is presented in the conclusions.

\paragraph{(vi) Novelties presented in the paper}
The main novelty is the development of the PLoM methodology with partition. We also propose a novel algorithm for identifying the optimal values of the hyperparameters related to the construction of the reduced-order diffusion-map basis. For covering the cases for which the normalization introduced by the PLoM is lost with the use of a partition, we introduce constraints by using the Kullback-Leibler minimum cross-entropy principle \cite{Soize2020a}. The quantification of the concentration of the probability measure for the PLoM, which is performed with the distance introduced in \cite{Soize2020c}, is extended for the PLoM with partition, and is completed by a novel result formulated in terms of a probability upper bound of the measure of concentration.

\section{Methodology}
\label{sec:Meth}

\subsection{Supervised problem and training set}
\label{sec:Meth-1}
Let $(\bfw,\bfu)\mapsto \bff(\bfw,\bfu)$ be any measurable mapping on
$\RR^{n_w}\times \RR^{n_u}$ with values in $\RR^{n_q}$ representing a
mathematical/computational model. Let $\bfW$ be the $\RR^{n_w}$-valued
random control parameter and let $\bfU$ be a $\RR^{n_u}$-valued random
non-controlled parameter,  defined on a probability space
$(\Theta,\curT,\curP)$. The random vectors $\bfW$ and $\bfU$ are
assumed to be statistically independent and they are generally non-Gaussian. The probability distributions $P_\bfW(d\bfw)$ $= p_\bfW(\bfw)\, d\bfw$ and $P_\bfU(d\bfu) = p_\bfU(\bfu)\, d\bfu$ are defined by the probability density functions  $p_\bfW$ and $p_\bfU$ with respect to the Lebesgue measures $d\bfw$ and $d\bfu$ on $\RR^{n_w}$ and $\RR^{n_u}$.  Let $\bfQ$ be the $\RR^{n_q}$-valued random variable (QoI) defined on $(\Theta,\curT,\curP)$ such that $\bfQ = \bff(\bfW,\bfU)$. It is assumed that $N \geq 3$ independent realizations $\{\bfq_d^j,j=1,\ldots , N\}$ of $\bfQ$ have been computed  such that $\bfq_d^j = \bff(\bfw_d^j,\bfu_d^j)$ in which $\{\bfw_d^{j} , j=1,\ldots , N\}$ and $\{\bfu_d^{j} , j=1,\ldots , N\}$ are $N$ independent realizations of $\bfW$ and $\bfU$ (subscript $d$ refers to the training set).  We then consider the random variable $\bfX$ with values in $\RR^n$, such that
$\bfX = (\bfQ,\bfW)$ with $n= n_q+n_w$. The training set (initial
data set) related to random vector  $\bfX$ is then made up of the $N$
independent realizations $\{\bfx_d^{j} , j=1,\ldots , N\}$ in which
$\bfx_d^j = (\bfq_d^j,\bfw_d^j) \in \RR^n$ (note that $\bfU$ is not
included in $\bfX$). Since generally the data pertains to
heterogeneous features with potentially wildly distinct supports, it
is assumed that the training set has been suitably scaled for the
purpose of computational statistics. Let us assume that the measurable mapping $\bff$ is such that the conditional probability distribution  $P_{\bfQ\vert\bfW}(d\bfq\vert \bfw)$ given $\bfW = \bfw$ admits a conditional probability density function. It can be deduced (see \cite{Soize2020a}) that the probability distribution $P_\bfX(d\bfx)$ of $\bfX$ admits a density $\bfx\mapsto p_\bfX(\bfx)$ with respect to the Lebesgue measure $d\bfx$ on $\RR^n$. The PLoM \cite{Soize2016,Soize2020c} allows for generating the learned set made up of
$N_\ar \gg N$ realizations $\{\bfx_\ar^{\ell} , \ell=1,\ldots , N_\ar\}$ that allows for deducing
$\{(\bfq_\ar^\ell , \bfw_\ar^\ell) = \bfx_\ar^\ell, \ell=1,\ldots , N_\ar\}$ without using the computational model, but using only the training
set (subscript $ar$ refers to the learned set).
\subsection{Principal component analysis (PCA) of random vector $\bfX$}
\label{sec:Meth-2}
Let $\underline{\bfx}_d \in \RR^n$ and $[C_\bfX] \in \MM_n^+$ be the  mean vector and the covariance matrix of $\bfX$
estimated with the training set.
Let $\mu_1\geq \mu_2\geq \ldots \geq \mu_\nu > 0$ be the $\nu$ largest eigenvalues and let $\bfvarphi^1,\ldots,\bfvarphi^\nu$ be the associated orthonormal eigenvectors of $[C_\bfX]$. The integer $\nu\leq n$ is such that, for a given $\varepsilon_\pPCA > 0$, we have
$\perr_\pPCA(\nu) = 1- {\sum_{\alpha=1}^{\nu} \mu_{\alpha}}/{\tr[C_\bfX]} \leq \varepsilon_\pPCA $.
The PCA of $\bfX$ allows for representing $\bfX$ by $\bfX^\nu$ such that
$\bfX^\nu = \underline{\bfx}_d +[\Phi]\,[\mu]^{1/2}\, \bfH$ such that
$E\{\Vert \bfX-\bfX^\nu\Vert^2 \} \leq \varepsilon_\pPCA \, E\{ \Vert\bfX\Vert^2\}$,
in which $[\Phi] = [\bfvarphi^1 \ldots \bfvarphi^\nu] \in \MM_{n,\nu}$ such that $[\Phi]^T\, [\Phi] = [I_\nu]$ and $[\mu]$ is the diagonal $(\nu\times\nu)$ matrix such that $[\mu]_{\alpha\beta} = \mu_\alpha\delta_{\alpha\beta}$.
From a numerical point of view, if $N < n$, then matrix $[C_\bfX]$ is not estimated and $\nu, [\mu]$, and
$[\Phi]$ are directly computed using a thin SVD \cite{Golub1993} of the matrix whose $N$ columns are $(\bfx_d^j-\underline{\bfx}_d)$ for $j=1,\ldots , N$.
The $\RR^\nu$-valued random variable $\bfH$ is obtained by projection,
$\bfH = [\mu]^{-1/2}\, [\Phi]^T\,(\bfX -\underline{\bfx}_d)$,
and its $N$ independent realizations $\{\bfeta_d^j, j=1,\ldots , N\}$ are such that
$\bfeta_d^j = [\mu]^{-1/2}\, [\Phi]^T\,(\bfx_d^j -\underline{\bfx}_d)\in\RR^\nu$.
Using $\{\bfeta_d^j, j=1,\ldots , N\}$, the estimates of the mean vector and the covariance matrix
of $\bfH$ verify $\underline{\bfeta}_d = \bfzero_\nu$ and $[C_\bfH] = [I_\nu]$.
We define the matrix $[\eta_d] = [ \bfeta_d^1 \ldots \bfeta_d^N]  \in\MM_{\nu,N}$ whose columns are the $N$ realizations of $\bfH$, which is such that
\begin{equation} \label{eq:Meth-1}
\Vert [\eta_d] \Vert^2 = \tr\{[\eta_d]^T\, [\eta_d]\} = \sum_{j=1}^N \Vert\bfeta_d^j\Vert^2 = \nu(N-1)\, .
\end{equation}
\subsection{PLoM method with no group (No-Group PLoM)}
\label{sec:Meth-3}
The PLoM analysis used is the one presented in \cite{Soize2016} whose complete mathematical analysis is performed in \cite{Soize2020c}.

\subsubsection{Nonparametric estimate of the pdf of $\bfH$}
\label{sec:Meth-3.1}
A modification of the multidimensional Gaussian kernel-density estimation method \cite{Duong2008,Duong2005,Filippone2011,Zougab2014}
is used for constructing the nonparametric estimate $p_\bfH^{(N)}$ on $\RR^\nu$ of the pdf $p_\bfH$ of random vector $\bfH$, which is written (see Theorem~3.1 of \cite{Soize2020c} for the convergence with respect to $N$) as
\begin{equation} \label{eq:Meth-2}
p_\bfH^{(N)}(\bfeta) = \frac{1}{N} \sum_{j=1}^N \, \frac{1}{(\sqrt{2\pi}\,\widehat s)^\nu}\, \exp\{-\frac{1}{2\widehat s^2}\Vert\frac{\widehat s}{s} \, \bfeta_d^j -\bfeta\Vert^2\} \quad , \quad \forall \bfeta \in \RR^\nu \, ,
\end{equation}
in which $s = (N(\nu+2)/4)^{-{1}/{(\nu+4)}}$ is the usual Silverman bandwidth (since $[C_\bfH] = [I_\nu]$, see for instance, \cite{Bowman1997}) and where $\widehat s =   s\,(s^2 \! + \! (N \! - \! 1)/{N})^{-1/2}$ has been introduced in order that $\int_{\RR^\nu} \bfeta\, p_\bfH^{(N)}(\bfeta)\, d\bfeta = 0_\nu$ and $\int_{\RR^\nu} \bfeta \otimes \bfeta\,\, p_\bfH^{(N)}(\bfeta)\, d\bfeta = [I_\nu]$.

\subsubsection{Construction of a reduced-order diffusion-map basis (ROB-DM)}
\label{sec:Meth-3.2}
To identify the subset around which the initial data are concentrated, the PLoM relies on the diffusion-map method \cite{Coifman2006,Lafon2006}. The Gaussian kernel is used.
Let $[K]$  and  $[b]$ be the matrices such that, for all $i$ and $j$ in $\{1,\ldots , N\}$,
$[K]_{ij} = \exp\{-(4\,\varepsilon_\DM)^{-1} \Vert\bfeta_d^i-\bfeta_d^j\Vert^2\}$ and
$[b]_{ij} = \delta_{ij} \, b_i$ with $b_i = \sum_{j=1}^N [K]_{ij}$, in which $\varepsilon_\DM >0$ is a smoothing parameter (the non symmetric matrix $\PP = [b]^{-1}[K]\in\MM_N$ is the transition matrix of a Markov chain that yields the probability of transition in one step).
The eigenvalues $\lambda_1,\ldots,\lambda_N$ and the associated eigenvectors $\bfpsi^1,\ldots,\bfpsi^N$ of the right-eigen\-value problem
$[\PP]\, \bfpsi^\alpha = \lambda_\alpha\, \bfpsi^\alpha$ are such that $ 1=\lambda_1 > \lambda_2 \geq \ldots \geq \lambda_N $
and are computed by solving the generalized eigenvalue problem $[K]\, \bfpsi^\alpha = \lambda_\alpha\, [b]\,\bfpsi^\alpha$ with the normalization
$<\![b]\,\bfpsi^\alpha,\bfpsi^\beta\!> =\delta_{\alpha\beta}$. The eigenvector $\bfpsi^1$ associated with $\lambda_1=1$ is a constant vector.
For a given integer $\kappa \geq 0$,  the diffusion-map basis $\{\bfg^1,\ldots,\bfg^\alpha,\ldots, \bfg^N\}$ is a vector basis of $\RR^N$ defined by $\bfg^\alpha = \lambda^\kappa_\alpha\,\bfpsi^\alpha$.
For a given integer $m$ with $3 \leq m \leq N$, the reduced-order diffusion-map basis of order $m$ is defined as the family
$\{\bfg^1,\ldots,\bfg^m\}$ that is represented by the matrix $[g_m] = [\bfg^1 \ldots \bfg^m]\in\MM_{N,m}$ with
$\bfg^\alpha = (g_1^\alpha,\ldots ,g_N^\alpha)$ and $[g_m]_{\ell\alpha} = g_\ell^\alpha$. This ROB-DM depends on two parameters, $\varepsilon_\DM$ and $m$, which have to be identified. It is proven in \cite{Soize2020c}, that the PLoM method does not depend of $\kappa$ that can therefore be chosen to $0$.

It should be noted that, if $\nu=1$, then there is no really interest
to use the ROB-DM and in this case, we propose to take $m=N$ and
$[g_N]= [I_N]$. For non-trivial applications analyzed with the PLoM
without partition, we always have $\nu> 1$ and even, $1\ll\nu \leq
n$. However for the PLoM method with partition, optimal partitions
can be found for which some groups may have dimension $1$  (hence the consideration of possible cases of this type).

\subsubsection{Novel algorithm for identifying the optimal values $\varepsilon_\opt$ and $m_\opt$ of $\varepsilon_\DM$ and $m$}
\label{sec:Meth-3.3}
Let us assume that $\nu\geq 2$. For estimating the optimal values
$\varepsilon_\opt$ of $\varepsilon_\DM$ and $m_\opt$ of $m$, the
criterion of the eigenvalues given  in Section~5.2 of
\cite{Soize2020c} must be satisfied for the PLoM method to be applicable. This criterion can be summarized as follows. We have to find the value $m_\opt \leq N$ of $m$ and the smallest value $\varepsilon_\opt > 0 $ of $\varepsilon_\DM$ such that
\begin{equation} \label{eq:Meth-3}
1=\lambda_1 > \lambda_2(\varepsilon_\opt) \simeq \ldots \simeq \lambda_{m_\ppopt}(\varepsilon_\opt) \gg \lambda_{m_\ppopt+1}(\varepsilon_\opt)\geq \ldots \geq \lambda_N(\varepsilon_\opt) > 0\, ,
\end{equation}
with a jump in amplitude equal to $10$ between $\lambda_{m_\ppopt}(\varepsilon_\opt)$ and $\lambda_{m_\ppopt +1}(\varepsilon_\opt)$.
This property means that we have to find $m_\opt \leq N$ and the smallest positive value $\varepsilon_\opt$  in order (i) to have $\lambda_2(\varepsilon_\opt) < 1$ (one must not have several eigenvalues in the neighborhood of $1$) and (ii) to obtain a plateau  for $\lambda_2(\varepsilon_\opt)$ to $\lambda_{m_\ppopt}(\varepsilon_\opt)$ with a jump of amplitude $10$ between $\lambda_{m_\ppopt}(\varepsilon_\opt)$ and $\lambda_{m_\ppopt +1}(\varepsilon_\opt)$.
A further in-depth analysis makes it possible to state the following new criterion and algorithm to easily estimate $\varepsilon_\opt$ and $m_\opt$. Let $\varepsilon_\DM\mapsto \Jump(\varepsilon_\DM)$ be the function on $ ] 0 , + \infty [ $ defined by
$\Jump(\varepsilon_\DM) = \lambda_{m_\ppopt+1}(\varepsilon_\DM) / \lambda_2(\varepsilon_\DM)$.
\begin{algorithm}
\caption{Algorithm for estimating the optimal values $m_\opt$ of $m$ and $\varepsilon_\opt$ of $\varepsilon_\DM$}
\label{algorithm:Meth-1}
\begin{algorithmic}[1]
\IF{$\nu = 1$}
\STATE{Set the value of $m$ to $m_\opt = N$ and $[g_N] = [I_N]$.}
\ENDIF
\IF{$\nu \geq 2$}
\STATE{Set the value of $m$ to $m_\opt = \nu+1$.}
\STATE{Identify the smallest possible value $\varepsilon_\opt$ of $\varepsilon_\DM$  in order that $\Jump(\varepsilon_\opt) \leq 0.1$
and such that Eq.~\eqref{eq:Meth-3} be verified.}
\ENDIF
\end{algorithmic}
\end{algorithm}
The novel algorithm is thus given in Algorithm~\ref{algorithm:Meth-1} and
Figure~\ref{fig:figureMeth-1} shows an illustration:  we have $m_\opt =
\nu+1= 61$; the optimal value of $\varepsilon_\DM$ that satisfies the
criteria is $\varepsilon_\opt= 65$ and yields
Figure~\ref{fig:figureMeth-1a}; if a smaller value than $65$ is chosen, for
instance the value $5$, then there will be many eigenvalues close to
$1$ as shown in Figure~\ref{fig:figureMeth-1b}; if the smallest value for
$\varepsilon_\DM$ is not selected, for example taking the value $100$,
then  the plateau is not obtained as shown in
Figure~\ref{fig:figureMeth-1c}. For these two bad values of
$\varepsilon_\DM$, the calculated diffusion-map basis is not adapted
to the PLoM procedure.
\begin{figure}[tbhp]
  \centering
  \subfloat[Correct value $\varepsilon_\opt$ of $\varepsilon_\DM$]{\label{fig:figureMeth-1a}\includegraphics[width=4.2cm]{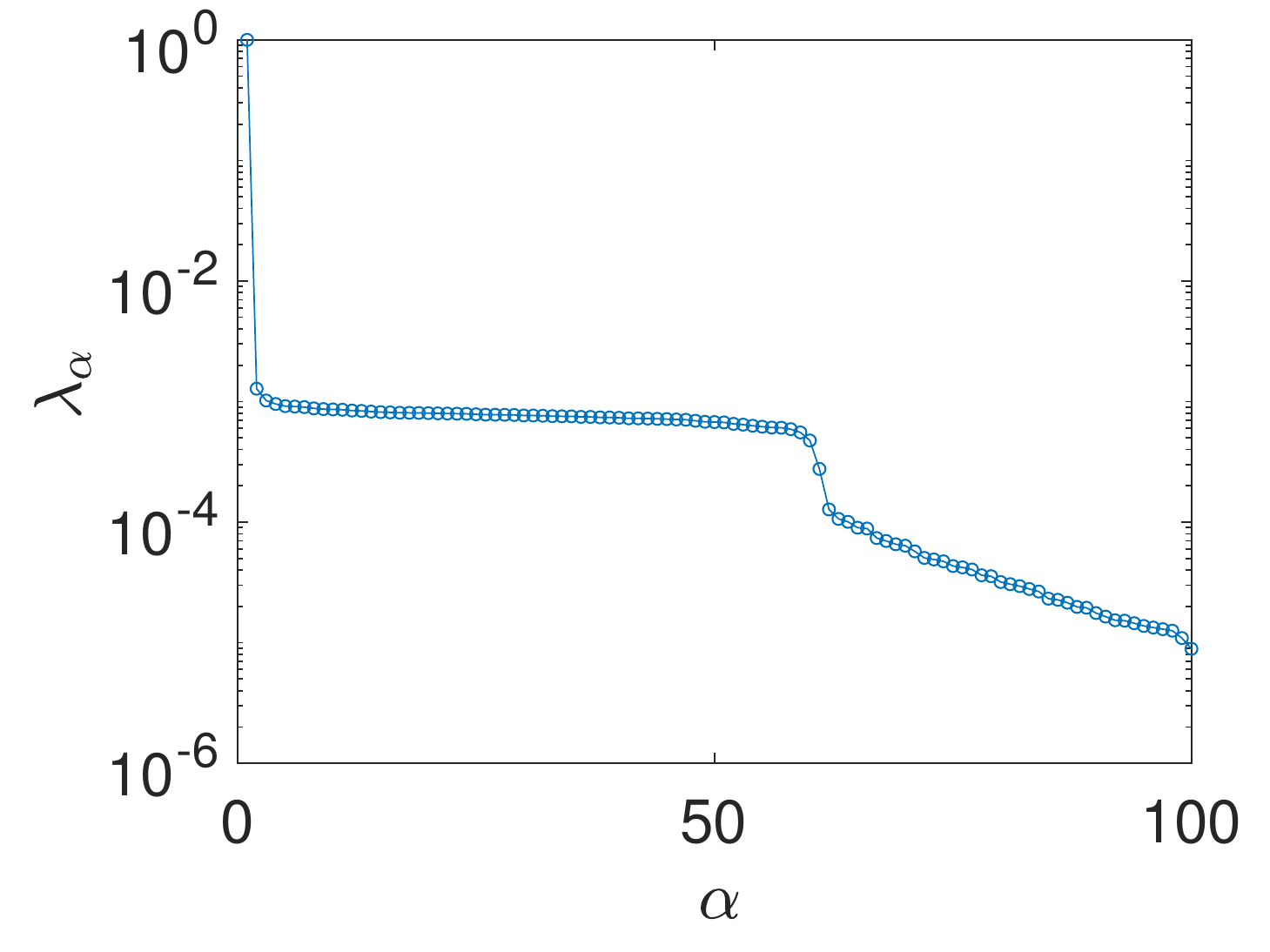}}\hfil
  \subfloat[Bad value of $\varepsilon_\DM$ that is smaller than  $\varepsilon_\opt$]{\label{fig:figureMeth-1b}\includegraphics[width=4.2cm]{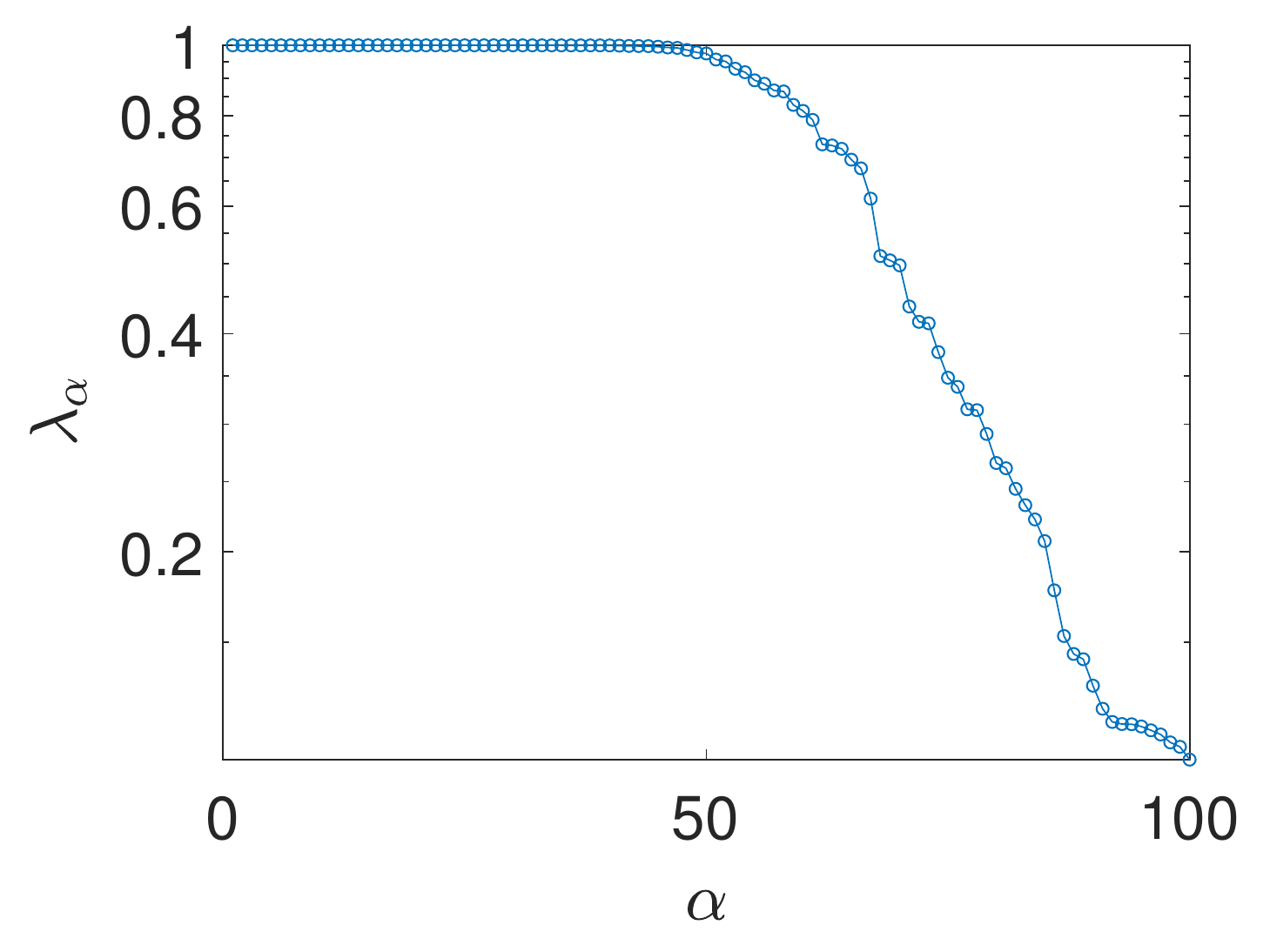}} \hfil
  \subfloat[Bad value of $\varepsilon_\DM$ that is larger than  $\varepsilon_\opt$]{\label{fig:figureMeth-1c}\includegraphics[width=4.2cm]{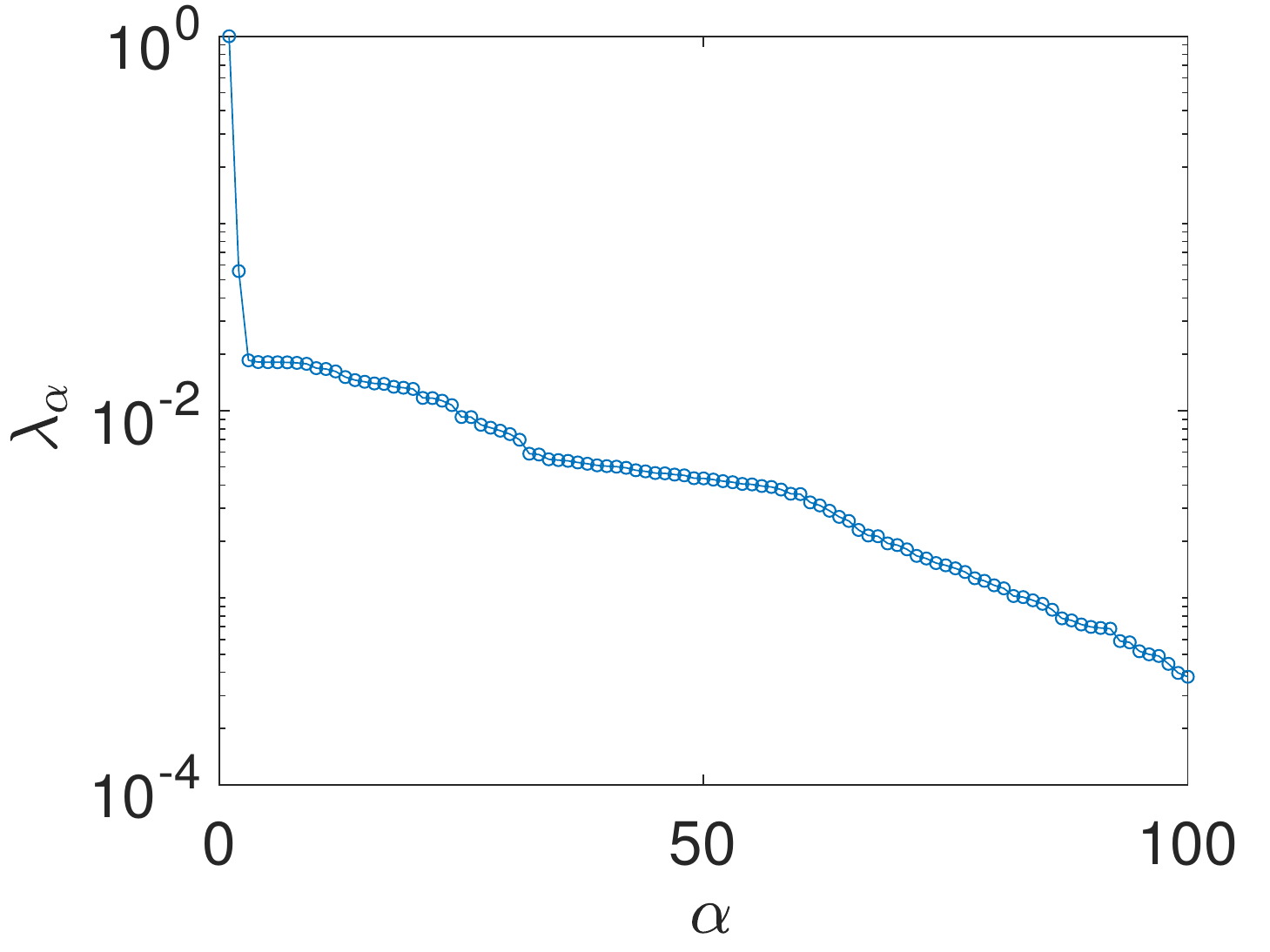}}
  \caption{Illustration of the  criterion effects defined by Eq.~\eqref{eq:Meth-3} for identifying $\varepsilon_\opt$.}
  \label{fig:figureMeth-1}
\end{figure}
\subsubsection{Random matrices $[\,\bfH^N]$, $[\,\bfH^N_{m}]$, $[\,\bfH^N_{m_\ppopt}]$, and MCMC generator}
\label{sec:Meth-3.4}
Let $\bfH^{(N)}$ be the $\RR^\nu$-valued random variable defined on $(\Theta,\curT,\curP)$ for which the pdf is $p_\bfH^{(N)}$ defined by Eq.~\eqref{eq:Meth-2}. Let $[\bfH^N]$ be the random matrix with values in $\MM_{\nu,N}$ such that $[\bfH^N] = [\bfH^1\ldots \bfH^N]$ in which $\bfH^1, \ldots , \bfH^N$ are $N$ independent copies of $\bfH^{(N)}$. It can be seen that $E\{\bfH^{(N)}\} = \bfzero_\nu$ and
$E\{\bfH^{(N)}\otimes \bfH^{(N)} \} =[I_\nu]$. Note that $\bfH^1, \ldots , \bfH^N$ are not taken as $N$ independent copies of $\bfH$ whose pdf $p_\bfH$ is unknown, but are taken as $N$ independent copies of $\bfH^{(N)}$ whose pdf $p_\bfH^{(N)}$ is known.
The PLoM method introduces the $\MM_{\nu,N}$-valued random matrix
$[\bfH^N_{m}] = [\bfZ_{m}]\, [g_{m}]^T$ with $3 \leq m \leq N$ ,
corresponding to a data-reduction representation of random matrix
$[\bfH^N]$, in which $[g_{m}]$ is the ROB-DM and where $[\bfZ_{m}]$ is
a $\MM_{\nu,m}$-valued random matrix for which its probability measure
$p_{[\bfZ_{m}]}([z])\, d[z]$ is explicitly described  by Proposition~2
of \cite{Soize2020c}. In the PLoM method, the MCMC generator of random
matrix $[\bfZ_{m}]$ belongs to the class of Hamiltonian Monte Carlo
methods \cite{Neal2011}, is explicitly described in \cite{Soize2016},
and is mathematically detailed in Theorem~6.3 of
\cite{Soize2020c}. For generating the learned set, the best
probability measure of $[\,\bfH^N_{m}]$  is obtained for $m = m_\opt$
and using the previously defined $[g_{m_\ppopt}]$. For these optimal quantities $m_\opt$ and $[g_{m_\ppopt}]$, the generator allows for computing $n_\MC$ realizations $\{[\bfz_\ar^{\ell}],\ell=1,\ldots , n_\MC\}$ of  $[\bfZ_{m_\ppopt}]$ and therefore, for deducing the $n_\MC$ realizations $\{[\bfeta_\ar^\ell],\ell=1,\ldots ,n_\MC\}$ of $[\bfH^N_{m_\ppopt}]$. The reshaping of matrix $[\bfeta_\ar^\ell] \in \MM_{\nu,N}$ allows for obtaining $N_\ar = n_\MC \times N$ additional realizations
$\{\bfeta_\ar^{\ell'},\ell' =1,\ldots ,N_\ar\}$ of $\bfH$. These additional realizations allow for estimating converged statistics on $\bfH$ and then on $\bfX$, such as pdf, moments, or conditional expectation of the type $E\{\bfxi(\bfQ) \,\vert \,\bfW = \bfw_0\}$ for $\bfw_0$ given in $\RR^{n_w}$ and for any given vector-valued function $\bfxi$ defined on $\RR^{n_q}$.

\subsubsection{Quantifying the concentration of the probability measure of random matrix $[\bfH_{m_\ppopt}^N]$}
\label{sec:Meth-3.5}
In \cite{Soize2020c}, for $3\leq m\leq N$, we have introduce a $L^2$-distance $d_N(m)$ of random matrix $[\bfH_m^N]$ to matrix $[\eta_d]$  in order to quantify the concentration of the probability measure of random matrix $[\bfH^N_{m}]$, which is informed by the initial data set represented by matrix $[\eta_d]$. The square of this distance is defined by
\begin{equation} \label{eq:Meth-4}
d_N^2(m) = E\{\Vert [\bfH_m^N] - [\eta_d]\Vert^2\} /\Vert [\eta_d]\Vert^2\, .
\end{equation}
Let $\curM_\opt =\{m_\opt,m_\opt+1,\ldots , N\}$ in which $m_\opt$ is the optimal value of $m$ previously defined.
Theorem~7.8 of \cite{Soize2020c} shows that
$\min_{m\in\curM_\opt} d_N^2(m)  \leq 1 + m_\opt/(N-1) < d_N^2(N)$
which means that the PLoM method, for $m=m_\opt$ and $[g_{m_\ppopt}]$ is a better method than the usual one corresponding to $d_N^2(N)= 1+N/(N-1)$.
Using the $n_\MC$ realizations $\{[\bfeta_\ar^\ell],\ell=1,\ldots ,n_\MC\}$ of $[\bfH^N_{m_\ppopt}]$, we have the estimate
$d_N^2(m_\opt) \simeq (1/n_\MC)\sum_{\ell=1}^{n_\ppMC}\{\Vert [\bfeta_\ar^\ell]  - [\eta_d]\Vert^2\} /\Vert [\eta_d]\Vert^2$.

\subsection{PLoM analysis with group (With-Group PLoM)}
\label{sec:Meth-4}
In this section, for $\nu\geq 2$, we present the extension of the PLoM analysis for which statistically independent  groups are constructed using an optimal partition  of
random vector $\bfH$.

\subsubsection{Construction of the optimal partition of $\bfH$}
\label{sec:Meth-4.1}
From the training set $\{\bfeta^{j},j=1,\ldots,$ $N\}$, the optimal partition of $\bfH =(H_1,\ldots , H_\nu)$ is performed using the algorithm proposed in \cite{Soize2017f}. Such a partition is composed of $n_p$ groups consisting in $n_p$ mutually independent random vectors $\bfY^1,\ldots ,\bfY^{n_p}$. Since $\bfH$ is a normalized random vector (zero mean vector and covariance matrix equal to the identity matrix), for $i=1,\ldots, n_p$, $\bfY^i$ is a normalized $\RR^{\nu_i}$-valued random variable
$\bfY^i = (Y^i_1,\ldots , Y^i_{\nu_i}) = (H_{r_1^i},\ldots , H_{r_{\nu_i}^i})$ in which $1 \leq r_1^i < r_2^i < \ldots < r_{\nu_i}^i\leq \nu$, with $\nu =\nu_1 + \ldots + \nu_{n_p}$, and where $Y^i_k = H_{r_k^i}$. Random vector $\bfY^i$ is non-Gaussian and such that the estimate of its mean vector is $\underline{\bfeta^i} = \bfzero_{\nu_i}$ and the estimate of its covariance matrix is $[C_{\bfY^i}] = [I_{\nu_i}]$.
We then have $\bfH = {perm}(\bfY^1,\ldots ,\bfY^{n_p})$ in which \textit{perm} is the permutation operator acting on the components of vector $\widetilde\bfH = (\bfY^1,\ldots ,\bfY^{n_p})$ in order to reconstitute $\bfH = perm(\widetilde\bfH)$.
For each group $i$, the training set is represented by the matrix $[\eta^i_d]\in\MM_{\nu_i,N}$ whose columns are the $N$ realizations $\{\bfeta_d^{i,j},j=1,\ldots,N \}$ of the $\RR^{\nu_i}$-valued random variable $\bfY^i$,  which are deduced from an adapted extraction (due to the permutations)  of the components of vectors $\{\bfeta_d^{j},j=1,\ldots,N\}$.
The partition is identified by constructing the function $i_\pref\mapsto \tau(i_\pref)$  of the mutual information defined by Eq.~(3.44) of \cite{Soize2017f} and then by deducing the optimal level $i_\pref^{\,\opt}$ defined by Eq.~(3.46) of \cite{Soize2017f}.

\subsubsection{Use of the PLoM for each independent group}
\label{sec:Meth-4.2}
Let $i$ be fixed in $\{1,\ldots , n_p\}$. The PLoM method (summarized in Section~\ref{sec:Meth-3}) is applied to the $\RR^{\nu_i}$-valued random variable $\bfY^i$ of the optimal partition  $\bfY^1,\ldots ,\bfY^{n_p}$ of  $\bfH ={perm}(\bfY^1,\ldots ,\bfY^{n_p})$. The parameters of the PLoM are thus the following.

\noindent 1) The Silverman bandwidth is $s_i = (N(\nu_i+2)/4)^{-{1}/{(\nu_i+4)}}$  (since $[C_{\bfY^i}] = [I_{\nu_i}]$) and the modified bandwidth is $\widehat s_i =   s_i\,(s_i^2 \! + \! (N \! - \! 1)/{N})^{-1/2}$.

\noindent 2) Algorithm~\ref{algorithm:Meth-1} is used. If $\nu_i=1$, then $m_{i,\opt} = N$ and $[g^i_N] = N$. If $\nu_i\geq 2$, the optimal parameter $m_{i,\opt}$ of the dimension $m_i$ of the ROB-DM${}^i$ is such that $m_{i,\opt} = \nu_i +1$. The optimal parameter $\varepsilon_{i,\opt}$ of  $\varepsilon_{i,\DM}$ is calculated as explained in Section~\ref{sec:Meth-3.2}. The ROB-DM${}^i$ of order $m_{i,\opt}$ is represented by the matrix $[g^i_{m_{i,\ppopt}}]\in\MM_{N,m_{i,\ppopt}}$.

\noindent 3) The learned set of the random matrix $[\,\bfY^{N,i}_{m_{i,\ppopt}}] = [\bfZ^i_{m_{i,\ppopt}}]\, [g^i_{m_{i,\ppopt}}]^T$  is computed for $m_i = m_{i,\opt}$ and by using $[g^i_{m_{i,\ppopt}}]$.  Finally, the $n_\MC$ realizations $\{ [\bfeta_\ar^{i,\ell}],\ell=1,\ldots ,n_\MC\}$ of $[\bfY^{N,i}_{m_{i,\ppopt}}]$ are computed with the MCMC generator and by reshaping, we obtain the $N_\ar = n_\MC \times N$ additional realizations $\{\bfeta_\ar^{i,\ell'},\ell'=1,\ldots,N_\ar\}$.

\subsubsection{Possible lost of the normalization}
\label{sec:Meth-4.3}
Numerical experiments have been done for numerous cases with respect to the number of groups and the dimension of each group. These experiments have shown the following.
In general, the mean value of $\bfY^i$, estimated using the
$N_\ar$ additional realizations
$\{\bfeta_\ar^{i,\ell'},\ell'=1,\ldots,N_\ar\}$, is sufficiently close
to zero. Likewise, the estimate of the covariance matrix is
sufficiently close to a diagonal matrix. However, sometimes the
diagonal of the estimated  covariance matrix can be lower than $1$
(for instance $0.5$). Such a case can occur for relatively small value
of $\nu_i$ (but not systematically and not only; this behavior is
application-dependent).  In these situations, normalization and
structure can be recovered by imposing constraints in the PLoM method.

\subsubsection{Constraints on the second-order moments of the components of $\bfY^i$ if loss of normalization occurs}
\label{sec:Meth-4.4}
As explained in Section~\ref{sec:Meth-4.3}, if appropriate for group $i$,
constraints $\{ E\{(Y^i_k)^2\}  = 1 , k=1,\ldots, \nu_i\}$ can be readily introduced  in the PLoM.
For that, we use the method and the iterative algorithm presented in Sections~5.5 and 5.6 of \cite{Soize2020a}.
The method consists of constructing the generator using the
PLoM for each independent group, defined in Section~\ref{sec:Meth-4.2}, and the  Kullback-Leibler minimum cross-entropy principle. The
resulting optimization problem is formulated using Lagrange multipliers associated with the constraints. The optimal solution of the Lagrange multipliers is
computed using an efficient iterative algorithm. At each iteration, the MCMC generator of the PLoM is used.
The  constraints are rewritten as
\begin{equation} \label{eq:Meth-5}
E\{\bfh^i(\bfY^i)\} = \bfb^i \, ,
\end{equation}
in which the function $\bfh^i = (h^i_1,\ldots ,h^i_{\nu_i})$ and the vector $\bfb^i = (b^i_1,\ldots , b^i_{\nu_i})$ are such that
$h^i_k(\bfY^i) = (Y^i_k)^2$ and $b^i_k =1$ for $k$ in $\{1,\ldots , \nu_i\}$.
Eqs.~(71) and (72) of \cite{Soize2020a} involve the Lagrange multiplier $\bflambda = (\lambda_1,\ldots,\lambda_{\nu_i})\in \RR^{\nu_i}$ associated with the constraints defined by Eq.~\eqref{eq:Meth-5}. These two equations, which define the nonlinear mapping $[u]\mapsto [L^i_\bflambda([u])]$ from $\MM_{\nu_i,N}$ into $\MM_{\nu_i,N}$ (drift of the It\^o stochastic differential equation of the PLoM generator), have to be modified as follows.
For $\alpha=1,\ldots ,\nu_i$, for $\ell=1,\ldots ,N$,  and for $[u] = [\bfu^1 \ldots \bfu^N]$ in $\MM_{\nu_i,N}$, we have
\begin{align}
& [L^i_\bflambda([u])]_{\alpha\ell}  = \frac{1}{\rho_i(\bfu^\ell)} \frac{\partial \rho_i(\bfu^\ell)}{ \partial u_\alpha^\ell}
                       - 2\lambda_\alpha u_\alpha^\ell \, , \nonumber \\
 \rho_i(\bfu^\ell)  = & \frac{1}{N} \sum_{j=1}^N \, \frac{1}{(\sqrt{2\pi}\,\widehat s_i)^{\nu_i}} \, \exp\{-\frac{1}{2\widehat s_i^2}\Vert\frac{\widehat s_i}{s_i} \, \bfeta_d^{i,j} -\bfu^\ell\Vert^2\} \, . \nonumber
\end{align}
The iteration algorithm computes the sequence $\{\bflambda^\iota\}_{\iota \geq 1}$ that is convergent. If difficulties of convergence appear, a relaxation factor (less than $1$) is introduced for computing $\bflambda^{\iota+1}$ as a function of $\bflambda^\iota$.
For controlling the convergence as a function of iteration number $\iota$, we use the error function $\iota\mapsto \perr_i(\iota)$  defined by
\begin{equation} \label{eq:Meth-6}
\perr_i(\iota) =  \Vert \bfb^i - E\{\bfh^i(\bfY^i_{\bflambda^\iota})\} \Vert / \Vert \bfb^i\Vert \, .
\end{equation}
At each iteration $\iota$, $E\{\bfh^i(\bfY^i_{\bflambda^\iota})\}$ is estimated with the $N_\ar$ additional realizations deduced by reshaping of the
$n_\MC$ realizations of the $\MM_{\nu_i,N}$-valued random matrix $[\bfY^{N,i}_{m_{i,\ppopt}}(\bflambda^\iota)]$ that depends on $\bflambda^\iota$. These realizations are generated by the MCMC algorithm of the PLoM under the constraints.

\subsubsection{Learned data set generated by With-Group PLoM}
\label{sec:Meth-4.5}
We have seen above (see Section~\ref{sec:Meth-4.2}-(3) how  the learned set $\{[\bfeta_\ar^{i,\ell},\ell=1,\ldots,n_\MC\}$ of random matrix $[\bfY^{N,i}_{m_{i,\ppopt}}]$ are generated using With-Group PLoM for each group $i=1,\ldots, n_p$ (using or not the constraints). From this information, we can directly deduce the learned set
$\{ [\bfeta_\ar^{wg,\ell}],\ell =1,\ldots,n_\MC \}$ of $[\bfH^{wg,N}_{\bfm_\ppopt}]$ that corresponds to the concatenation with an adapted extraction of the rows (due to the permutations) of matrices  $\{[\bfY^{N,i}_{m_{i,\ppopt}}] , i=1,\ldots , n_p\} $ and where $\bfm_\opt = (m_{1,\opt},\ldots , m_{n_p,\opt})$. We have introduced a superscript $wg$ for distinguishing With-Group PLoM from No-Group PLoM.
The reshaping of matrix $[\bfeta_\ar^{wg,\ell}] \in \MM_{\nu,N}$ allows for obtaining $N_\ar = n_\MC \times N$ additional realizations
$\{\bfeta_\ar^{wg,\ell'},\ell' =1,\ldots ,N_\ar\}$ of $\bfH$, computed using With-Group PLoM.

\subsubsection{Quantifying the concentration of the probability measure of random matrices $[\bfH^{wg,N}_{\bfm_\ppopt}]$ and $\{[\bfY^{N,i}_{m_\ppopt}], i=1,\ldots, n_p\}$}
\label{sec:Meth-4.6}
For $\bfm_\opt = (m_{1,\opt},\ldots ,m_{n_p,\opt})$, the square of the
distance  of the random matrix $[\bfH^{wg,N}_{\bfm_\ppopt}]$ to matrix
$[\eta_d]$ is directly given by Eq.~\eqref{eq:Meth-4} which is rewritten
here as,
\begin{equation} \label{eq:Meth-7}
d_{wg,N}^2(\bfm_\opt) = E\{\Vert [\bfH^{wg,N}_{\bfm_\ppopt}] - [\eta_d]\Vert^2\} /\Vert [\eta_d]\Vert^2\, .
\end{equation}
The mathematical expectation is estimated using the $n_\MC$ realizations $\{[\bfeta_\ar^{wg,\ell}],\ell=1,\ldots ,n_\MC\}$. Using again  Eq.~\eqref{eq:Meth-4},  for $i \in\{1,\ldots , n_p\}$, the square of the distance  of random matrix $[\bfY^{N,i}_{m_{i,\ppopt}}]$ to matrix $[\eta^i_d]$ is given by
\begin{equation} \label{eq:Meth-8}
d_{i,N}^2(m_{i,\opt}) = E\{\Vert [\bfY^{N,i}_{m_{i,\ppopt}}] - [\eta^i_d]\Vert^2\} /\Vert [\eta^i_d]\Vert^2\, ,
\end{equation}
which is estimated using the $n_\MC$ realizations $\{[\bfeta_\ar^{i,\ell}],\ell=1,\ldots ,n_\MC\}$. Eq.~\eqref{eq:Meth-7} contains the information defined by Eq.~\eqref{eq:Meth-8}. Indeed it is easy to verify that we have the relation
\begin{equation} \label{eq:Meth-9}
d_{wg,N}^2(\bfm_\opt) = \sum_{i=1}^{n_p} (\nu_i/\nu) \, d_{i,N}^2(m_{i,\opt})\, .
\end{equation}
\subsubsection{How to quantify the gain obtained by using With-Group PLoM instead of No-Group PLoM when $n_p > 1$}
\label{sec:Meth-4.7}
For a given application, the first method consists in numerically comparing the estimates of $d_{wg,N}^2(\bfm_\opt)$ defined by Eq.~\eqref{eq:Meth-7} with $d_N^2(m_\opt)$ defined by
Eq.~\eqref{eq:Meth-4}. If there is a gain, we must have
\begin{equation} \label{eq:Meth-10}
d_{wg,N}^2(\bfm_\opt) < d_N^2(m_\opt)\, .
\end{equation}
This expected inequality for any applications for which $n_p > 1$ is
reinforced by the second method, which is encapsulated by the following proposition.
\begin{proposition}[Probability upper bound of the measure of concentration] \label{proposition:Meth-1}
Let $\varepsilon$ be a given real number such that $0< \varepsilon < 1$.
Let $d_N^2(m_\opt)$ be defined by Eq.~\eqref{eq:Meth-4} for $m=m_\opt$. We then have
\begin{equation} \label{eq:Meth-11}
{\rm Proba}\{\Vert [\bfH_{m_\ppopt}^N] - [\eta_d]\Vert^2 /\Vert [\eta_d]\Vert^2 \geq \varepsilon\} \leq d_N^2(m_\opt)/\varepsilon \, .
\end{equation}
Let $r$ be the positive real number (geometric mean) such that
$r = \{\Pi_{i=1}^{n_p} d_{i,N}^2(m_{i,\opt})\}^{1/n_p}$ in which $d_{i,N}^2(m_{i,\opt})$ is defined by
Eq.~\eqref{eq:Meth-8}.  We then have
\begin{equation} \label{eq:Meth-12}
{\rm Proba}\{\Vert [\bfH^{wg,N}_{\bfm_\ppopt}] - [\eta_d]\Vert^2 /\Vert [\eta_d]\Vert^2 \geq \varepsilon \}  \leq (r/\varepsilon )^{n_p}\, .
\end{equation}
\end{proposition}
\begin{proof} [Proof of Proposition~\ref{proposition:Meth-1}]
(i) Using the Markov inequality to the left hand-side member of Eq.~\eqref{eq:Meth-11} directly yields Eq.~\eqref{eq:Meth-11}.
(ii) Let us introduce the simplified following notations:
$\Xi_i = \Vert [\bfY^{N,i}_{m_{i,\ppopt}}] - [\eta^i_d]\Vert^2$ and $\zeta_i = \Vert [\eta^i_d]\Vert^2$. Therefore, Eq.~\eqref{eq:Meth-8} can be rewritten as
$d_{i,N}^2(m_{i,\opt}) = E\{\Xi_i\}/\zeta_i$. If $\forall i\in \{1,\ldots ,n_p\}$ we have $\Xi_i \geq \varepsilon\, \zeta_i \,\,  a.s$, then
$\Vert [\bfH^{wg,N}_{\bfm_\ppopt}] - [\eta_d]\Vert^2 = \sum_{i=1}^{n_p} \Xi_i \geq \varepsilon \sum_{i=1}^{n_p} \zeta_i = \varepsilon \, \Vert [\eta_d]\Vert^2
\,\, a.s$, that is to say $\Vert [\bfH^{wg,N}_{\bfm_\ppopt}] - [\eta_d]\Vert^2 / \Vert [\eta_d]\Vert^2 \geq \varepsilon \,\, a.s$.
(iii) Using result (ii) above, it can be deduced that
${\rm Proba}\{ \cap_{i=1}^{n_p} \{ \Xi_i/\zeta_i \geq \varepsilon\} \}
   = {\rm Proba}\{\Vert [\bfH^{wg,N}_{\bfm_\ppopt}] - [\eta_d]\Vert^2  / \Vert [\eta_d]\Vert^2 \geq \varepsilon \}$.
\noindent (iv) Due to the partition, the random matrices $[\bfY^{N,1}_{m_{1,\ppopt}}],\ldots , [\bfY^{N,n_p}_{m_{n_p,\ppopt}}]$ are statistically independent, and thus  $\Xi_1,\ldots , \Xi_{n_p}$ are statistically independent. Therefore, we can write,
${\rm Proba}\{ \cap_{i=1}^{n_p} \{ \Xi_i/\zeta_i \geq \varepsilon\} \}
   =    \Pi_{i=1}^{n_p} {\rm Proba}\{\Xi_i /\zeta_i \geq \varepsilon \}$.
\noindent (v) The results (iii) and (iv) above yield
${\rm Proba}\{\Vert [\bfH^{wg,N}_{\bfm_\ppopt}] - [\eta_d]\Vert^2  / \Vert [\eta_d]\Vert^2 \geq \varepsilon \}
   =    \Pi_{i=1}^{n_p} {\rm Proba}\{\Xi_i /\zeta_i \geq \varepsilon \}$.
The use of the Markov inequality allows us to write,
${\rm Proba}\{\Xi_i /\zeta_i \geq \varepsilon \} \leq E\{\Xi_i\}/(\varepsilon\,\xi_i) = d_{i,N}^2(m_{i,\opt})/\varepsilon$. Substituting this inequation into the  right hand-side member of the last equality allows us to write
${\rm Proba}\{\Vert [\bfH^{wg,N}_{\bfm_\ppopt}] - [\eta_d]\Vert^2  / \Vert [\eta_d]\Vert^2 \geq \varepsilon \}
   \leq    \Pi_{i=1}^{n_p} \{ d_{i,N}^2(m_{i,\opt})/\varepsilon \}$ $ = (r/\varepsilon )^{n_p}$,
which is Eq.~\eqref{eq:Meth-12}.
\end{proof}

\section{Application~1}
\label{sec:Ap1}

The probabilistic model is chosen so that the partition in terms of statistically independent groups
is known. This will serve to validate the proposed methodology. This application can easily be reproduced. We directly construct the normalized non-Gaussian $\RR^\nu$-valued random variable $\bfH = (H_1,\ldots , H_\nu)$ with $\nu = 60$. Its probabilistic model is described in Appendix~\ref{AppendixA}. The random vector $\bfX$ from which $\bfH$ is deduced by a PCA is not constructed. It should be noted that this application is very difficult for the learning methods taking into account the high degree of the polynomials in the model, which induces a complexity of the geometry of the support of the probability measure  of $\bfH$.

A reference data set with  $N_\pref = 1\, 000\, 000$ independent realizations and the training set with $N = 1\, 200$ independent realizations $\{\bfeta_d^{j},j=1,\ldots,N\}$ are generated using the probabilistic model of $\bfH$. The learned set is generated by the PLoM method (without or with groups) with $N_\ar = 1\, 200\, 000$ realizations $\{\bfeta_\ar^{\ell},\ell=1,\ldots,N_\ar\}$ ($N_\ar = n_\MC\times N$ with $n_\MC = 1\, 000$).
It should be noted that the mean vector $\underline{\bfeta}$ and  the covariance matrix $[C_\bfH]$ of $\bfH$, which are estimated with the
$N$ realizations of the training set,  are such that $\underline{\bfeta}_d = \bfzero_\nu$ and $[C_\bfH] = [I_\nu]$.
\subsection{PLoM analysis with no group (No-Group PLoM)}
\label{sec:Ap1-1}
Algorithm~\ref{algorithm:Meth-1} is used for the calculation of the reduced-order diffusion-map basis $[g_{m_\ppopt}]$ of the $\RR^\nu$-valued random variable $\bfH$. The optimal dimension is $m_\opt = \nu+1 = 61$.
Figure~\ref{fig:figureAP1-1a} displays the function $\varepsilon_\DM \mapsto \Jump(\varepsilon_\DM)$ and shows that the optimal value $\varepsilon_\opt$ of the smoothing parameter $\varepsilon_\DM$ is $\varepsilon_\opt= 656$ for which  $\Jump(\varepsilon_\opt) = 0.1$. For this value $\varepsilon_\opt$ of $\varepsilon_\DM$, Figure~\ref{fig:figureAP1-1b} shows the graph of function $\alpha\mapsto \lambda_\alpha(\varepsilon_\opt)$. It can be seen that the criterion defined by Eq.~\eqref{eq:Meth-3} is satisfied.
\begin{figure}[tbhp]
  \centering
  \subfloat[Identifying the value $\varepsilon_\opt$ of $\varepsilon_\DM $]{\label{fig:figureAP1-1a}\includegraphics[width=5.5cm]{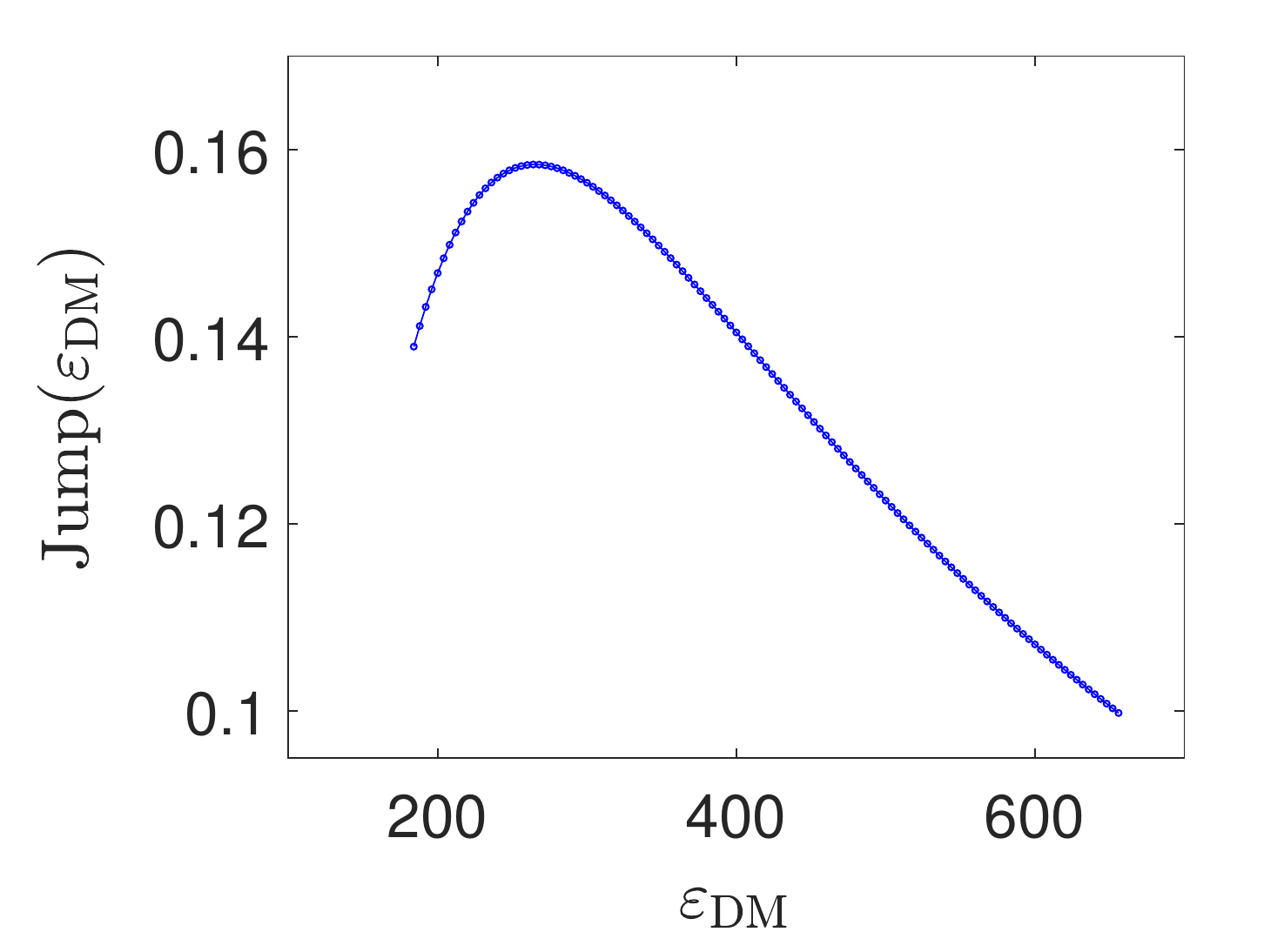}}
  \hfil
  \subfloat[Eigenvalues of the transition matrix] {\label{fig:figureAP1-1b} \includegraphics[width=5.5cm] {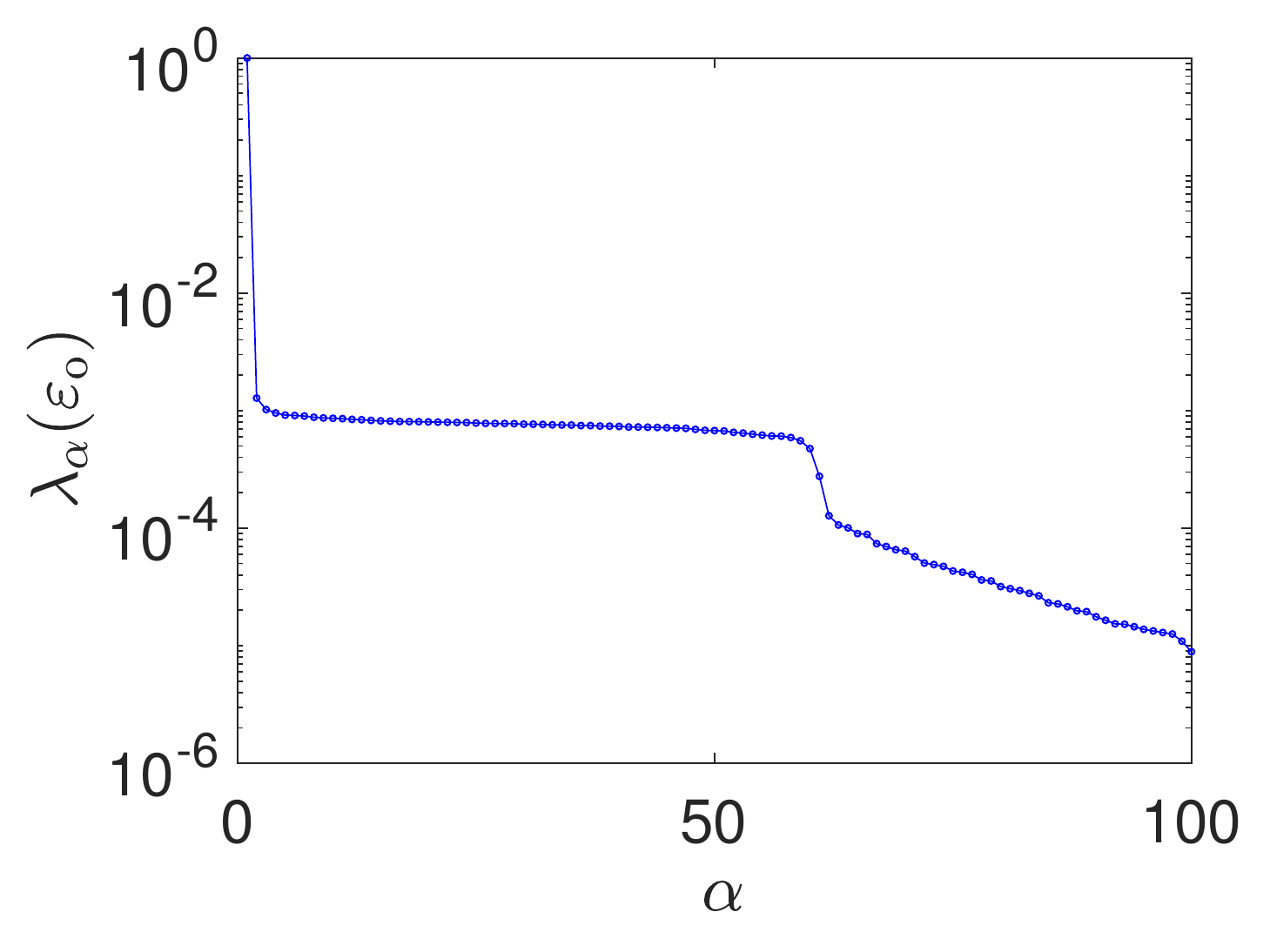}}
  \caption{Diffusion map basis of the standard PLoM (No group).}
  \label{fig:figureAP1-1}
\end{figure}
The PLoM algorithm with no group is then used  for generating the learned set $\{\bfeta_\ar^{\ell},\ell=1,\ldots,N_\ar\}$.
Figure~\ref{fig:figureAP1-2} shows the pdf of each one of the random variables $H_4$, $H_5$, $H_6$, and $H_7$ estimated with the learned set.  Each pdf is estimated (i) with the $N$ realizations of the training set, (ii) with the $N_\pref$ realizations of the reference data set, (iii) with the $N_\ar$ additional realizations generated with the Hamiltonian MCMC algorithm corresponding to the PLoM with $m_\opt = N$ and $[g_{m_\popt}] = [I_N]$, and referenced as "\textit{No-PLoM}", and finally, with the $N_\ar$ realizations of the learned set constructed  with the PLoM for which the partition in groups is not taken into account and referenced as "\textit{No-Group PLoM}" (in this case no constraints are applied). It can be seen that the No-PLoM  estimation yields a big scattering with an important increase of the dispersion (and thus a loss of the concentration of the probability measure) while No-Group PLoM preserves the concentration of the probability measure (as expected) and the pdfs' estimations are good enough. These estimations will be improved by using the PLoM with groups and referenced as "\textit{With-Group PLoM}".
\begin{figure}[tbhp]
  \centering
  \subfloat[pdf of $H_4$]{\label{fig:figureAP1-2a} \includegraphics[width=6.5cm] {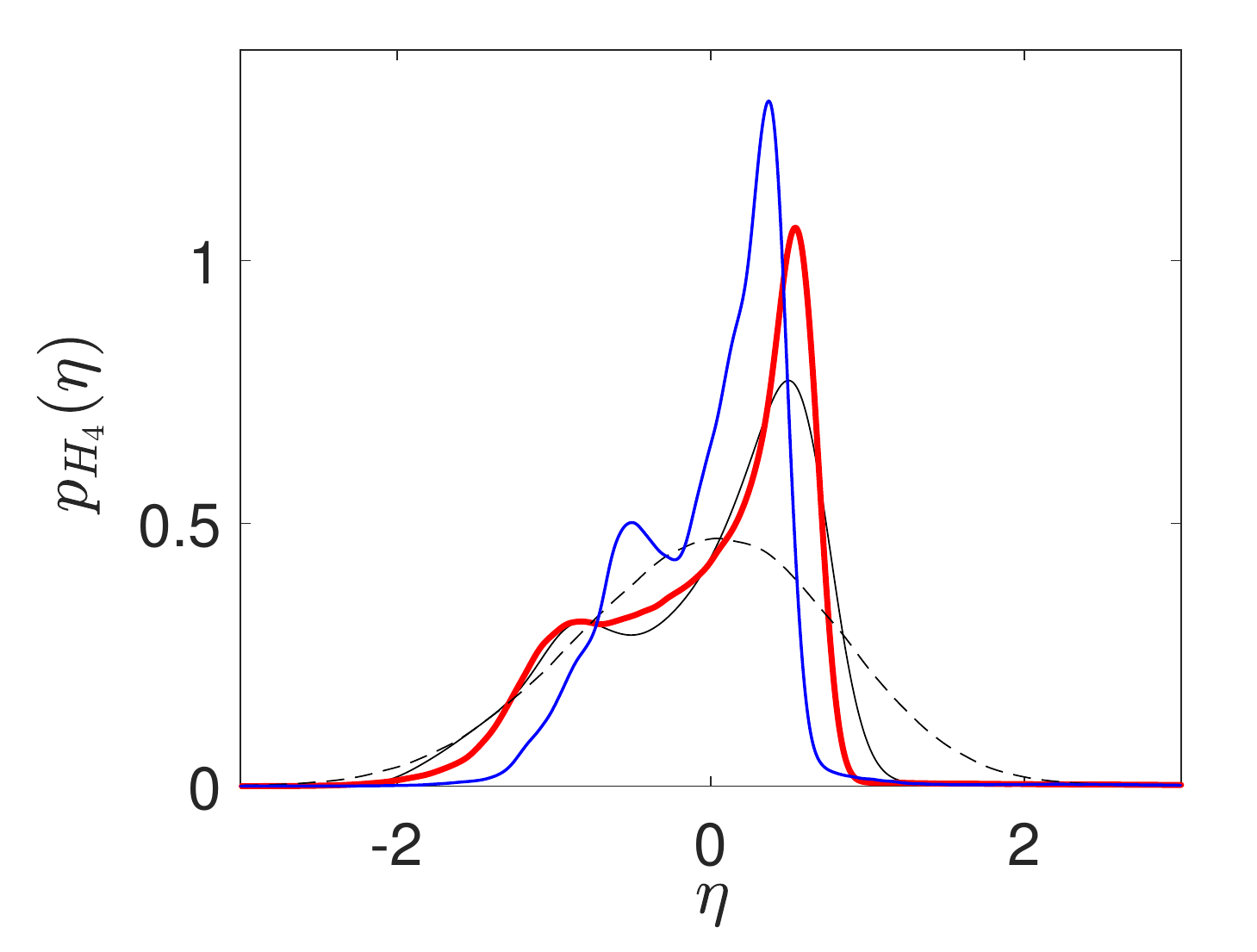}}
  \subfloat[pdf of $H_5$]{\label{fig:figureAP1-2b} \includegraphics[width=6.5cm] {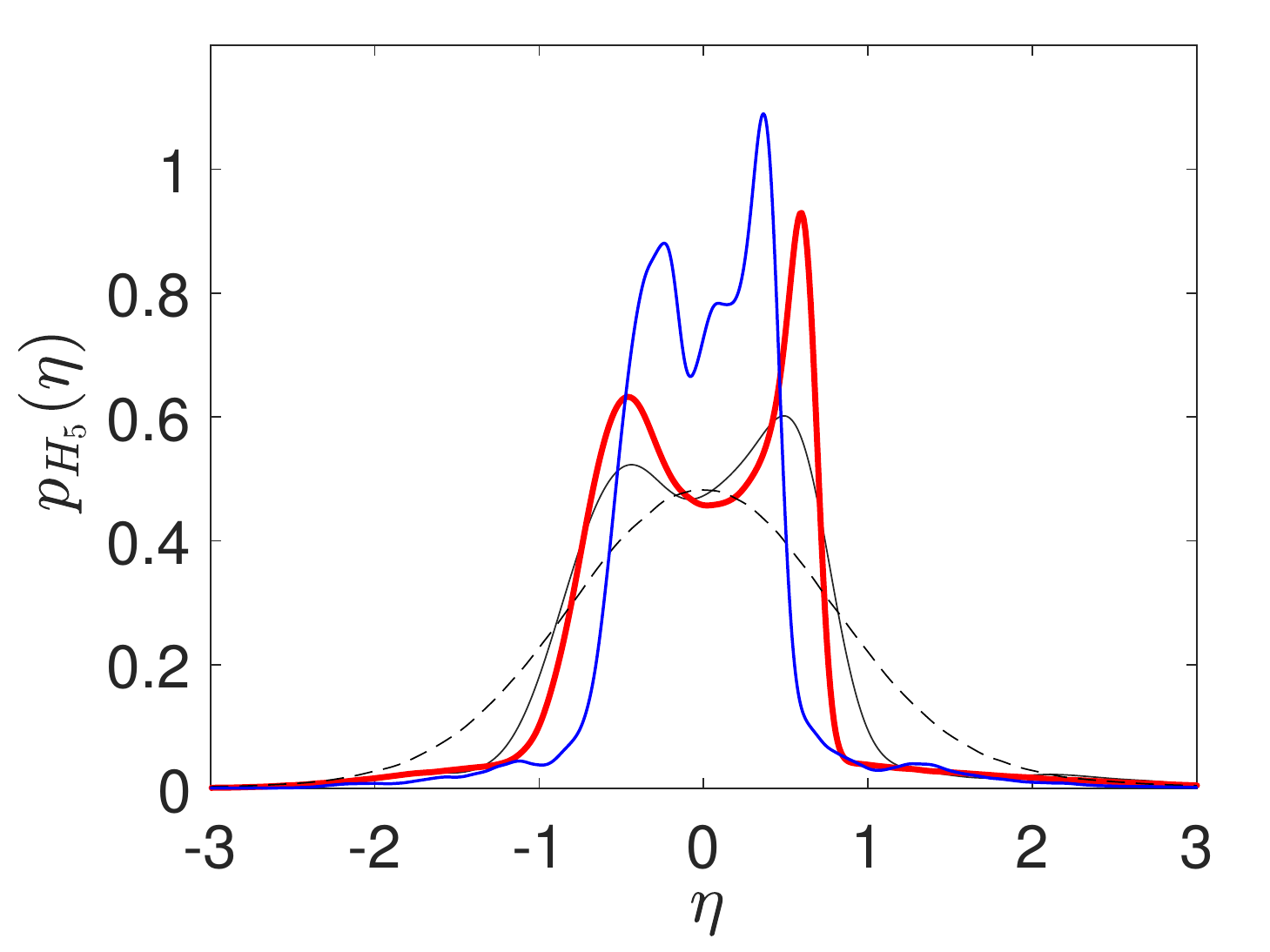}}\\
  \subfloat[pdf of $H_6$]{\label{fig:figureAP1-2c} \includegraphics[width=6.5cm] {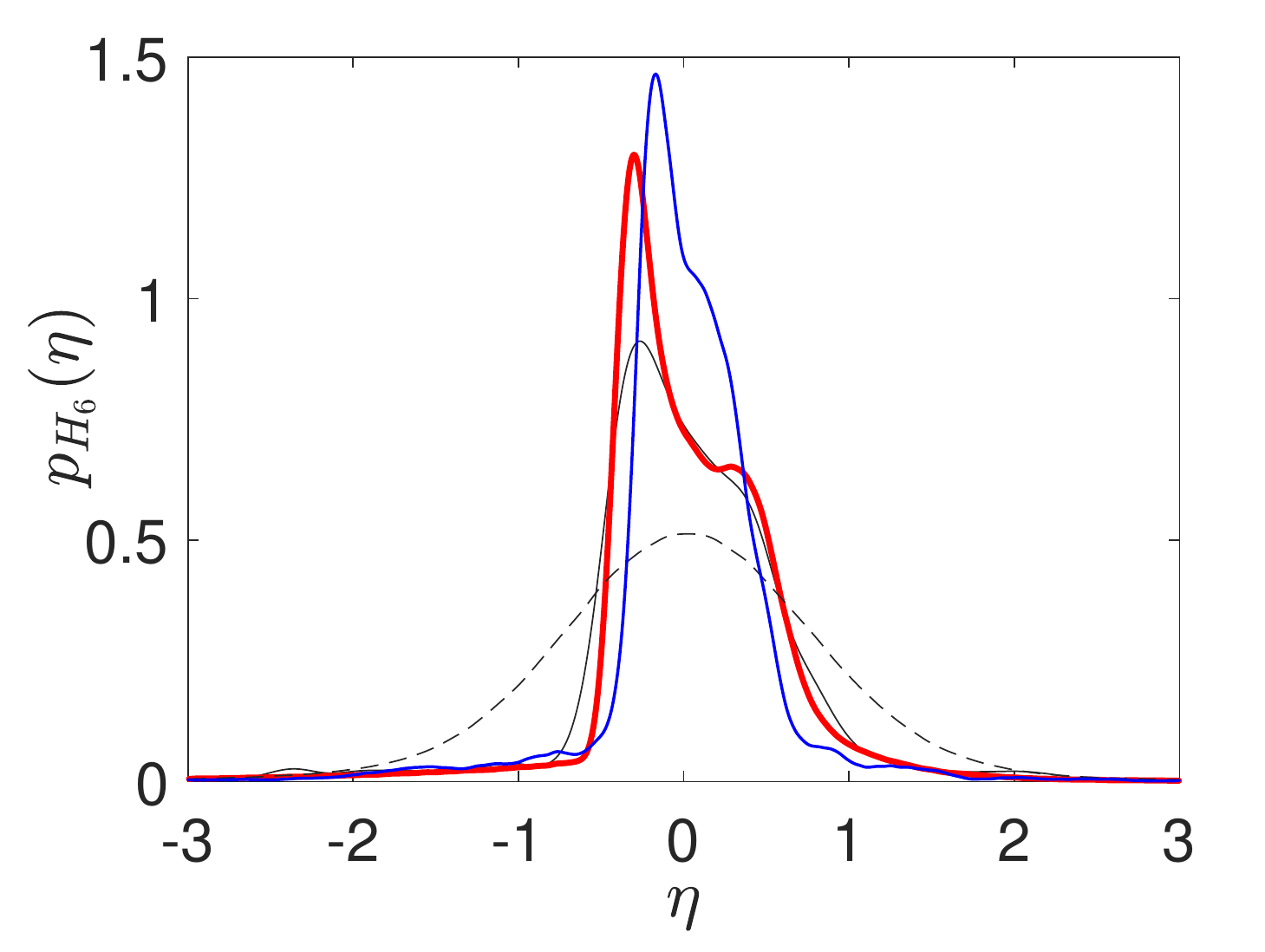}}
  \subfloat[pdf of $H_7$]{\label{fig:figureAP1-2d} \includegraphics[width=6.5cm] {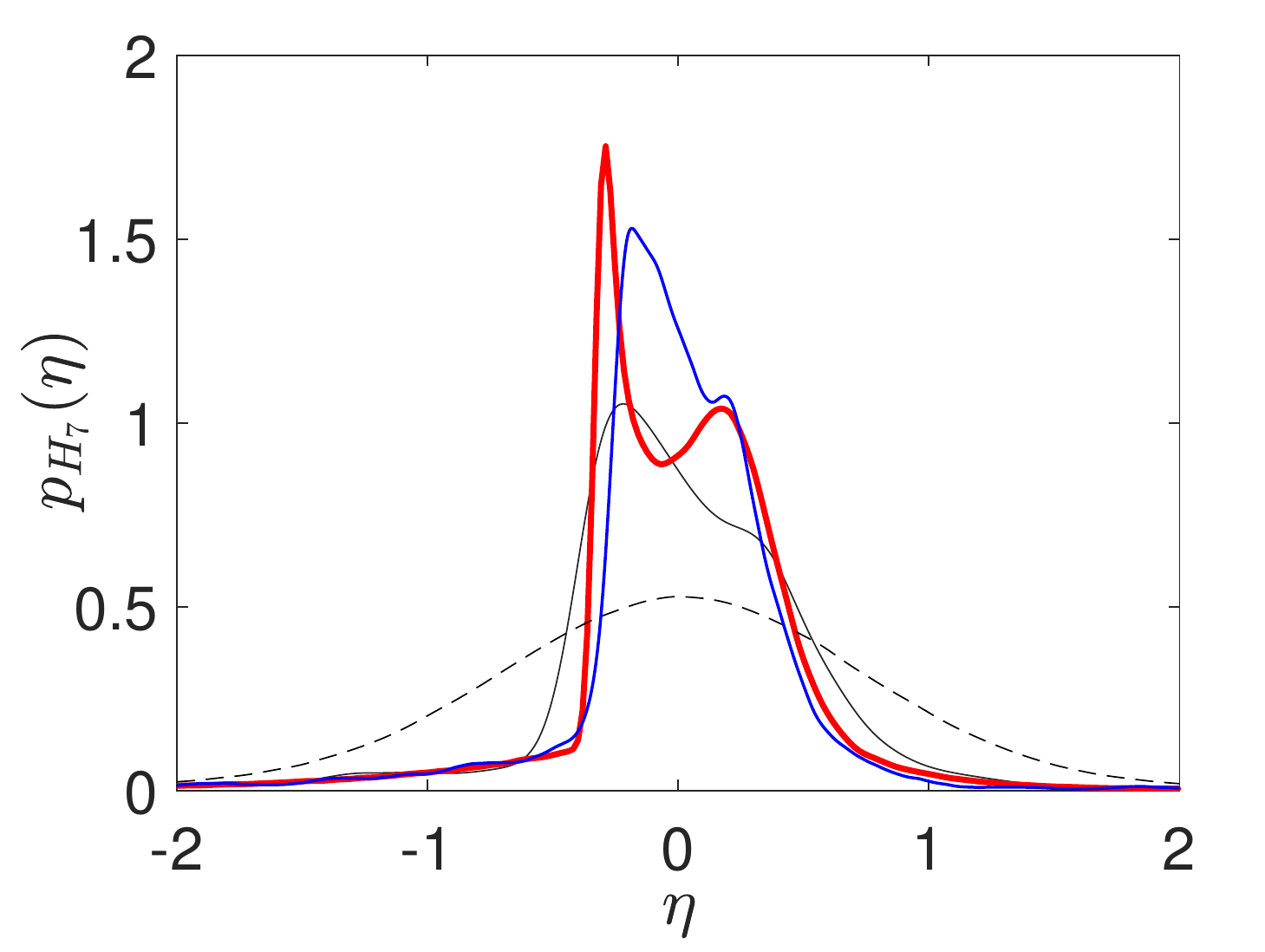}}
  \caption{pdf estimated with (i) the training set (black thin), (ii) the reference data set (red thick), (iii) No-PLoM (dashed) , and (iv) No-Group PLoM (blue thin).}
  \label{fig:figureAP1-2}
\end{figure}
\subsection{Computing the partition}
\label{sec:Ap1-2}
The optimal partition is computed as explained in Section~\ref{sec:Meth-4.1}.
Figure~\ref{fig:figureAP1-3a} displays the graph of $i_\pref\mapsto \tau(i_\pref)$, which shows that $i_\pref^{\,\opt}= 0.013$. Finally, the algorithm identifies the partition and finds $n_p = 3$ groups  with $\nu_1=10$, $\nu_2= 20$, and $\nu_3=30$ and with $\bfY^1= (H_1,\ldots, H_{10})$,
$\bfY^2= (H_{11},\ldots, H_{30})$, and $\bfY^3= (H_{31},\ldots, $ $H_{60})$, which correspond to the model introduced in Appendix~\ref{AppendixA} for generating the training set. This result constitutes an additional validation of the optimal partition algorithm that is used for non-Gaussian random vectors. For illustration, Figure~\ref{fig:figureAP1-3b} displays the graph of the joint pdf of random variables $H_1$ and $H_2$.
\begin{figure}[tbhp]
  \centering
  \subfloat[Graph of function $i_\pref\mapsto \tau(i_\pref)$]{\label{fig:figureAP1-3a} \includegraphics[width=5.5cm] {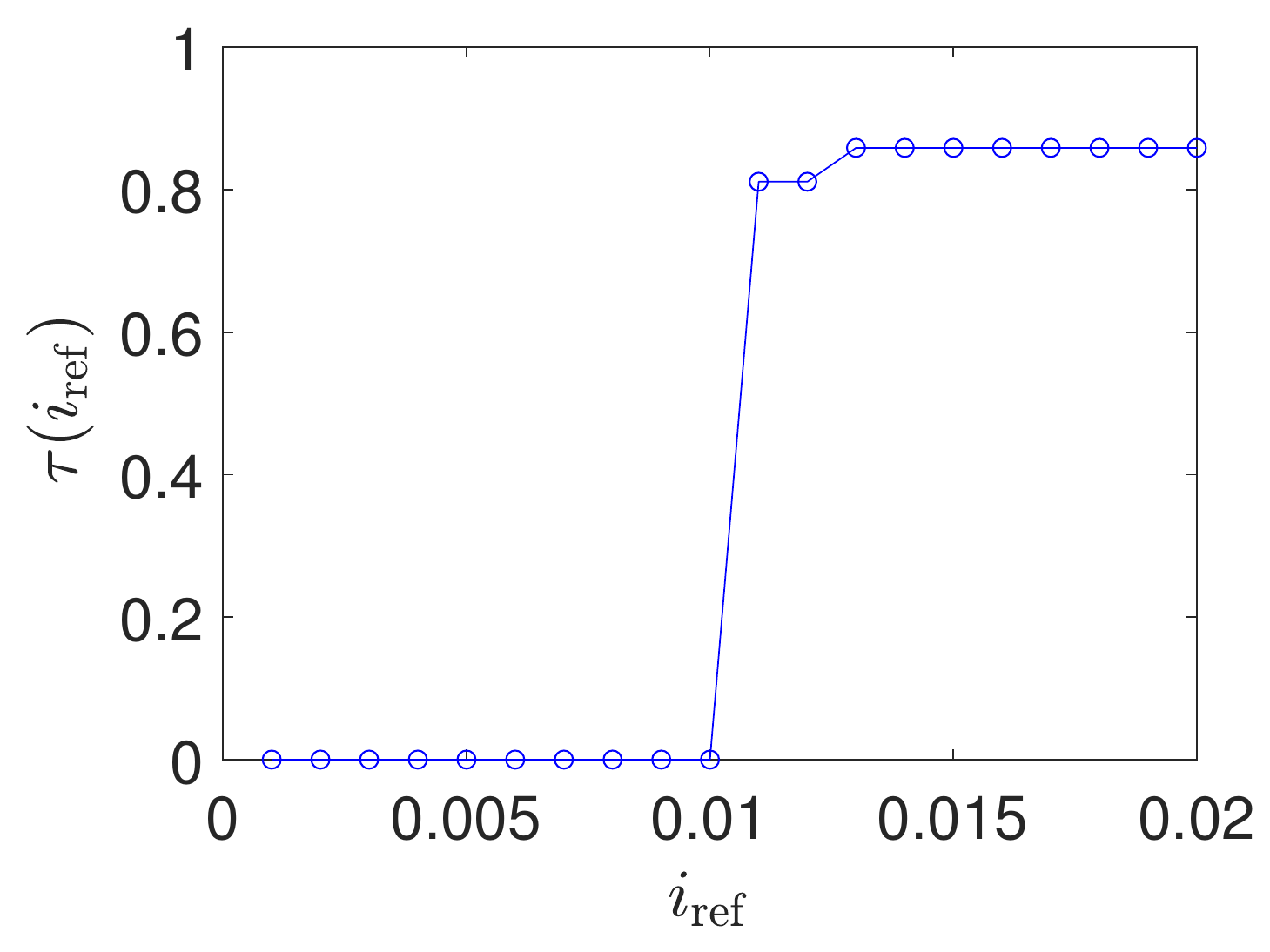}}\hfil
  \subfloat[Joint pdf of $H_1$ and $H_2$]   {\label{fig:figureAP1-3b} \includegraphics[width=4.8cm] {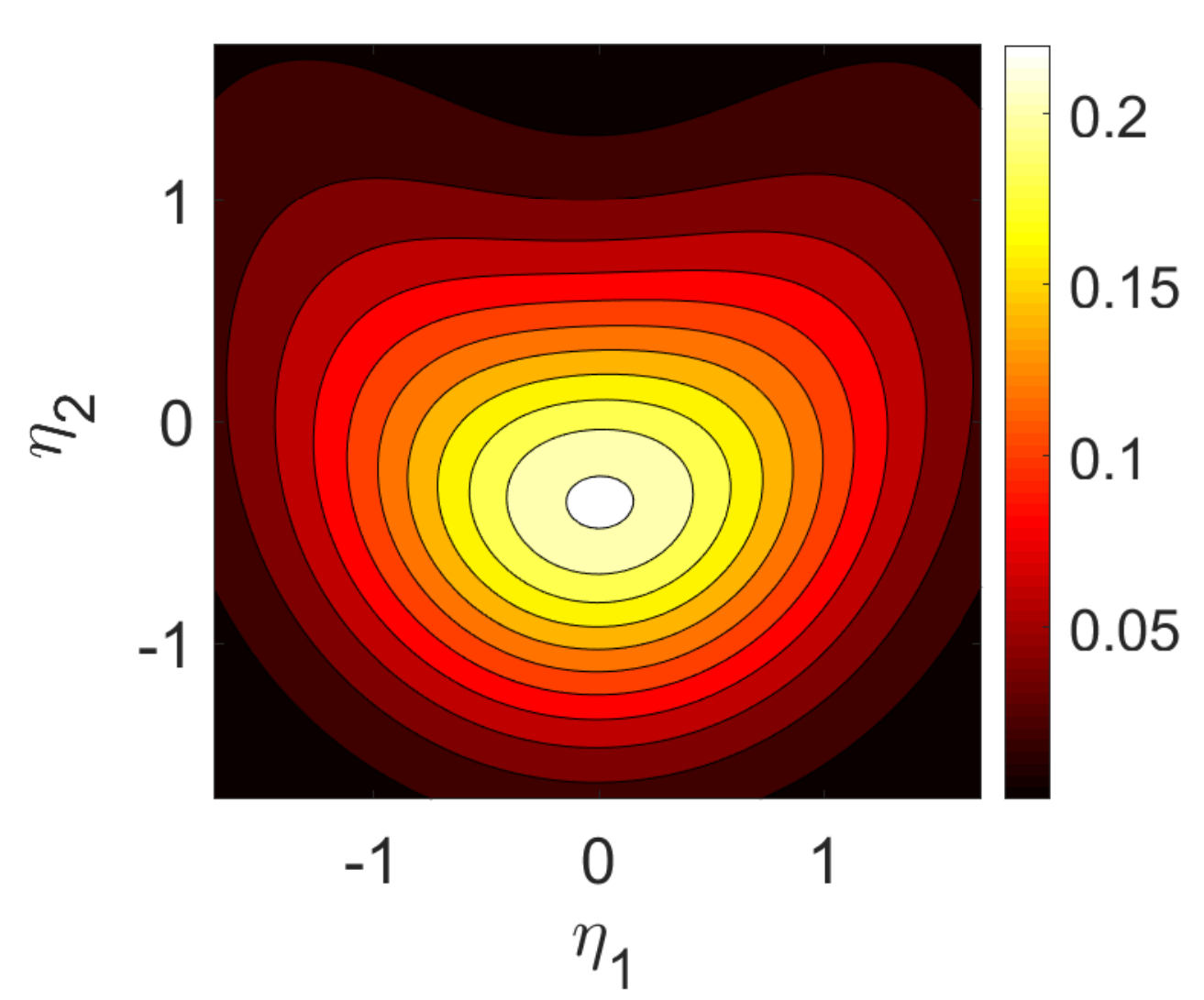}}
  \caption{Partition of $\bfH$ in $n_p$ mutually independent random vectors $\bfY^1,\ldots ,\bfY^{n_p}$.}
  \label{fig:figureAP1-3}
\end{figure}
\subsection{PLoM analysis with groups (With-Group PLoM)}
\label{sec:Ap1-3}
Algorithm~\ref{algorithm:Meth-1} is used for each group $i=1,2,3$. We then have $m_{i,\opt} = \nu_i+1$. The training set $\{\bfeta_d^{i,j},j=1,\ldots,N\}$ of $\bfY^i$ is used. A similar graph to the one shown in Figure~\ref{fig:figureAP1-1a} is constructed for identifying the optimal value $\varepsilon_{i,\opt}$ of the smoothing parameter $\varepsilon_{i,\DM}$ yielding $\varepsilon_{1,\opt} = 412$, $\varepsilon_{2,\opt} = 896$, and $\varepsilon_{3,\opt} = 1\, 132$.
Figure~\ref{fig:figureAP1-4a} shows the distribution of the eigenvalues of the transition matrix of each group $i$ computed for $\varepsilon_{i,\DM} =\varepsilon_{i,\opt}$.
It can be seen that all the required criteria are satisfied.
\begin{figure}[tbhp]
  \centering
  \subfloat[Eigenvalues of the transition matrix of group $i$ computed for $\varepsilon_{i,\DM} =\varepsilon_{i,\opt}$]
  {\label{fig:figureAP1-4a} \includegraphics[width=5.5cm]{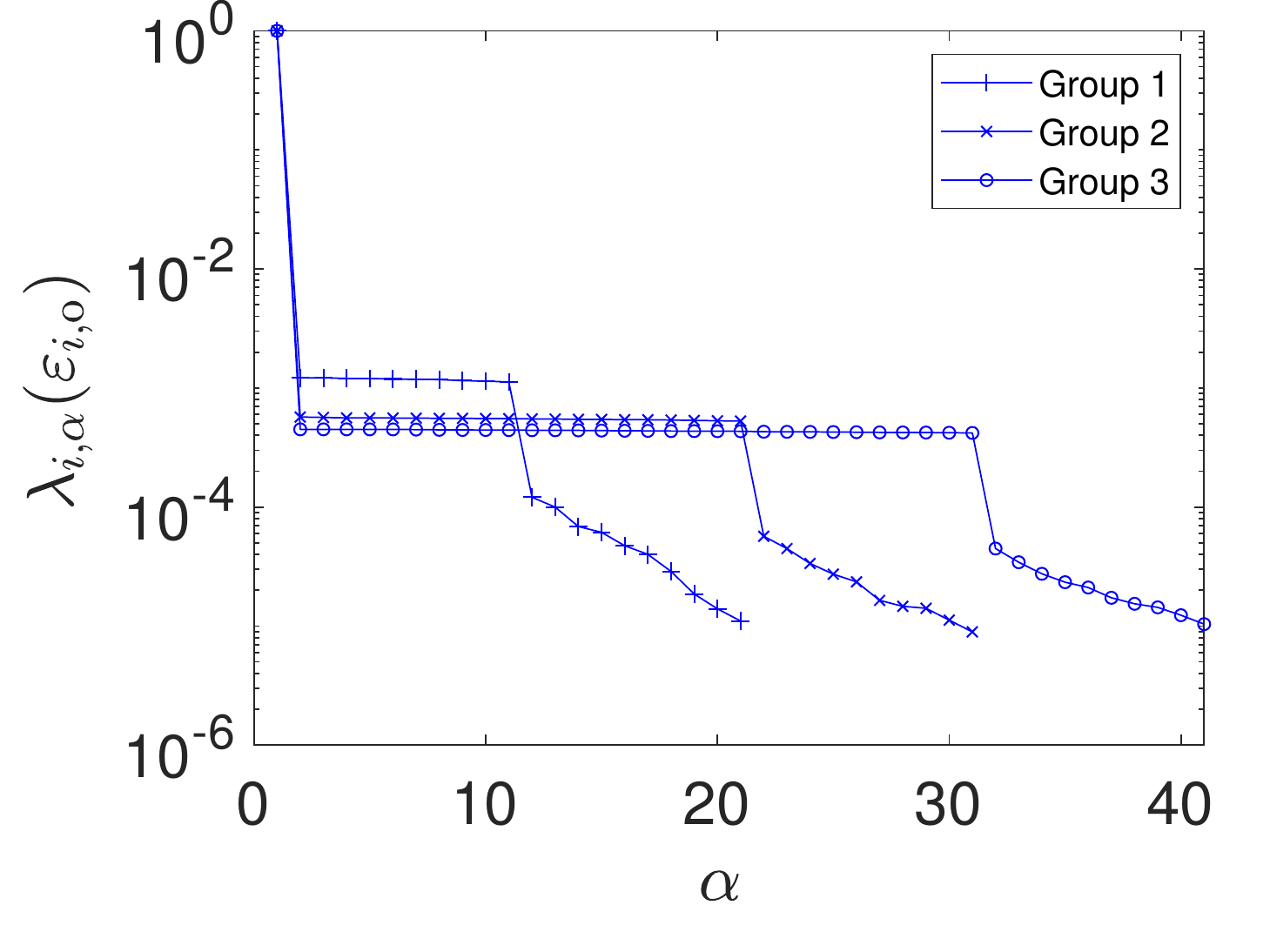}}
  \hfil
  \subfloat[Error function $\iota\mapsto \pperr_i(\iota)$ of group $i$ for iteration number $\iota$ of the iteration algorithm ]
  {\label{fig:figureAP1-4b} \includegraphics[width=5.5cm] {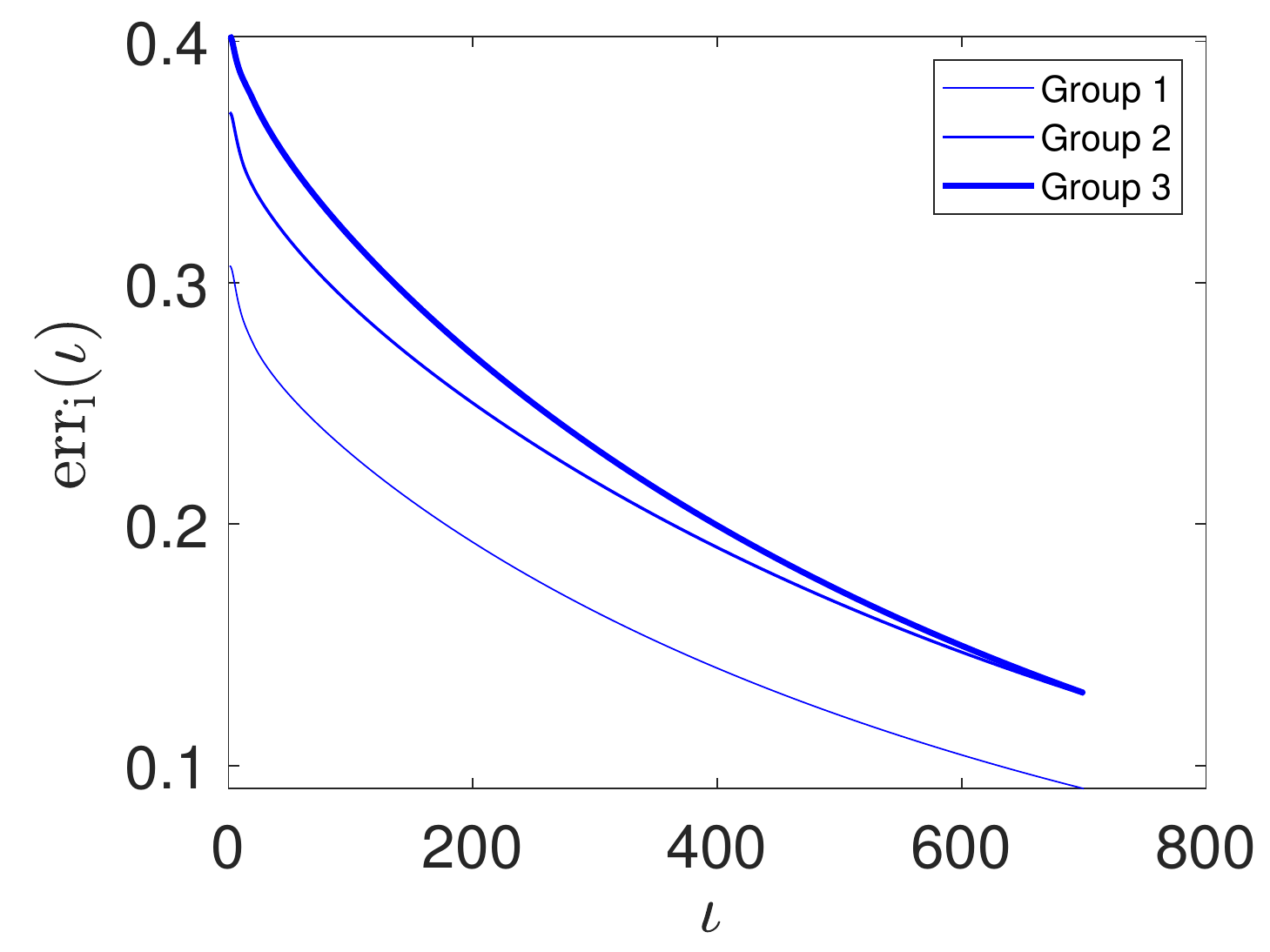}}
  \caption{Diffusion map basis and error function of the iteration algorithm for each one of the $3$ groups, $i=1,2,3$.}
  \label{fig:figureAP1-4}
\end{figure}
\begin{figure}[tbhp]
  \centering
  \subfloat[Mean value of $H_k$] {\label{fig:figureAP1-5a} \includegraphics[width=5.5cm]{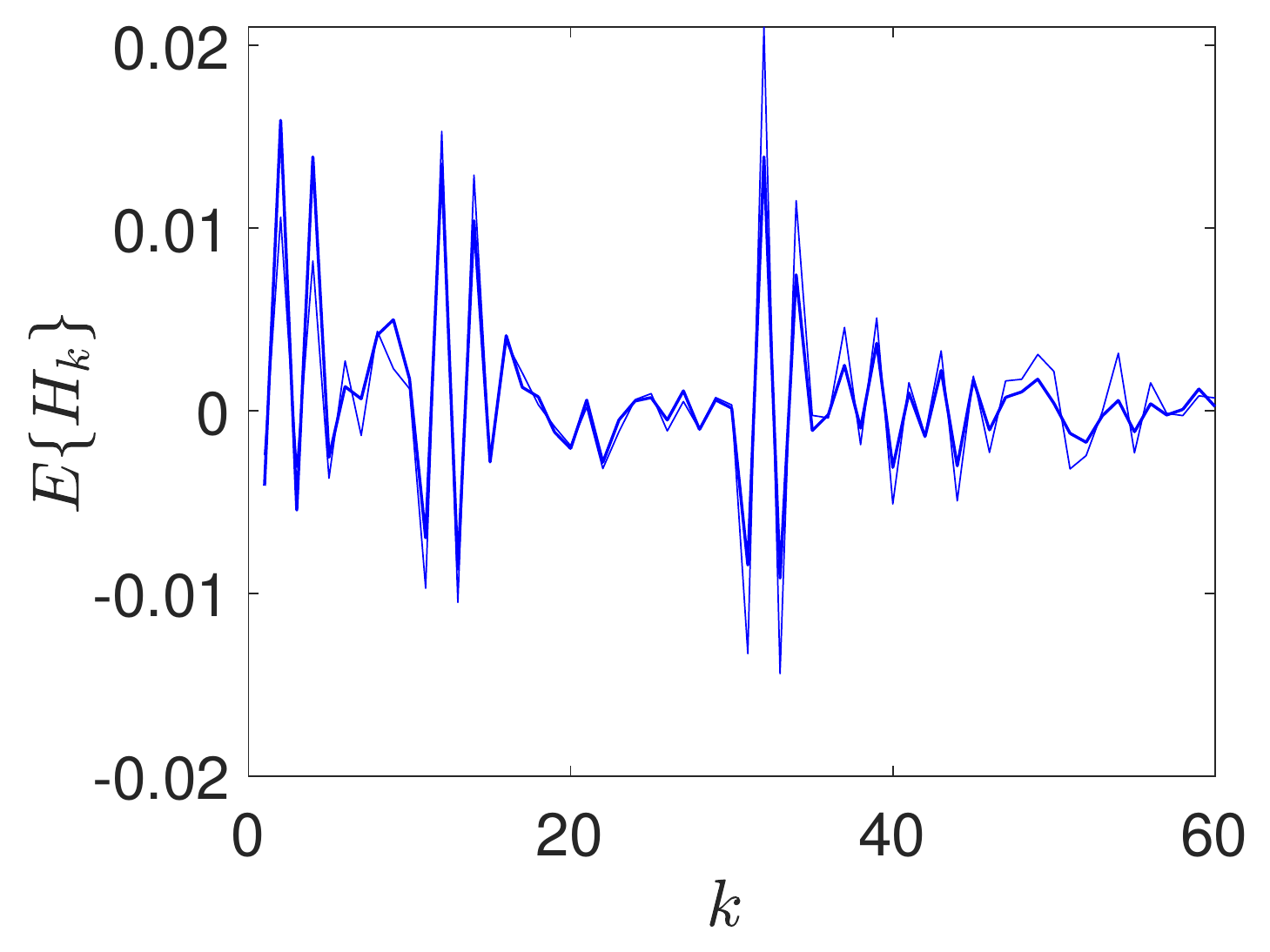}}
  \hfil
  \subfloat[Standard deviation of $H_k$]  {\label{fig:figureAP1-5b} \includegraphics[width=5.5cm] {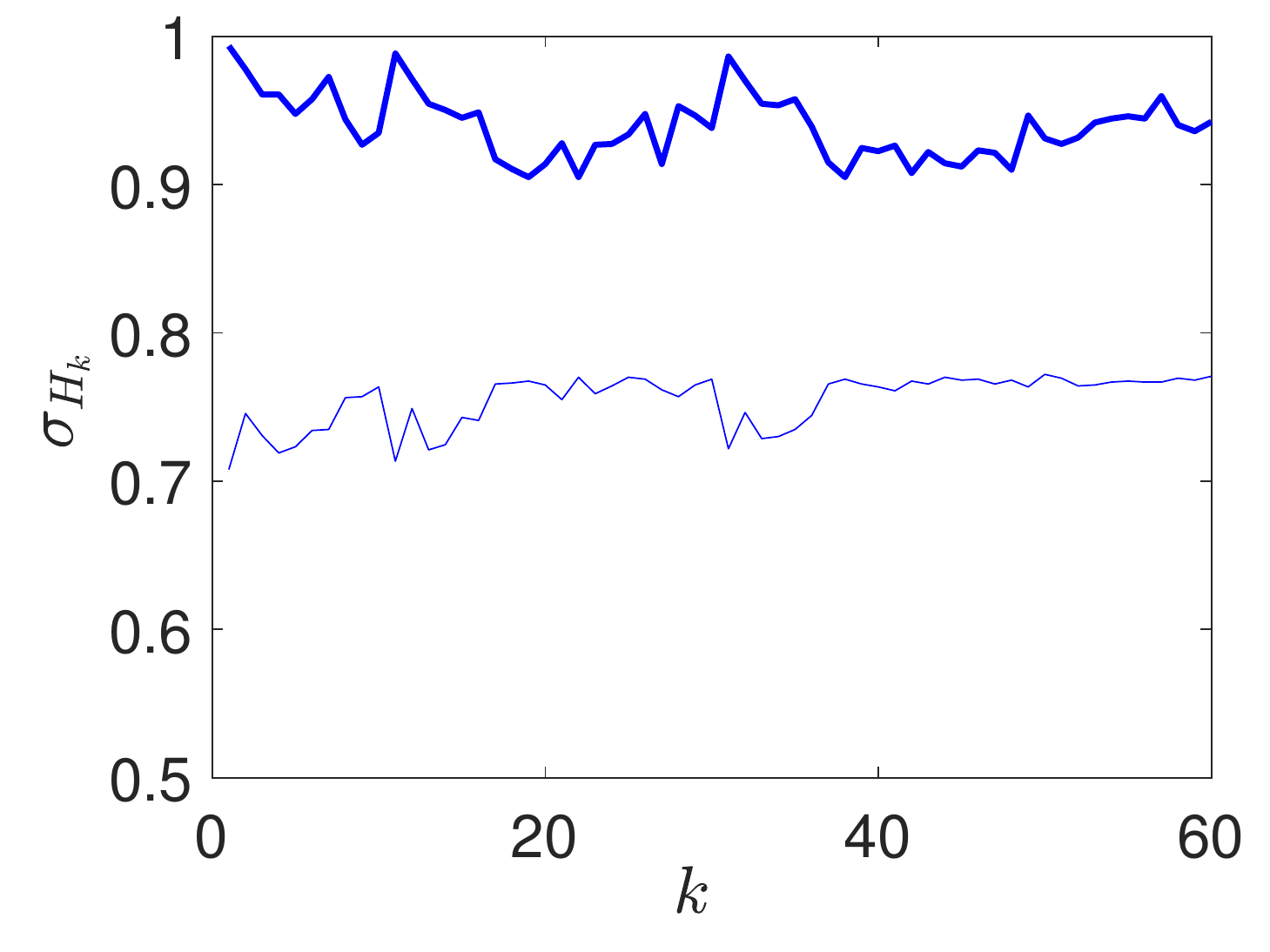}}
  \caption{Mean value and standard deviation of the components $H_k$, $k=1,\ldots, \nu$, of $\bfH$ estimated using the learned set generated
  by No-Group PLoM  (thin line) and by With-Group PLoM  (thick line).}
  \label{fig:figureAP1-5}
\end{figure}
For each group $i=1,2,3$, the PLoM method with groups is used for generating the learned set $\{\bfeta_\ar^{i,j},j=1,\ldots,N_\ar\}$.
The constraints $E\{(Y^i_k)^2\} = 1$ for $k\in\{1,\ldots , \nu_i\}$ are applied and the iterative algorithm introduced in Section~\ref{sec:Meth-4.4} is used. Figure~\ref{fig:figureAP1-4b} displays the error function $\iota\mapsto \perr_i(\iota)$ of group $i$, defined by Eq.~\eqref{eq:Meth-6}, which shows the convergence of the iterative algorithm.
It should be noted that the convergence could have been pushed
further, but the numerical experiments showed that the additional gain obtained is negligible. In addition, numerical experiments have been carried out to compare the efficiency of the type of constraints. We have verified that taking into account all the constraints (mean of
$\bfY^i$ equal to $\bfzero$ and covariance matrix $[C_{\bfY^i}] = [I_{\nu_i}]$) did not provide significant improvements on the preservation of the concentration of the probability measure compared to the sole application of the constraints $E\{(Y^i_k)^2\} = 1$ for $k\in\{1,\ldots , \nu_i\}$. Figure~\ref{fig:figureAP1-5} shows the mean value and the standard deviation of the components $H_k$, $k=1,\ldots, \nu$ of $\bfH$ estimated using the learned set generated
by No-Group PLoM (see Section~\ref{sec:Ap1-1}) and by With-Group PLoM. Figure~\ref{fig:figureAP1-5a} shows that the mean values are reasonably small with respect to $1$ and therefore that it is not necessary to improve it by introducing the constraints for the mean. Figure~\ref{fig:figureAP1-5b} shows that the standard deviations are improved by using With-Group PLoM for which the constraints are taken into account.
Figure~\ref{fig:figureAP1-6} shows the pdf of  $H_4$, $H_5$, $H_6$, and $H_7$ estimated with the learned set.  Similarly to Section~\ref{sec:Ap1-3}, each pdf is estimated (i) with the $N$ realizations of the training set, (ii) with the $N_\pref$ realizations of the reference data set, (iii) with the $N_\ar$ additional realizations computed by No-PLoM, and (iv) with the $N_\ar$ realizations computed by  With-Group PLoM.  Comparing Figure~\ref{fig:figureAP1-2} with Figure~\ref{fig:figureAP1-6}, it can be seen that the use of groups improves the pdfs' estimations as expected. It can also be noted that the estimates are excellent for this very difficult case in particular by comparing with the usual approach (see the dashed lines corresponding to No-PLoM).
\begin{figure}[tbhp]
  \centering
  \subfloat[pdf of $H_4$]{\label{fig:figureAP1-6a} \includegraphics[width=6.5cm] {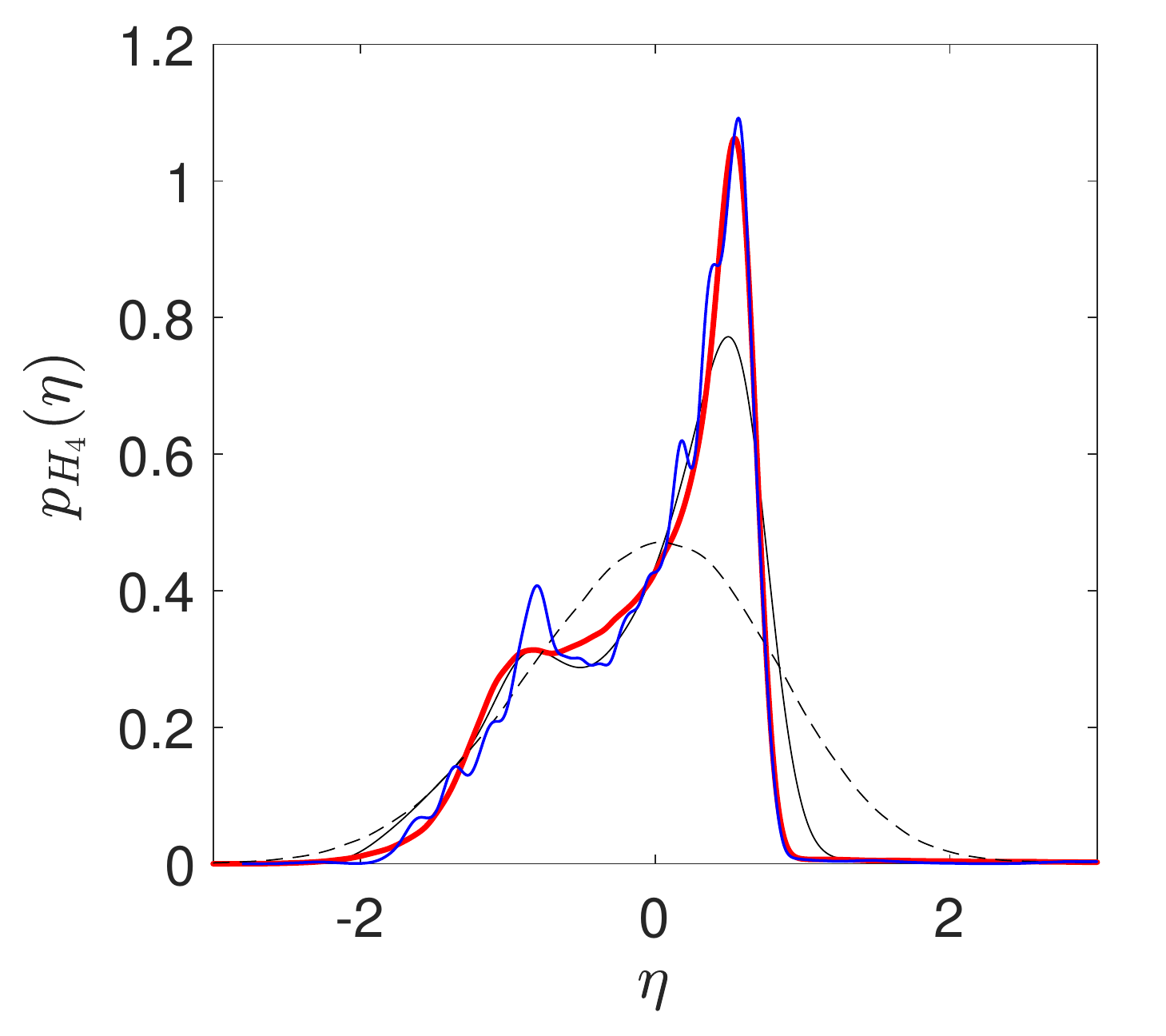}}
  \subfloat[pdf of $H_5$]{\label{fig:figureAP1-6b} \includegraphics[width=6.5cm] {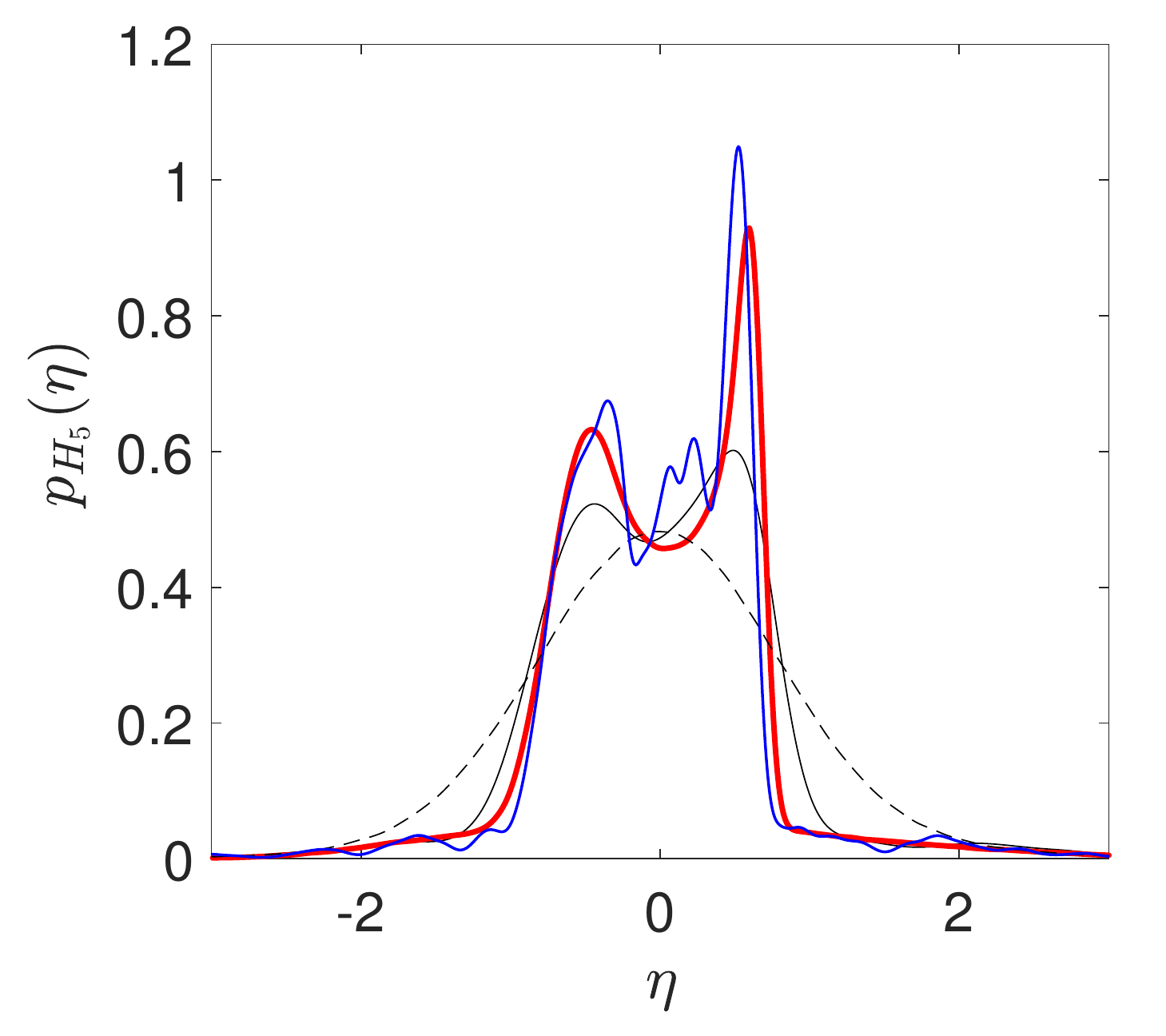}}\\
  \subfloat[pdf of $H_6$]{\label{fig:figureAP1-6c} \includegraphics[width=6.5cm] {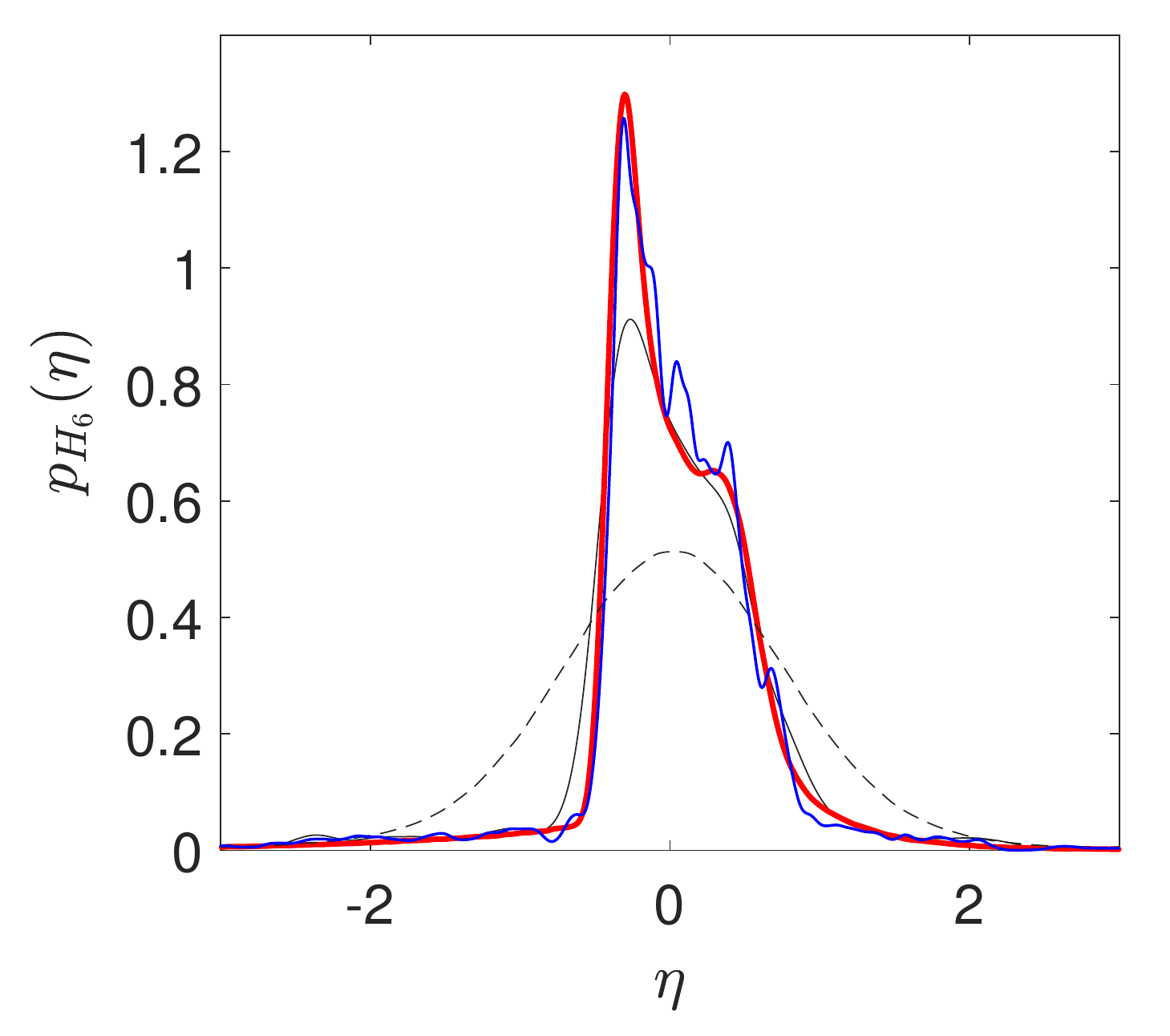}}
  \subfloat[pdf of $H_7$]{\label{fig:figureAP1-6d} \includegraphics[width=6.5cm] {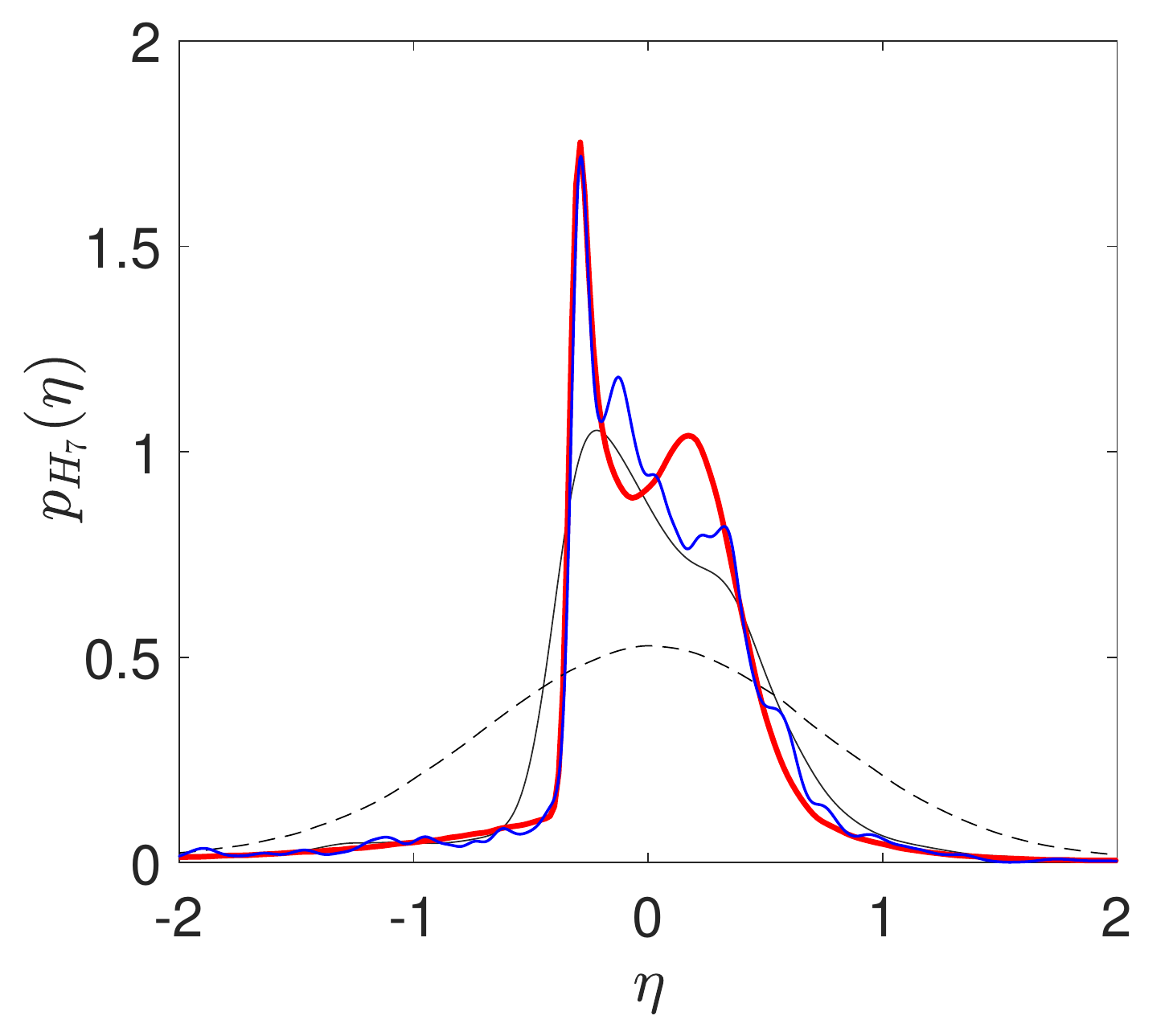}}
  \caption{pdf estimated with (i) the training set (black thin), (ii) the reference data set (red thick), (iii) No-PLoM (dashed) , and (iv) With-Group PLoM (blue thin).}
  \label{fig:figureAP1-6}
\end{figure}
\subsection{Quantifying the concentration of the probability measure}
\label{sec:Ap1-4}
For No PLoM, the computations are performed as explained in Section~\ref{sec:Ap1-1}, for No-Group PLoM as in Sections~\ref{sec:Meth-3} and \ref{sec:Ap1-1}, and for With-Group PLoM as in Sections~\ref{sec:Ap1-3} and \ref{sec:Meth-4}.

(i) The results concerning the  concentration of the probability measure are summarized in Table~\ref{tab:table1}. For No PLoM,  $d^2_N(N)$ is computed with Eq.~\eqref{eq:Meth-4} for which $m=N=1\, 200$. For No-Group PLoM, $d^2_N(m_\opt)$ is also computed with Eq.~\eqref{eq:Meth-4} but with $m=m_\opt = 61$. For With-Group PLoM, $d^2_{wg,N}(\bfm_\opt)$ is computed with Eq.~\eqref{eq:Meth-7} for which $\bfm_\opt = (m_{1,\opt}, m_{2,\opt},m_{3,\opt})$ with $m_{1,\opt} =11$, $m_{2,\opt}=21$, and $m_{3,\opt} =31$, and where $d^2_{i,N}(m_{i,\opt})$ are computed using Eq.~\eqref{eq:Meth-8}, which yields $d^2_{1,N}(m_{1,\opt})= 0.012$,
$d^2_{2,N}(m_{2,\opt})= 0.015$, and $d^2_{3,N}(m_{3,\opt})= 0.019$.
The results obtained are those that were hoped for. Without using the
PLoM method, we find numerically $d^2_N(N) \simeq 2 $ that is the
theoretical value (see Section~\ref{sec:Meth-3.5}). We also see that $d^2_N
(m_\opt) = 0.094 \ll 2$, which shows that the usual PLoM method
(without group) effectively preserves the concentration of the
probability measure unlike the usual MCMC method that does not allow
it. For the PLoM with groups, an improvement is observed relative to
the PLoM without group as indicated by the evaluation $d^2_{wg,N}(\bfm_\opt)= 0.016 \ll  d^2_N(m_\opt) = 0.094$. The quantification of the probability of the random relative distance defined by Eq.~\eqref{eq:Meth-12} confirms this improvement. Note that the probability $(r /\varepsilon)^{n_p}$ corresponds to an upper value, the probability being certainly smaller.
\begin{table}[tbhp]
  \caption{Concentration of the probability measure for Application~1}\label{tab:table1}
\begin{center}
 \begin{tabular}{|c|c|c|c|c|} \hline
                  \textbf{No}  & \textbf{PLoM}      &      \multicolumn{3}{c|}{\textbf{PLoM}}                 \\
                  \textbf{PLoM}& \textbf{No Group}  &      \multicolumn{3}{c|}{\textbf{With Group}}          \\ \hline
                               & $m_\opt =61$    &                        & \multicolumn{2}{c|}{$\rm{Proba}$ by Eq.~\eqref{eq:Meth-12}} \\
                    $d^2_N(N)$ & $d^2_N(m_\opt)$ & $d^2_{wg,N}(\bfm_\opt)$ &  $\varepsilon$ &  $\leq (\frac{r}{\varepsilon})^{n_p}$   \\ \hline
                    $2.00$     &   $0.094$       &     $0.016$             &   $0.05$           &  $\leq 0.028$                       \\
                               &                 &                         &   $0.10$           &  $\leq 0.0034$                       \\ \hline
  \end{tabular}
\end{center}
\end{table}

(ii) Concerning the visualization of the concentration of the probability measure, Figure~\ref{fig:figureAP1-7} shows the clouds of points $\{\eta_\ar^\ell, \ell = 1, \ldots, N_\ar\} $ for the components $(H_1, H_2, H_3)$, generated (a) without using PLoM (No-PLoM), (b) using the PLoM method without group (No-Group PLoM), and (c) using the PLoM method with groups (With-Group PLoM). The three figures confirm the analysis presented in point (i) above.
\begin{figure}[tbhp]
  \centering
  \subfloat[No PLoM]{\label{fig:figureAP1-7a} \includegraphics[width=4.2cm] {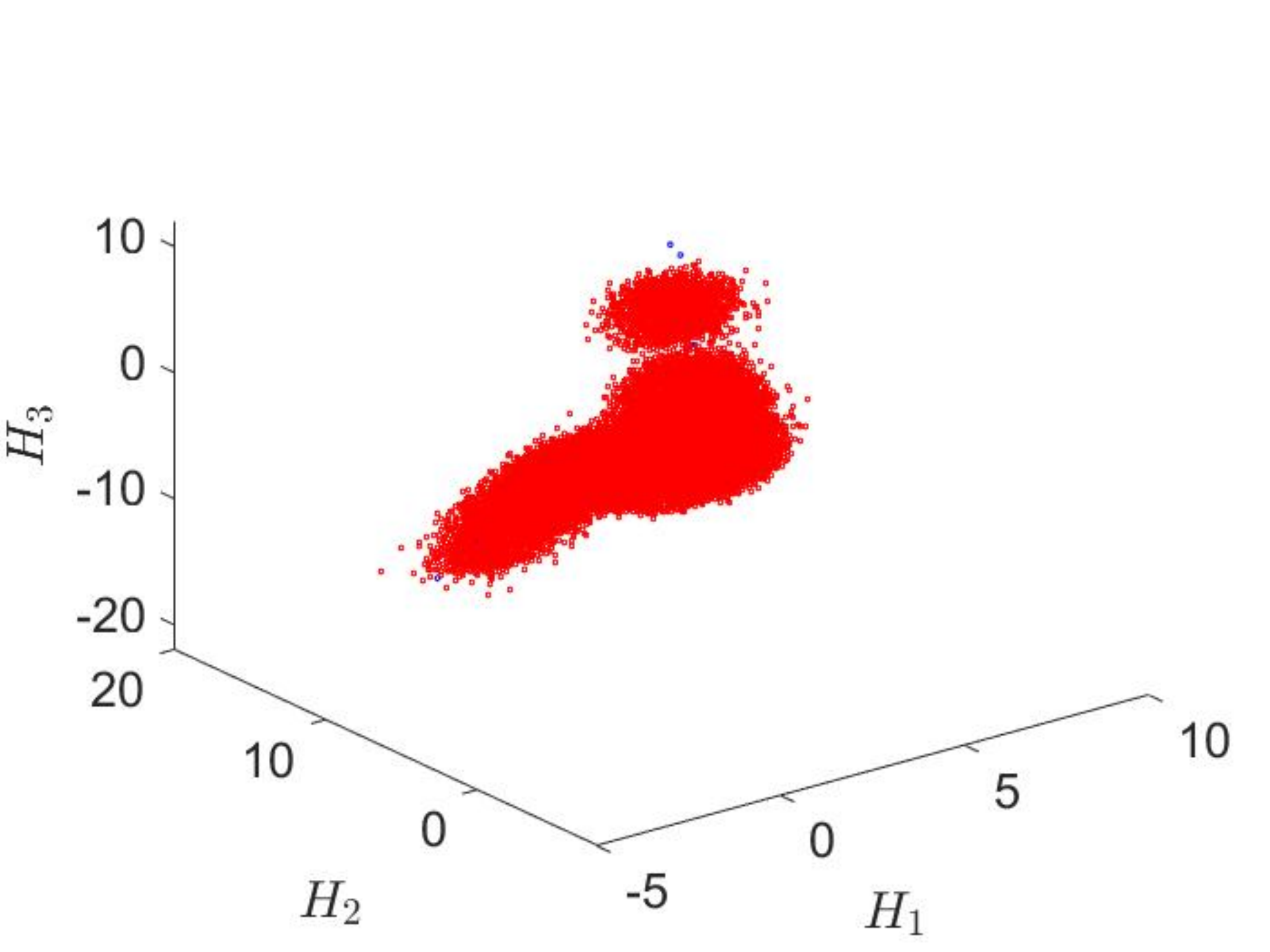}}
  \subfloat[No-Group PLoM]{\label{fig:figureAP1-7b} \includegraphics[width=4.2cm] {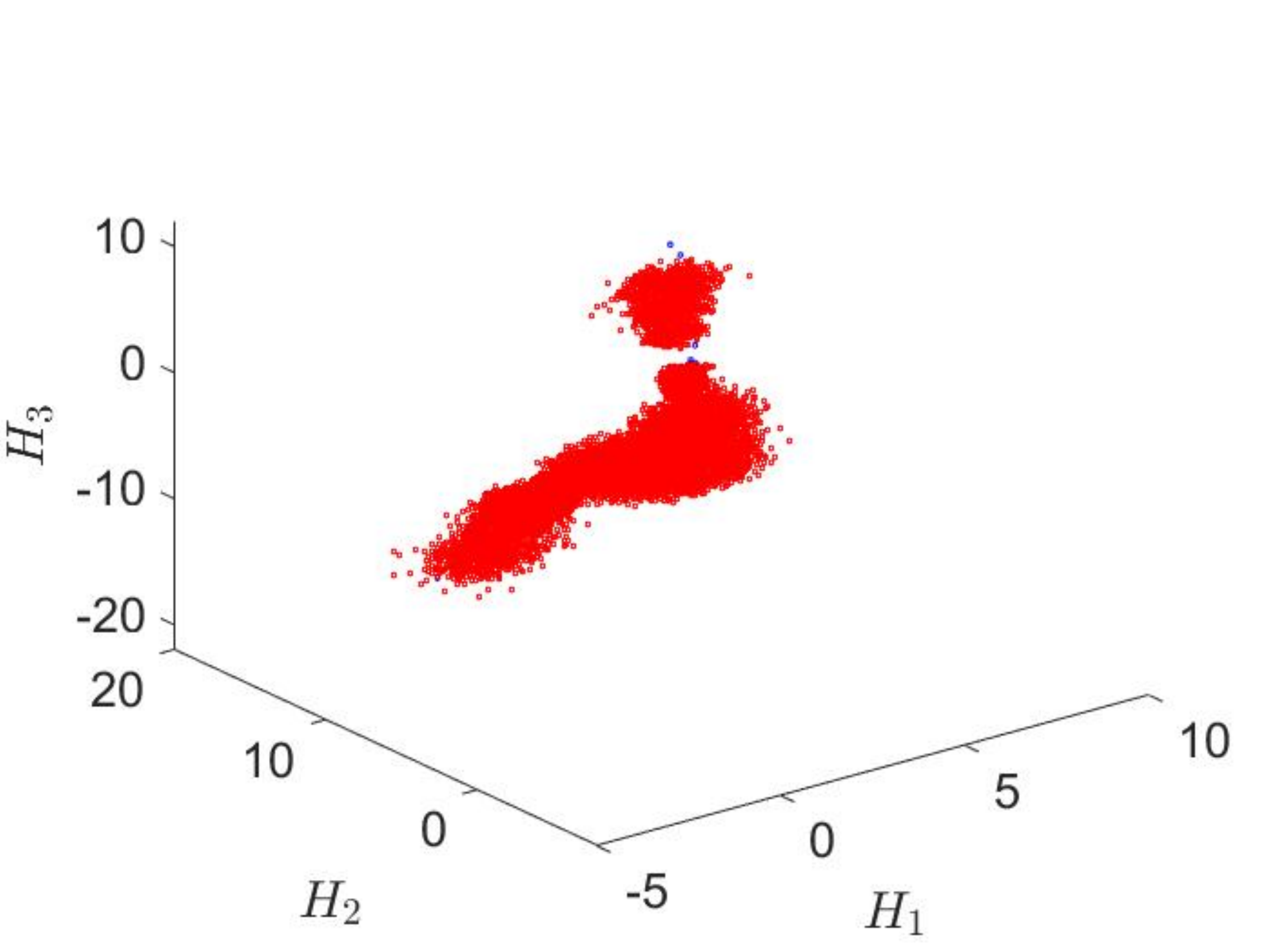}}
  \subfloat[With-Group PLoM]{\label{fig:figureAP1-7c} \includegraphics[width=4.2cm] {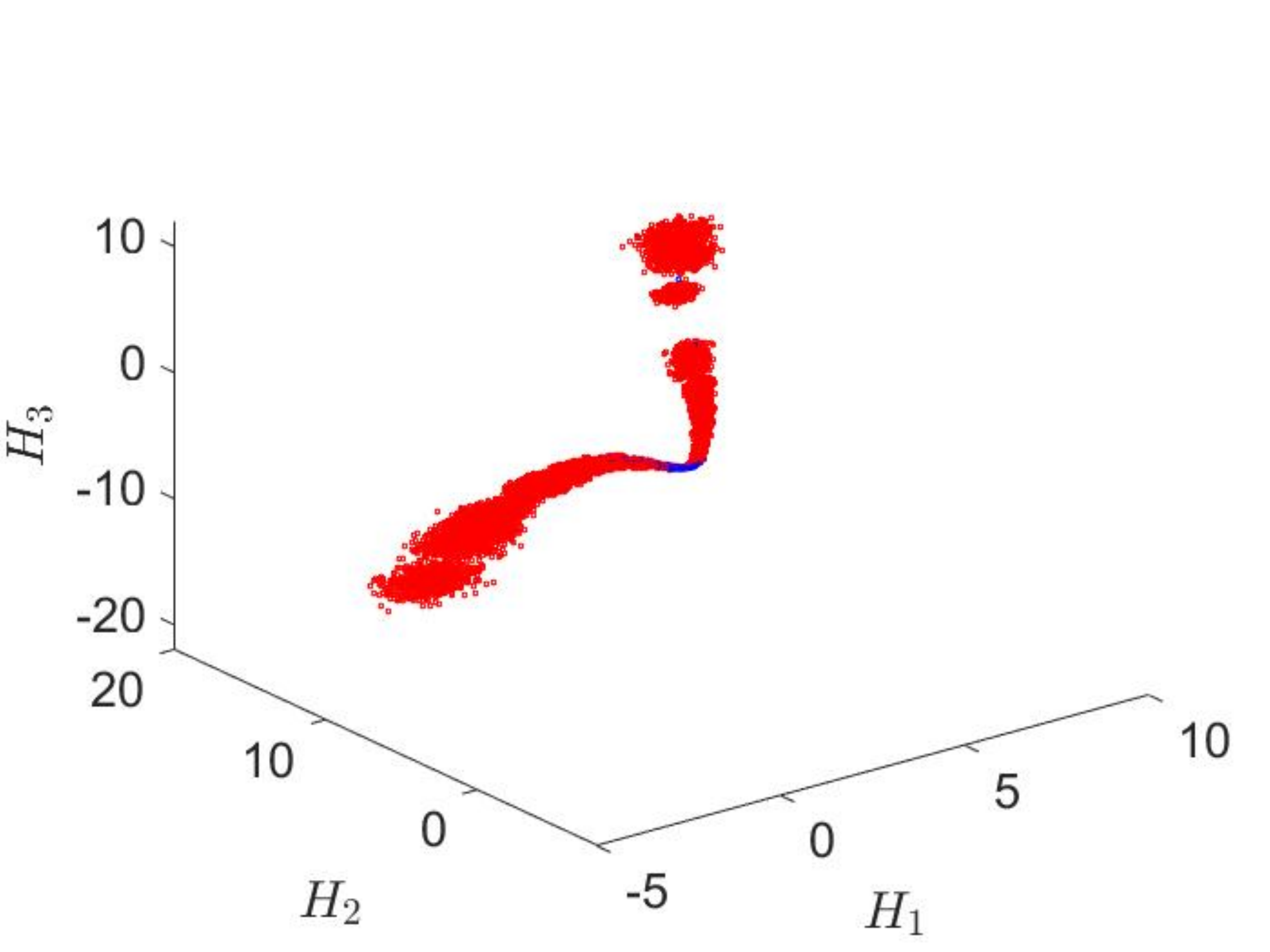}}
  \caption{Clouds of the $N_\ar = 1\, 200\, 000$ realizations of $(H_1,H_2,H_3)$ computed with No-PLoM, No-Group PLoM, and With-Group PLoM.}
  \label{fig:figureAP1-7}
\end{figure}

\section{Application~2}
\label{sec:Ap2}
The second application is devoted to a supervised learning problem $\bfQ=\bff(\bfW,\bfU)$ (see Section~\ref{sec:Meth-1}) in high dimension, for which the uncontrolled random parameter is the $\RR^{n_u}$-valued random variable $\bfU$ with $n_u = 420\, 000$, the random control random parameter is the $\RR^{n_w}$-valued random variable $\bfW$ with $n_w = 2$, and the QoI is the $\RR^{n_q}$-valued random variable $\bfQ$ with $n_q = 10\, 098$.
\subsection{Generation of the training set and reference data set}
\label{sec:Ap2-1}
This application relates to a linear elastic system modeled by an elliptic stochastic boundary problem (BVP) in a 3D bounded domain $\Omega$, described in the SI Units. The generic point of $\Omega$ is $\bfzeta = (\zeta_1,\zeta_2,\zeta_3)$ in an orthonormal Cartesian coordinate system $(O,\zeta_1,\zeta_2,\zeta_3)$ with $O =(0,0,0)$. The outward unit normal to $\partial\Omega = \Gamma_0 \cup \Gamma$ is denoted by $\bfn(\bfzeta)$. There is a zero Dirichlet condition on $\Gamma_0$ and a Neumann condition on $\Gamma$. Domain $\Omega$ is occupied by a random linear elastic medium (heterogeneous material). The uncontrolled parameter of the system is the fourth-order tensor-valued non-Gaussian elasticity random field $\{\KK = \{\KK_{ijkh}(\bfzeta)\}_{1\leq i,j,k,h \leq 3}, \bfzeta\in\Omega\}$ (random coefficients of the partial differential operator) for which the mean value is isotropic and the statistical fluctuations are anisotropic. The control parameter $\bfw=(w_1,w_2)$ of the system consists of $w_1 = \log(L_\pcorr)$ in which
$L_\pcorr$ is a spatial correlation length and of $w_2 = \log(\delta_G)$ in which $\delta_{G}$ is a dispersion parameter, which allow  the statistical fluctuations of $\KK$ to be controlled. The observation of the system  is the $\RR^3$-valued random displacement field $\bfV =(V_1,V_2,V_3)$ on $\Omega$, which is the random solution of the stochastic BVP,
\begin{align}
- \divergence \, {\bfSigma}  &=  \bfzero      \quad \hbox{in}   \quad \Omega   \, ,               \nonumber \\
\bfV &  = \bfzero \quad \hbox{on} \quad \Gamma_0\, ,                                               \nonumber \\
\bfSigma \, \bfn & =     \bfcurG^\Gamma  \quad \hbox{on}  \quad \Gamma \, .                       \nonumber
\end{align}
The stress tensor $\bfSigma = \{\Sigma_{ij}\}_{1\leq i,j\leq 3}$ is related to the strain tensor $\bfE(\bfV) = \{E_{kh}(\bfV)\}_{1\leq k,h \leq 3}$ by the
constitutive equation, $\Sigma_{ij}(\bfzeta)= \KK_{ijkh}(\bfzeta) \, E_{kh}(\bfV(\bfzeta))$
in which the strain tensor is such that $E_{kh}(\bfV) = (\partial V_k/\partial\zeta_h + \partial V_h/\partial\zeta_k)/2$.
The geometry, the surface force field $\bfcurG^\Gamma$, the probabilistic model of the elasticity random field $\KK$ that depends on parameter $\bfw$, and the finite element discretization of the weak formulation of the stochastic BVP are detailed in Appendix~\ref{AppendixB}.

The control parameter $\bfw$ is modelled by a $\RR^2$-valued random variable $\bfW=(W_1,W_2)$.
The random vectors $\bfU$, $\bfW$, and $\bfQ$, for which the dimensions are $n_u = 420\, 000$, $n_w = 2$, and $n_q = 10\, 098$, are defined in Appendix~\ref{AppendixB}. The random vectors $\bfW$ and $\bfU$ are  statistically independent.
The dimension $n = n_q+n_w$ of random vector $\bfX = (\bfQ,\bfW)$ is thus $n= 10\, 000$.

The training set is generated as explained in Section~\ref{sec:Meth-1} for which $N= 100$ independent realizations, $\bfu_d^j$ and $\bfw_d^j$, of $\bfU$ and $\bfW$ are generated using the probabilistic model detailed in Appendix~\ref{AppendixB}. For each $j \in\{1,\ldots, N\}$, the realization $\bfq_d^j$ of $\bfQ$ is computed by solving  the BVP using the computational model (finite element discretization of the BVP), which is such that $\bfq_d^j = \bff(\bfw_d^j,\bfu_d^j)$ (note that $\bff$ is not explicitly known and results from the solution of the BVP). The training set related to random vector  $\bfX =(\bfQ,\bfW)$ is then made up of the $N$ independent realizations $\{\bfx_d^{j} , j=1,\ldots , N\}$ in which $\bfx_d^j = (\bfq_d^j,\bfw_d^j) \in \RR^n$.

The reference data set $\{\bfx_\pref^{\ell} , \ell=1,\ldots , N_\pref\}$ for $\bfX$ is generated as the training set but with $N_\pref = 20\, 000$ independent realizations. Computations have been made for $N_\pref = 15\,000$, $18\,000$, and $20\,000$, which have shown that the pdf of each observed component of $\bfQ$ were converged for $N_\pref = 20\,000$ (note that the construction of the reference has been very CPU time consuming).

The learned sets generated without using the PLoM method (No PLoM), or using the PLoM method with no group (No-Group PLoM), or with groups (With-Group PLoM)
will be all performed with $N_\ar = 200\, 000$ realizations
$\{\bfeta_\ar^{\ell},\ell=1,\ldots,N_\ar\}$ ($N_\ar = n_\MC\times N$ with $n_\MC = 2\, 000$).

\subsection{PLoM analysis without and with partition}
\label{sec:Ap2-2}
In this section, we give the main results without too many details (paper length limitation), knowing that we have already presented a detailed analysis for Application~1.

\paragraph{(i) PCA of random vector $\bfX$}
Since $n=10\,000 \ll N = 100$, the eigenvalue problem of $[C_\bfX]$ is solved using a thin SVD of matrix $[x_d] = [\bfx_d^1 \ldots \bfx_d^N]\in \MM_{n,N}$, which thus does not require the assembling of $[C_\bfX]$ (as explained in Section~\ref{sec:Meth-2}). Figure~\ref{fig:figureAP2-1a} shows the distribution of the eigenvalues $\mu_\alpha$.
For constructing the PCA representation, $\bfX^\nu = \underline{\bfx}_d +[\Phi]\,[\mu]^{1/2}\, \bfH$, of $\bfX$, we have chosen $\varepsilon_\PCA = 0.001$ that yields $\nu = 27$. Following Section~\ref{sec:Meth-2}, the matrix $[\eta_d] \in \MM_{\nu,N}$ is constructed with the $\nu$ realizations $\bfeta^1_d,\ldots , \bfeta^N_d$ of the $\RR^\nu$-valued random variable $\bfH$.

\paragraph{(ii) Reduced-order diffusion-map basis for No-Group PLoM}
Algorithm~\ref{algorithm:Meth-1} is used for the calculation of the reduced-order diffusion-map basis $[g_{m_\ppopt}]$ of the $\RR^\nu$-valued random variable $\bfH$. For the optimal values $m_\opt = \nu+1 = 28$ and
$\varepsilon_\opt =103$, Figure~\ref{fig:figureAP2-1b} shows the eigenvalues $\lambda_\alpha$ of the transition matrix.
\begin{figure}[tbhp]
  \centering
  \subfloat[Eigenvalues of the covariance matrix of $\bfX$]{\label{fig:figureAP2-1a}\includegraphics[width=5.5cm]{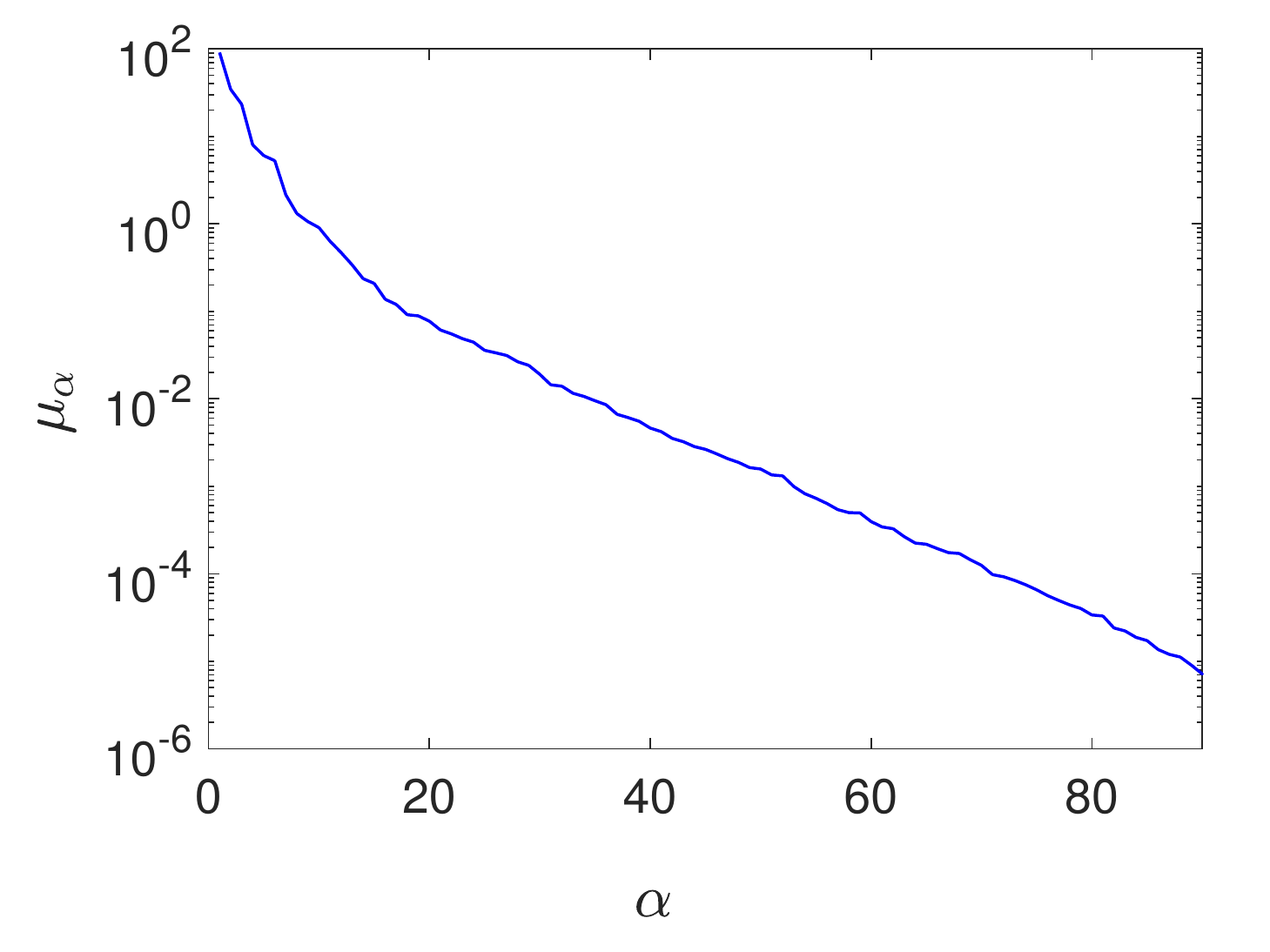}}
  \hfil
  \subfloat[Eigenvalues of the transition matrix] {\label{fig:figureAP2-1b} \includegraphics[width=5.5cm] {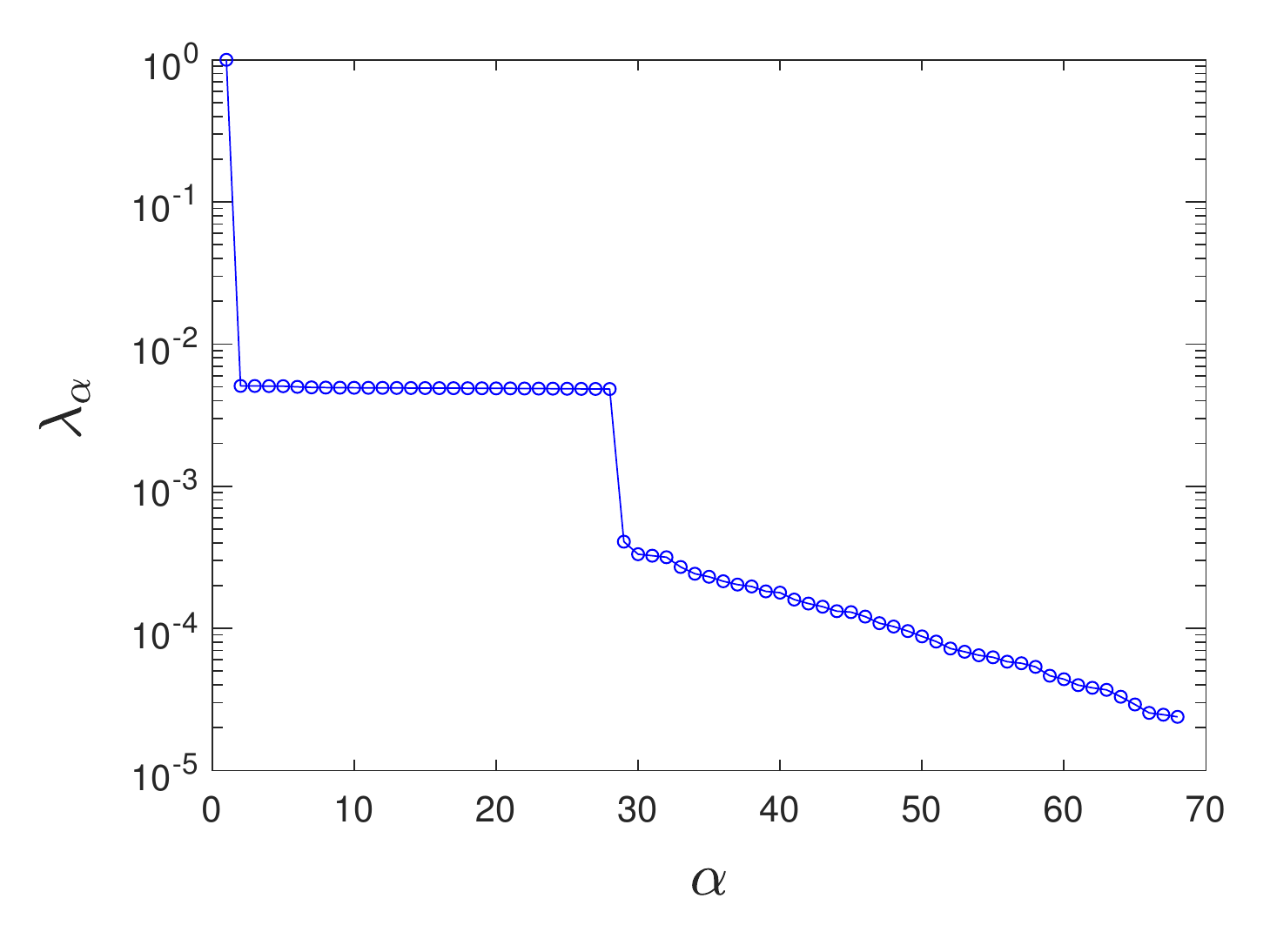}}
  \caption{PCA eigenvalues and diffusion map eigenvalues of the PLoM without group (No-Group PLoM).}
  \label{fig:figureAP2-1}
\end{figure}
\paragraph{(iii) Construction of the optimal partition of $\bfH$}
The optimal partition is computed as explained in Section~\ref{sec:Meth-4.1}.
Figure~\ref{fig:figureAP2-2a} displays the graph of $i_\pref\mapsto \tau(i_\pref)$, which shows that $i_\pref^{\,\opt}= 0.112$. The algorithm identifies the partition and finds $n_p = 9$ groups such that
$\nu_1=5$ with $\bfY^1 =(H_1,H_2,H_4,H_{16},H_{19})$,
$\nu_2= 1$ with $Y^2 =H_3$,
$\nu_3= 15$ with $\bfY^3 =(H_5\, \hbox{to}\, H_{11}, H_{14}, H_{15}, H_{17}, H_{18},H_{20}, H_{24} \,\hbox{to} \,H_{26}$),
and $\nu_4 = \ldots = \nu_9 = 1$ with $Y^4 =H_{12}$, $Y^5 =H_{13}$, $Y^6 =H_{21}$, $Y^7 =H_{22}$, $Y^8 =H_{23}$,  and $Y^9 =H_{27}$.
For each one of the two groups $i=1$ and $3$ (having a length greater than $1$), the optimal values are $m_{1,\opt} = 6$, $m_{3,\opt} = 16$, $\varepsilon_{1,\opt}= 37.7$, and $\varepsilon_{3,\opt}= 103$. For these optimal values of $\varepsilon_{i,\opt}$, Figure~\ref{fig:figureAP2-2b} shows the distribution of eigenvalues $\lambda_{i,\alpha}$  of the transition matrix. For the groups $i\neq 1$ and $3$, we have $m_{i,\opt}=N$ (see Algorithm~\ref{algorithm:Meth-1}).
\begin{figure} [tbhp]
  \centering
  \subfloat[Graph of function $i_\pref\mapsto \tau(i_\pref)$]{\label{fig:figureAP2-2a} \includegraphics[width=5.5cm] {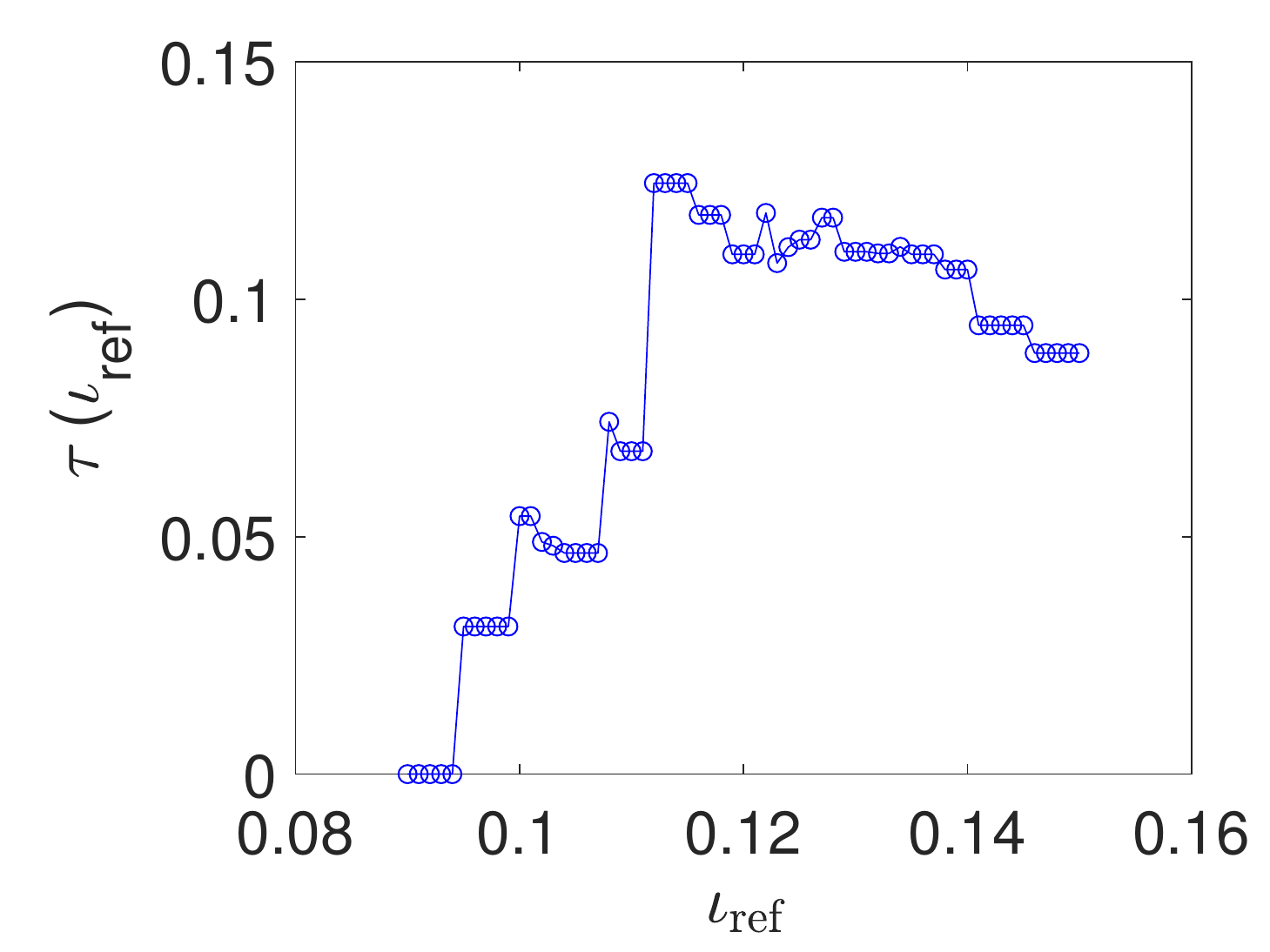}}
  \hfil
  \subfloat[Eigenvalues of the transition matrix]   {\label{fig:figureAP2-2b} \includegraphics[width=5.3cm] {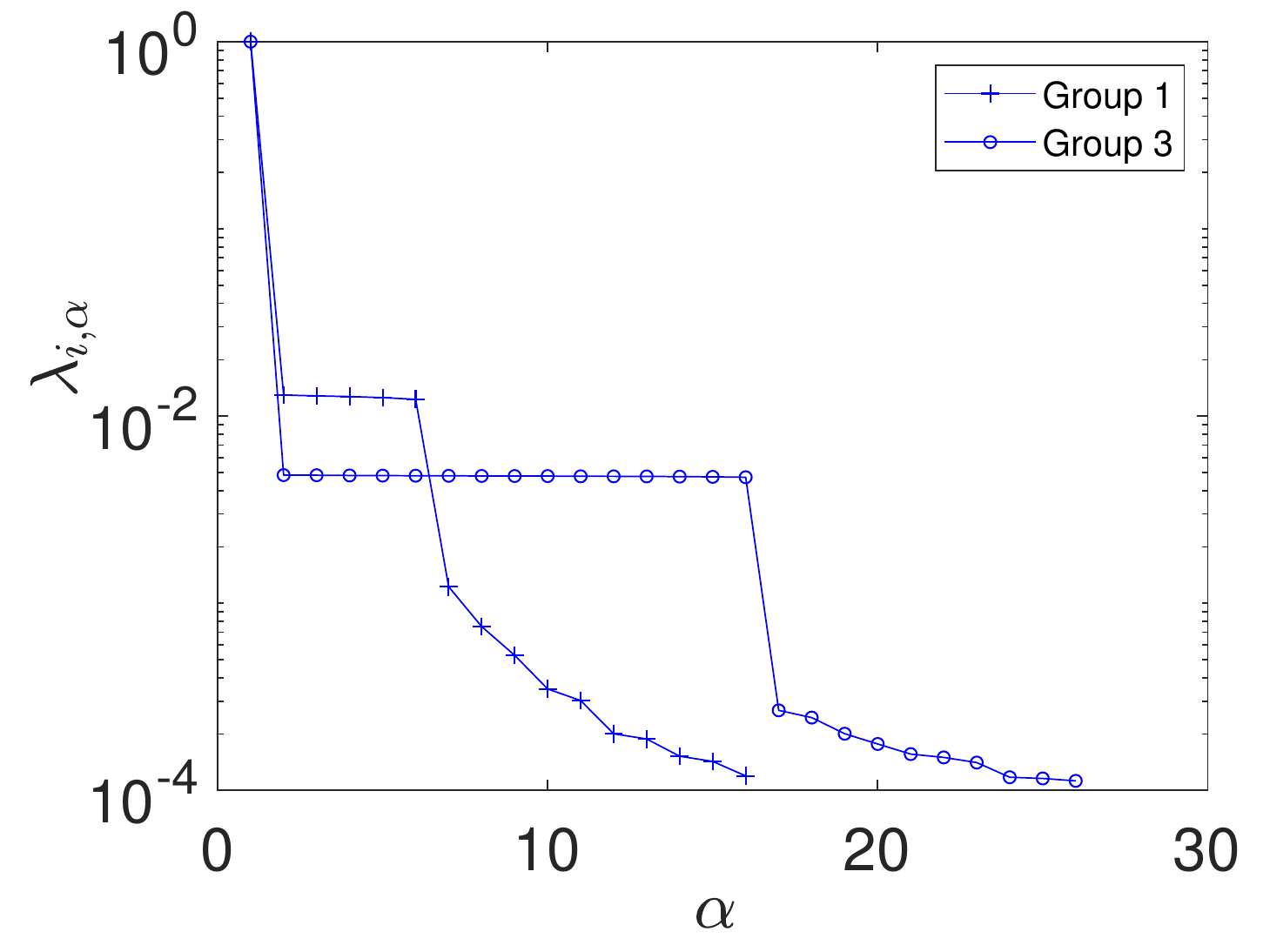}}
  \caption{Partition of $\bfH$ in $n_p=9$ mutually independent random vectors $\bfY^1,\ldots ,\bfY^{n_p}$ and diffusion map eigenvalues of the PLoM for groups $1$ and $3$.}
  \label{fig:figureAP2-2}
\end{figure}
\paragraph{(iv) Influence of the constraints of all the components of $\bfY^i$}
No-Group PLoM is performed without any constraints applied to random vector $\bfH$. With-Group PLoM is performed, group by group, in applying, for $i=1,\ldots,n_p=9$, the constraints $E\{(Y^i_k)^2\} =1$ for $k\in\{1,\ldots , \nu_i\}$. For all the components $k=1,\ldots, \nu=27$ of $\bfH$, Figure~\ref{fig:figureAP2-3} shows the mean value $E\{H_k\}$ and the standard deviation $\sigma_{H_k}$ that are estimated by No-Group PLoM and by With-Group PLoM. We can see that the mean values remain much lower than $1$ although no constraint is applied to the mean, as well for No-Group PLoM as for With-Group PLoM.
We can also see that the standard deviation of the components are already close to $1$ for No-Group PLoM although no constraint is applied to the second-order moments. As expected, for With-Group PLoM for which the constraints  are applied to the second-order moments, the standard deviations are almost equal to $1$.
\begin{figure}[tbhp]
  \centering
  \subfloat[Mean value of $H_k$] {\label{fig:figureAP2-3a} \includegraphics[width=5.4cm]{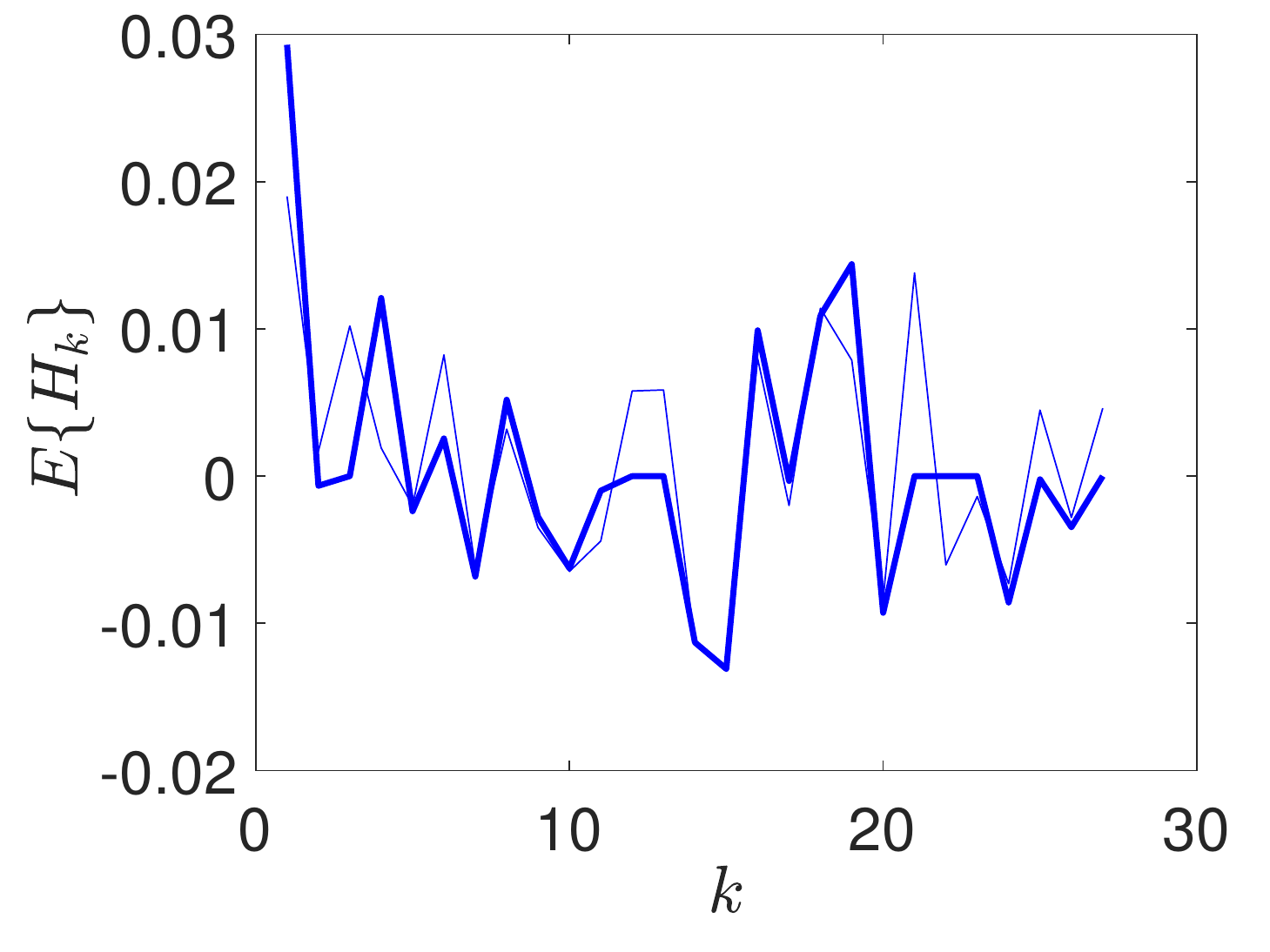}}
  \hfil
  \subfloat[Standard deviation of $H_k$]  {\label{fig:figureAP2-3b} \includegraphics[width=5.5cm] {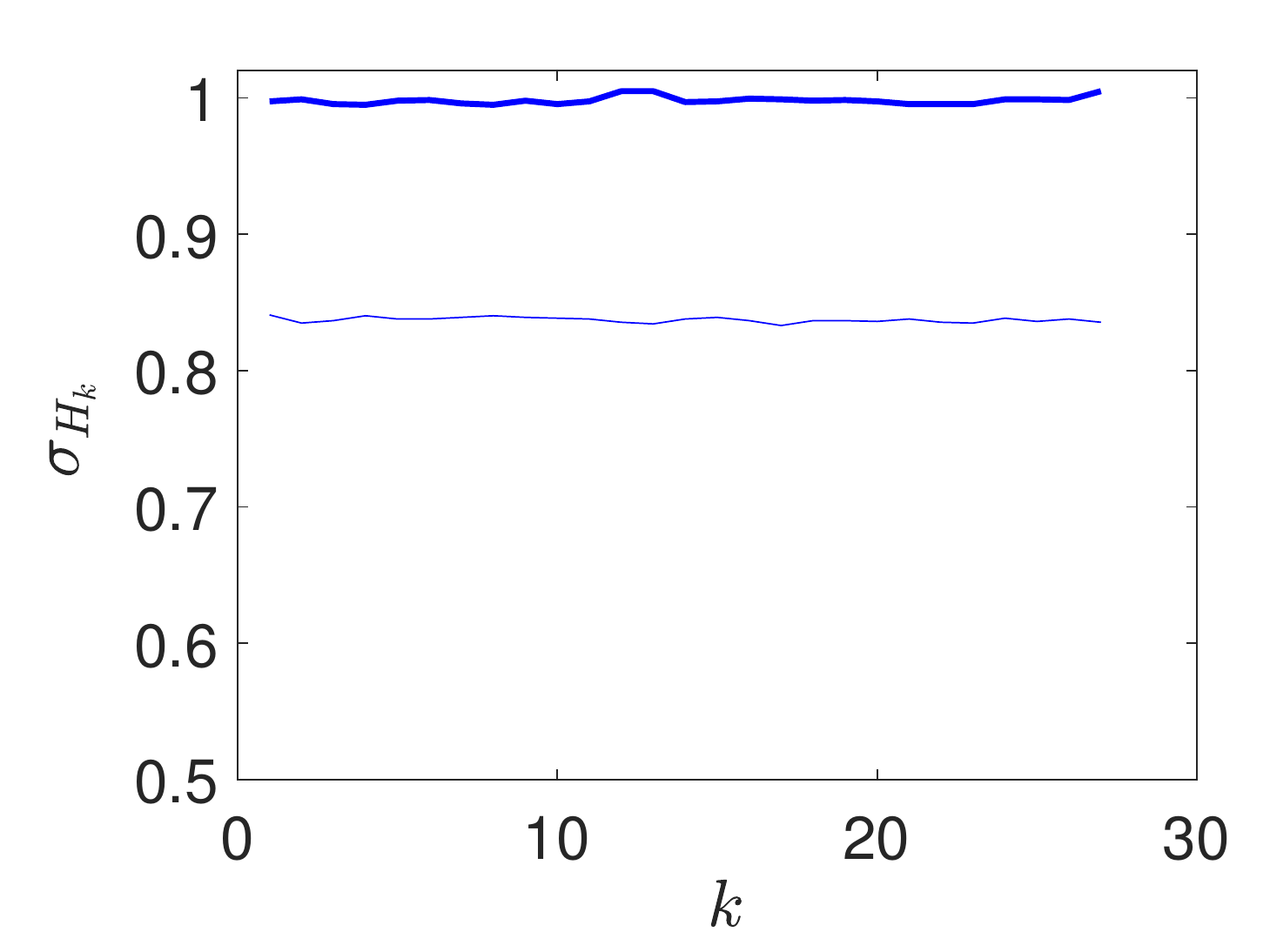}}
  \caption{Mean value and standard deviation of the components $H_k$, $k=1,\ldots, \nu$, of $\bfH$ estimated using the learned set generated
  with No-Group PLoM (thin line) and  With-Group PLoM (thick line).}
  \label{fig:figureAP2-3}
\end{figure}
\paragraph{(v) pdf of observations estimated by the PLoM}
The pdf of components $Q_{17}$ and $Q_{7740}$ of $\bfQ$ are presented in Figure~\ref{fig:figureAP2-4}. Component $17$ corresponds to the $\zeta_2$-axis displacement $V_2(\bfzeta)$ at point $\bfzeta=(0,0,0.1)$ while component $7740$ corresponds to the $\zeta_3$-axis displacement $V_3(\bfzeta)$ at point $\bfzeta=(0.78,0,0.1)$.
\begin{figure}[tbhp]
  \subfloat[pdf of $Q_{17}$]{\label{fig:figureAP2-4a} \includegraphics[width=6.6cm] {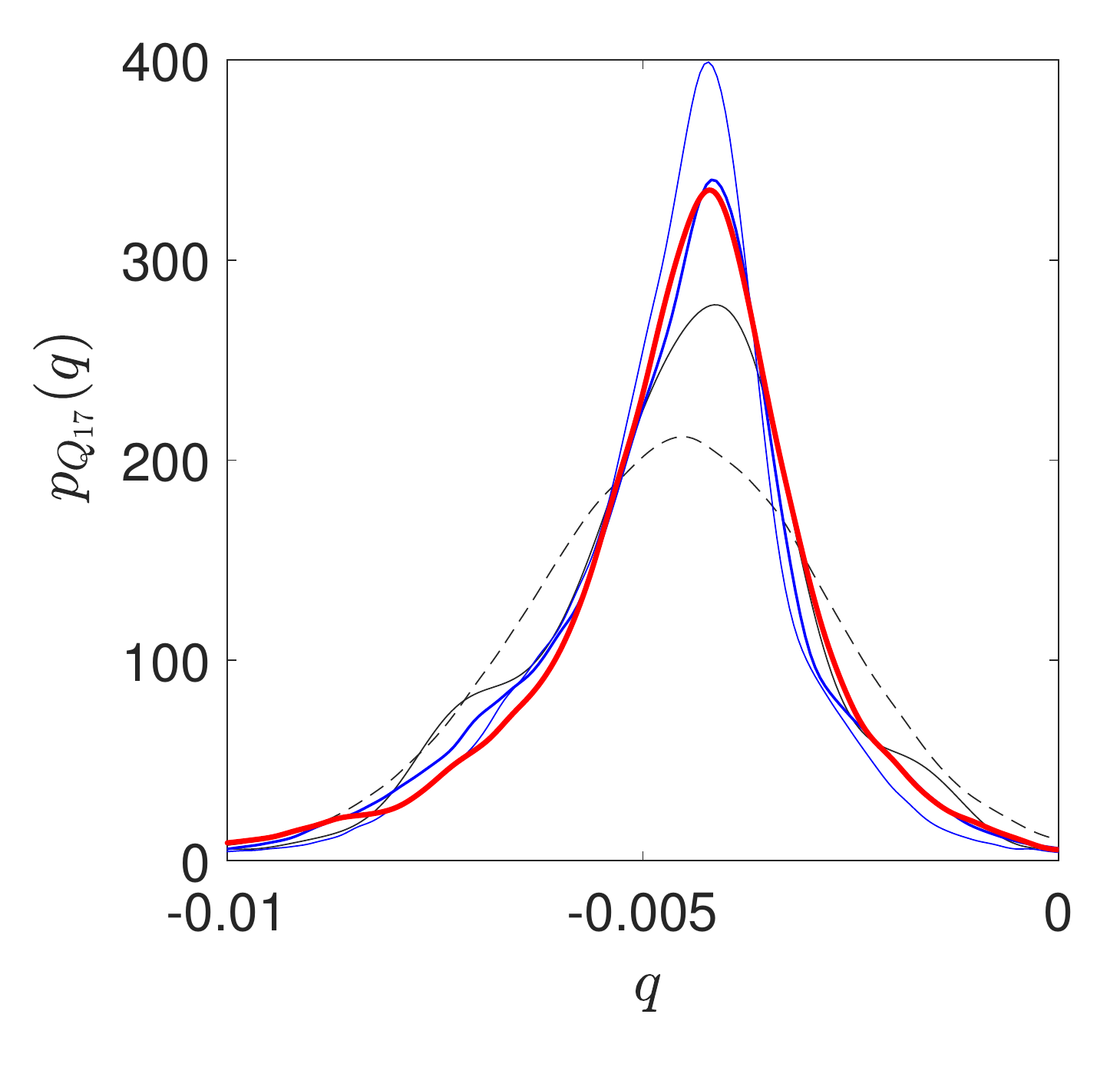}} \hfil
  \subfloat[pdf of $Q_{7740}$]{\label{fig:figureAP2-4b} \includegraphics[width=6.6cm] {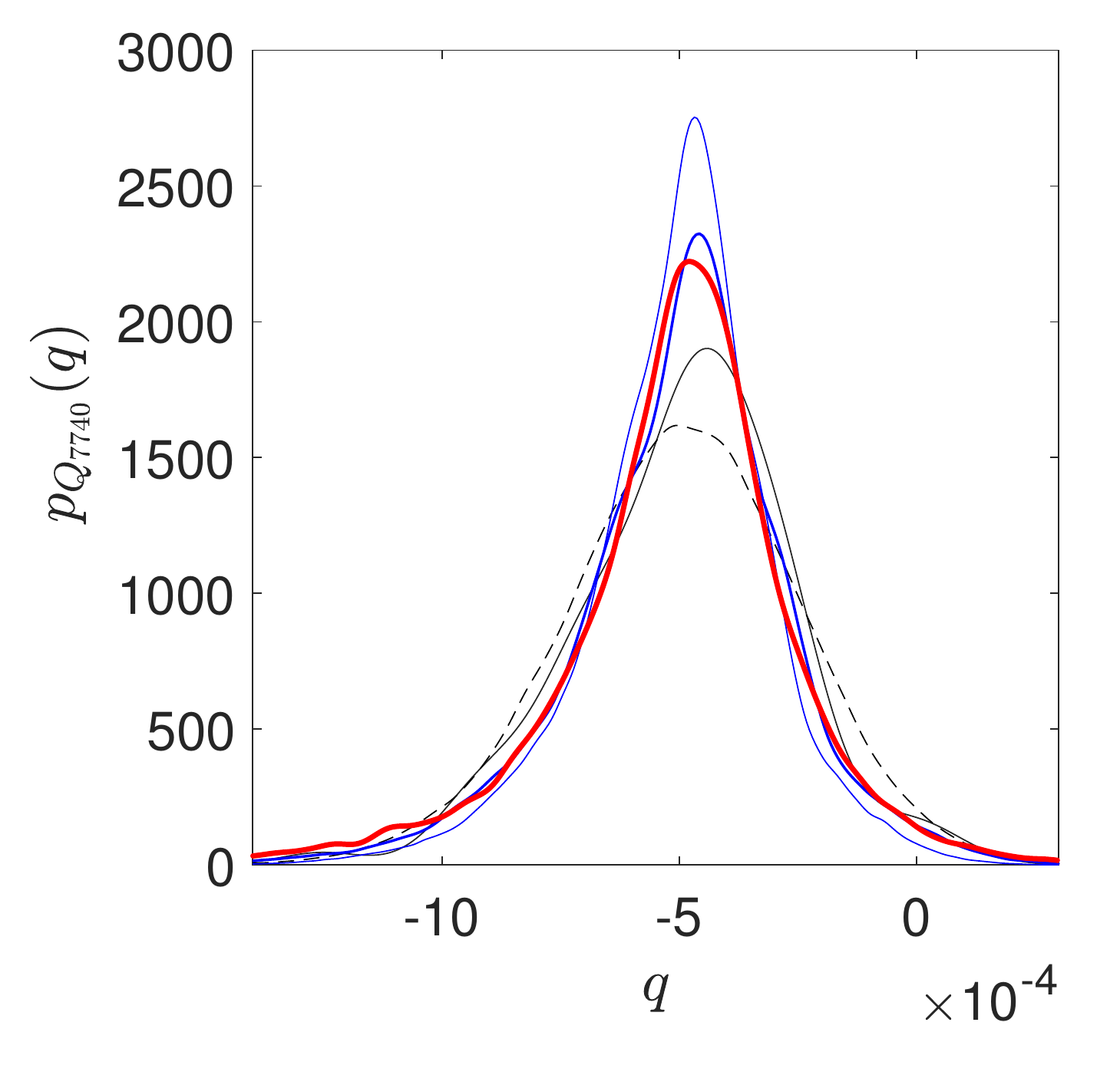}}
  \caption{pdf estimated with (i) the training dataset (black thin), (ii) the reference dataset (red thick), (iii) No-PLoM (dashed), (iv) No-Group PLoM (blue thin), and (v) With-Group PLoM (blue thick).}
  \label{fig:figureAP2-4}
\end{figure}
Figure~\ref{fig:figureAP2-4} shows the pdf estimated (i) with the $N_d=100$ points of the training set, (ii) with the $N_\pref = 20\, 000$ points of the reference data set, (iii) with  $N_\ar = 200\, 000$ additional realizations generated with an usual MCMC generator (without using the PLoM method), (iv) with the $N_\ar = 200\, 000$ additional realizations of the learned set generated by No-Group PLoM, and finally,  (v) with the $N_\ar = 200\, 000$ additional realizations of the learned set generated by With-Group PLoM for which a partition in $n_p=9$ groups has been identified.
It can be seen that the usual MCMC method (no PLoM) is not good at all, that  No-Group PLoM already gives a good estimation in comparison with the reference, and finally, that  With-Group PLoM gives an excellent estimation when compared to the reference.

\paragraph{(vi) Quantifying the concentration of the probability measure}
The analysis is carried out as in Section~\ref{sec:Ap1-4}.
The results concerning the  concentration of the probability measure is summarized in Table~\ref{tab:table2} and in Figure~\ref{fig:figureAP2-5}.
\begin{table}[tbhp]
  \caption{Concentration of the probability measure for Application~2}\label{tab:table2}
\begin{center}
 \begin{tabular}{|c|c|c|c|c|} \hline
     \textbf{No}  & \textbf{PLoM}    &      \multicolumn{3}{c|}{\textbf{PLoM}}                                              \\
     \textbf{PLoM}& \textbf{No Group}&      \multicolumn{3}{c|}{\textbf{With Group}}                                        \\
     \hline
                  &  $m_\opt =28$    &                         & \multicolumn{2}{c|}{$\rm{Proba}$ by Eq.~\eqref{eq:Meth-12}} \\
       $d^2_N(N)$ & $d^2_N(m_\opt)$  & $d^2_{wg,N}(\bfm_\opt)$ &  $\varepsilon$ &  $\leq (\frac{r}{\varepsilon})^{n_p}$     \\ \hline
       $2.00$     &   $0.16$         &     $0.044$             &    $0.05$           &   $\leq 4.3\times 10^{-5}$                        \\
                  &                  &                         &    $0.1$           &   $\leq 8.3\times 10^{-8}$            \\
     \hline
  \end{tabular}
\end{center}
\end{table}
For No PLoM,  $d^2_N(N)$ is computed with Eq.~\eqref{eq:Meth-4} for which $m=N=100$. For No-Group PLoM, $d^2_N(m_\opt)$ is also computed with Eq.~\eqref{eq:Meth-4} but with $m=m_\opt = 28$. For With-Group PLoM, $d^2_{wg,N}(\bfm_\opt)$ is computed with Eq.~\eqref{eq:Meth-7} for which $\bfm_\opt = (m_{1,\opt},\ldots ,m_{9,\opt})$.
The graph $i\mapsto d^2_{i,N}(m_{i,\ppopt})$ is computed using Eq.~\eqref{eq:Meth-8} and is plotted in Figure~\ref{fig:figureAP2-5}.
Without using the PLoM method, we find numerically $d^2_N(N) = 2$ that is the theoretical value (see Section~\ref{sec:Meth-3.5}). We also see that $d^2_N (m_\opt) = 0.16 \ll 2 $, which shows that the usual PLoM method (without group) effectively preserves the concentration of the probability measure unlike the usual MCMC methods that do not allow it. For the PLoM with groups, it can be seen an improvement with respect to the PLoM without group because $d^2_{wg,N}(\bfm_\opt)= 0.044 \ll  d^2_N(m_\opt) = 0.16$. The quantification of the probability of the random relative distance defined by Eq.~\eqref{eq:Meth-12} confirms this improvement. Note that the probability $(r/\varepsilon)^{n_p}$ corresponds to an upper value, the probability being certainly smaller.
\begin{figure}[tbhp]
  \centering
  \includegraphics[width=5.5cm]{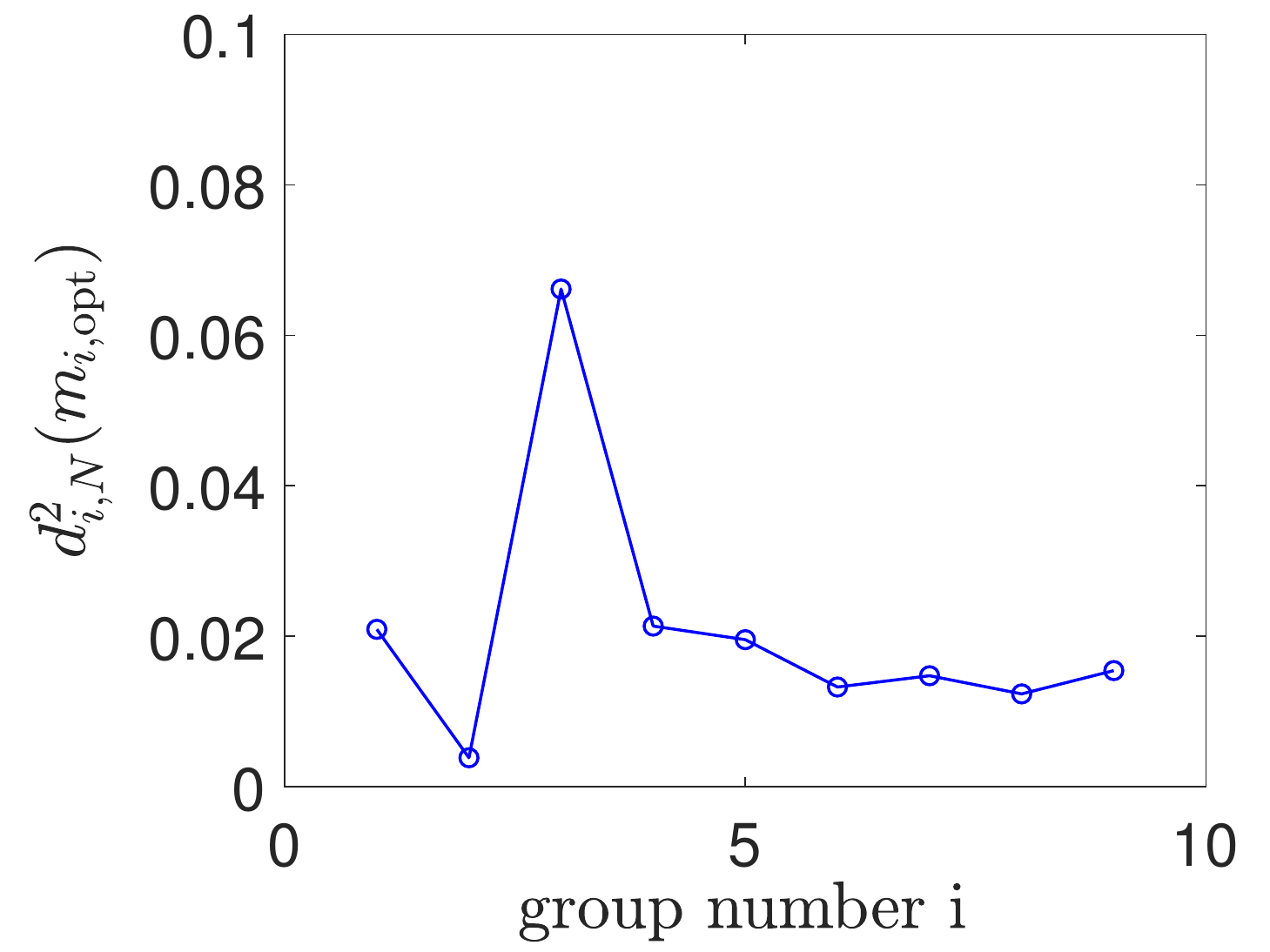}
\caption{Concentration of the probability measure for each group $i\in\{1,\ldots , 9\}$: graph of $i\mapsto d^2_{i,N}(m_{i,\ppopt})$  computed using Eq.~\eqref{eq:Meth-8}.}
\label{fig:figureAP2-5}
\end{figure}

\section{Discussion and conclusion}
%
The implementation of a partition in the PLoM method has provided an
opportunity to revisit, improve the efficiency, and simplify the
algorithm to identify the optimal values of the hyperparameters of the
reduced-order diffusion-map basis. This was made necessary for the
PLoM method with partition, because the number of groups identified
can be large and for each group of dimension greater than 1, the
reduced-order diffusion-map basis must be constructed. This new
efficient algorithm is common to PLoM with or without partition.

Still within the framework of the PLoM  carried out with partition,
we have made the following observations. If a group of the partition
has a relatively small dimension (a few units, or even one or two
dozen) and if the support of its probability measure  has a complex
geometry,  one could obtain a significant loss of normalization
compared to 1 (for instance 0.6 or 0.7 instead of 0.9 or 1). For
instance, such a situation can be encountered  by the presence of
numerous non-Gaussian stochastic germs that generate strong
statistical fluctuations (for example, up to ten times the standard
deviation for some components). For these cases, we have proposed to
introduce constraints on the second-order moments of the components of
such a group, by reusing the Kullback-Leibler minimum cross-entropy
principle that we have previously used for taking into account physics
constraints in the PLoM method. For instance, Application 1 is very
difficult due to the high degree of the polynomials  in the model;
although the realizations of the training set are centered and
have a covariance matrix equal to the identity matrix, the
fluctuations vary between $-10$ and +$10$ for some components (which
must be compared to a magnitude of $1$).
It should be noted that we have also developed, tested, and
implemented the general case of introducing constraints on the mean
vector (zero mean) and on the covariance matrix (identity matrix). One
then increases considerably the number of Lagrange multipliers to be
calculated by the iterative method, which induces a significant
numerical additional cost. We have not seen any significant
improvement compared to the only constraint related to the diagonal of
the covariance matrix (second-order moments equal to $1$ knowing that
the centering is reasonably well obtained without constraint on the
mean vector). Under these conditions we have only presented the
simplest case of constraints and we have demonstrated it on two applications.

In the recently published mathematical foundations of PLoM
\cite{Soize2020c}, to establish the main theorem, we introduced a
distance between the random matrix defined by PLoM and the
deterministic matrix that represent all the given points of the
training set. In the present paper, and in order to facilitate the
quantification of the preservation of the concentration of the
probability measure between the usual MCMC method (No PLoM), the PLoM
method without partition (No-Group PLoM), and the PLoM method with
partition (With-Group PLoM), we apply this distance to each group of
the partition. We have assessed it numerically for the two
applications. The results obtained confirm the theoretical results:
there is a significant loss of  concentration of the probability
measure for the usual MCMC method (No PLoM) while the PLoM method
without partition (No-Group PLoM) preserves well the concentration of
the probability measure. In addition, this distance shows that the PLoM method with partition further improves the preservation of the concentration of the probability measure compared to the PLoM method without partition, which was hoped for.
Finally, to complete the quantification of the concentration of the probability measure by the distance, we have also proven  a mathematical result of this quantification in terms of probability. This result shows that if the number of groups of the partition increases, then the gain of With-Group PLoM can be significantly improved compared to No-Group PLoM. We have numerically quantified these probabilities for the two applications.

This work contributes to the improvement of the PLoM method. The
results presented are very satisfactory for the two applications
which, while quite distinct, present significant challenges
to other statistical learning methods.

\appendix
\section{Probabilistic model of the random generator for Applications~1}
\label{AppendixA}
In this Appendix, any second-order random quantity $\bfS$ is defined on a probability space $(\Theta,\curT,\curP)$ and its mathematical expectation $E\{\bfS\}$ is estimated by $\underline{\bfs} =(1/N)\sum_{j=1}^N \bfs^j$ using $N$ independent realizations $\{\bfs^j = \bfS(\theta_j), j=1,\ldots , N\}$ of $\bfS$ with $\theta_j\in\Theta$.

The $\RR^\nu$-valued random variable $\bfH =(H_1,\ldots , H_\nu)$  is written as a partition of $n_p = 3$ independent random vectors $\bfY^1,\ldots ,\bfY^{n_p}$ such that, for $i=1,\ldots, n_p$, the normalized $\RR^{\nu_i}$-valued random variable $\bfY^i = (Y^i_1,\ldots , Y^i_{\nu_i})$ is non-Gaussian and such that the estimate
of its mean vector is $\underline{\bfeta^i} = \bfzero_{\nu_i}$ and the estimate of its covariance matrix is $[C_{\bfY^i}] = [I_{\nu_i}]$.
We have $\nu =\nu_1 + \ldots + \nu_{n_p}$ and we choose $\nu_1=10$, $\nu_2= 20$, and $\nu_3=30$.

For $i=1,2,3$, let $[\, b^i]$ be the deterministic matrix in $\MM_{\nu_i}$ defined by: $rng('default')$; $[\,b^i] = (0.15\times rand(\nu_i,\nu_i) + 0.85)/ \nu_i$ (in which $rng$ and $rand$ are the Matlab functions). Let $\bfU^i = 2\,[\, b^i]\, \bfcurU^i - \11$ be the  $\RR^{\nu_i}$-valued random variable in which $\11=(1,\ldots , 1)$ and where $\bfcurU^i = (\curU^i_1,\ldots,\curU^i_{\nu_i})$ is the random vector constituted of $\nu_i$ independent uniform random variables on $[0,1]$, whose $N$ independent realizations are $\{\bfcurU^i(\theta_j), j=1,\ldots , N\}$. The random vectors
$\bfcurU^1$, $\bfcurU^2$, and $\bfcurU^3$ are statistically independent.
Let $\bfcurM^i = (\curM^i_1,\ldots , \curM^i_{\nu_i})$ be the $\RR^{\nu_i}$-valued random variable in which,
for $k=1,\ldots , \nu_i$, $\curM^i_k$ is the random monomial $\curM^i_k = \sqrt{k!} \,(U^i_k)^k$ (thus the degree of this monomial is $k$).  Let be $\bfcurM^i_c = \bfcurM^i - \underline{\bfm^i}$  in which $\underline{\bfm^i}$ is the estimate of the mean value of $\bfcurM^i$. Let $[C_{\bfcurM^i}]$ be the estimate of the covariance matrix of $\bfcurM^i$ and let $[L^i]$ be the upper triangular matrix computed from the Cholesky factorization $[C_{\bfcurM^i}] = [L^i]^T \, [L^i]$. Therefore, the random vector $\bfY^i$ is constructed as $\bfY^i = ([L^i]^T)^{-1}\, \bfcurM^i_c$ whose the $N$ independent realizations $\{\bfeta_d^{i,j},j=1,\ldots,N\}$ are such that
$\bfeta_d^{i,j} = ([L^i]^T)^{-1}\, \bfcurM^i_c(\theta_j)$. The $N$ independent realizations $\{\bfeta_d^j,j=1,\ldots,N\}$ of $\bfH$ are such that
$\bfeta_d^j = (\bfeta_d^{1,j},\bfeta_d^{2,j},\bfeta_d^{3,j}) \in \RR^\nu = \RR^{\nu_1}\times\RR^{\nu_2}\times\RR^{\nu_3}$. Using these $N$ realizations, the estimate
of the mean vector of $\bfH$ is $\underline{\bfeta} = \bfzero_\nu$ and the estimate of its covariance matrix is $[C_{\bfH}] = [I_{\nu}]$. By construction, we have
$\bfY^1= (H_1,\ldots, H_{10})$, $\bfY^2= (H_{11},\ldots, H_{30})$, and $\bfY^3= (H_{31},\ldots, H_{60})$.

\section{Model and data for Applications~2}
\label{AppendixB}

\paragraph{(i) Geometry and surface force field $\bfcurG^\Gamma$}
The  3D bounded domain is defined by $\Omega = ]0, 1.0[ \, \times \, ]0, 0.2[ \, \times \,]0 , 0.1[$.
Its boundary is written as $\partial\Omega=\Gamma_0 \cup\Gamma$ in which $\Gamma = \partial \Omega \backslash \Gamma_0$ with
$\Gamma_0 =\{\zeta_1=1.0\, , \, 0\leq \zeta_2 \leq 0.2 \, , \, 0 \leq \zeta_3 \leq 0.1 \}$.
The surface force field $\bfcurG^\Gamma = (\curG^\Gamma_1, \curG^\Gamma_2, \curG^\Gamma_3)$ is zero on all $\Gamma$ except on the part $\{\zeta_1=0.0\, , \, 0\leq \zeta_2 \leq 0.2 \, , \, 0 \leq \zeta_3 \leq 0.1 \}$ for which $\curG^\Gamma_1(\zeta) = - 4.95\times 10^7\, N/m^2$,
$\curG^\Gamma_2(\zeta) = - 4.29\times 10^5\, N/m^2$, and $\curG^\Gamma_3(\zeta) = - 1.65\times 10^5\, N/m^2$.

\paragraph{(ii) Probabilistic modeling of the elasticity random field}
Random field $\KK$ is rewritten as $\KK_{ijkh}(\bfzeta) = [\bfK(\bfzeta)]_{IJ}$ with $I=(i,j)$ and $J=(k,h)$,
in which indices  $I$ and $J$ belong to  $\{1,\ldots , 6\}$, and where $\{[\bfK(\bfzeta)] ,\bfzeta \in \RR^3\}$ is a second-order $\MM_6^+$-valued non-Gaussian random field  indexed by $\RR^3$, which is assumed to be  statistically homogeneous.
Its mean function $[\underline K] \in\MM_6^+$ is thus independent of $\bfzeta$ and corresponds to the elasticity tensor of a homogeneous isotropic elastic medium whose Young modulus is $10^{10}\, N/m^2$ and Poisson coefficient $0.15$. The statistical fluctuations of random field $[\bfK]$ around $[\underline K]$ are those of a heterogeneous  anisotropic elastic medium.
The non-Gaussian $\MM_6^+$-valued  random field  $\{ [\bfK(\bfzeta)]\, ,\bfzeta\in \Omega \}$ is constructed using the stochastic model  \cite{Soize2006,Soize2017b} of random elasticity fields for heterogeneous anisotropic elastic media that are isotropic in statistical mean and exhibit anisotropic statistical fluctuations. Its parameterization consists of three spatial-correlation lengths, one dispersion parameter, and  a positive-definite lower bound. Random field $[\bfK]$ is written, for all $\bfzeta$ in $\RR^3$, as
$[\bfK(\bfzeta)] = [K_\ell] + [\underline L_\varepsilon]^T\, [{\bfG_0}(\bfzeta)]\, [\underline L_\varepsilon]$.
The lower-bound matrix is defined by $[K_\ell] = \varepsilon (1+ \varepsilon)^{-1}\, [\underline K] \in \MM_6^+$ in which $\varepsilon$ is chosen equal to $10^{-6}$. The upper triangular matrix $[\underline L_\varepsilon] \in\MM_6$ is written as
$[\underline L_\varepsilon] =  (1+ \varepsilon)^{-1/2}\, [\underline L]$ in which $[\underline K] = [\underline L]^T\, [\underline L]$ (Cholesky factorization).
The non-Gaussian random field $\{[{\bfG_0}(\bfzeta)],\bfzeta\in\RR^3\}$, which  is indexed by $\RR^3$ with values
in $\MM_6^+$, is homogeneous in $\RR^3$ and is a second-order random field whose modeling and generator are detailed Pages 272 to 274 of \cite{Soize2017b}.
For all $\bfzeta$ in $\RR^3$, the random matrix $[{\bfG_0}(\bfzeta)]$ is written as
$[\bfG_0(\bfzeta)] = [\bfL_G(\bfzeta)]^T \, [\bfL_G(\bfzeta)]$
in which $[\bfL_G(\bfzeta)]$ is an upper $(6\times 6)$ real triangular random matrix that depends of $n_G =21$ independent normalized Gaussian random variables.
Random field $[\bfG_0]$ depends on $3$ spatial correlation lengths, $L_{\pcorr,1}$, $L_{\pcorr,2}$, $L_{\pcorr,3}$, relative to each one of the three directions $\zeta_1$-, $\zeta_2$-, and $\zeta_3$-axes. It also depends on the dispersion parameter $\delta_G > 0$ that allows for controlling the level of statistical fluctuations. As explained in Section~\ref{sec:Ap2-1}, only two hyperparameters are kept: $L_\pcorr$ and $\delta_G$, for which we have chosen
$L_{\pcorr,1} = L_{\pcorr,2} = L_{\pcorr,3} = L_\pcorr$.

\paragraph{(iii) Finite element approximation of the stochastic boundary value problem and definition of random vector $\bfU$}
Domain $\Omega$ is meshed with $50\!\times\! 10\!\times\! 5 = 2\,500$ finite elements
using $8$-nodes finite elements. There are $3\,366$ nodes and $10\,098$ dofs (degrees of freedom) before applying the Dirichlet conditions.
The displacements are locked  at all the $66$ nodes belonging to surface $\Gamma_0$ and therefore, there are $198$ zero Dirichlet conditions.
There are $8$ integration points in each finite element. Consequently, there are  $N_i= 20\, 000$ integration points $\bfzeta^1,\ldots, \bfzeta^{N_i}$.
Let us consider the $\RR^{n_u}$-valued random variable $\bfU$ constituted of all the $n_u= N_i \!\times\! n_G = 20\, 000\!\times\! 21= 420\, 000$ independent normalized Gaussian random variables that allow the set $\{[\bfL_G(\bfzeta^1)] ,\ldots , [\bfL_G(\bfzeta^{N_i})]\}$ of random matrices to be generated.

\paragraph{(iv) Construction of random vectors $\bfQ$, and $\bfW$}
The $\RR^{n_q}$-valued random variable $\bfQ$ of the QoIs are constituted of the $10\, 098$ dofs of the discretization of the random displacement field $\bfV$.
The random vector $\bfW=(W_1,W_2)$ in such that $W_1= \log (L_\pcorr)$ and $W_2 = log(\delta_G)$. The random variables $L_\pcorr$ and $\delta_G$ are independent and uniform on $[0.1\, , 1.0]$ and $[0.1\, 0.5]$, respectively. We then have $L_\pcorr = 0.9 \, \curU_1 + 0.1$ and $\delta_G = 0.4\,  \curU_2 +0.1$ in which $\curU_1$ and $\curU_2$ are two independent uniform random variable on $[0\, ,1]$.

\section*{Acknowledgments}
Support for this work was partially provided through the Scientific Discovery through Advanced Computing (SciDAC) program funded by the U.S. Department of Energy, Office of Science, Advanced Scientific Computing Research

\bibliographystyle{elsarticle-num}
\bibliography{references}
\end{document}